\documentclass[11pt]{article} 

\usepackage[usenames,dvipsnames]{xcolor}
\definecolor{Gred}{RGB}{219, 50, 54}
\definecolor{ToCgreen}{RGB}{0, 128, 0}

\usepackage[letterpaper,margin=1in]{geometry}
\usepackage{titling}

\usepackage{microtype}
\usepackage{wrapfig}
\usepackage{amsmath}

\usepackage{framed}
\usepackage[noend]{algpseudocode}
\usepackage[lined,boxed,ruled,norelsize,algo2e,linesnumbered]{algorithm2e}

\usepackage{amsthm,amsmath,mathtools}
\usepackage{titlesec}
\titleformat*{\paragraph}{\bfseries}

\usepackage[colorlinks]{hyperref}
\hypersetup{
      colorlinks=true,
  citecolor=ToCgreen,
  linkcolor=Sepia,
  filecolor=Gred,
  urlcolor=Gred
  }
  
\usepackage[nameinlink]{cleveref}

\usepackage{empheq}
\usepackage{dsfont}
\usepackage{mathrsfs}
\usepackage{graphicx}
\usepackage{enumitem}
\usepackage{aliascnt} 
\usepackage{xspace}
\usepackage{booktabs}

\newcommand{\iprod}[1]{\langle #1\rangle}

\newcommand{\flatmat}{\vec{\Psi}}

\newcommand{\calK}{\mathcal{K}}

\newcommand{\bA}{\vec{A}}
\newcommand{\bB}{\vec{B}}
\newcommand{\bQ}{\vec{Q}}
\newcommand{\nc}{m} 

\newcommand{\triangleq}{\coloneqq}

\newcommand{\paramerr}{\upsilon}

\usepackage[font=small,labelfont=bf]{caption}

\newtheorem{theorem}{Theorem}[section]

\newtheorem{lemma}[theorem]{Lemma}
\newtheorem{informal theorem}[theorem]{Theorem (informal statement)}

\newtheorem{proposition}[theorem]{Proposition}
\newtheorem{corollary}[theorem]{Corollary}
\newtheorem{claim}[theorem]{Claim}
\newtheorem{fact}[theorem]{Fact}

\newtheorem{remark}[theorem]{Remark}

\newtheorem{definition}[theorem]{Definition}

\renewcommand{\vec}[1]{\boldsymbol{\mathrm{#1}}}

\newcommand\norm[1]{\| #1 \|}

\DeclareMathOperator*{\pr}{\mathbb{P}}
\DeclareMathOperator*{\E}{\mathbb{E}}

\newcommand{\normal}{\mathcal{N}}

\DeclareMathOperator*{\argmin}{argmin}

\newcommand{\mc}[1]{\mathcal{#1}}

\newcommand{\tr}{\mathrm{tr}}
\newcommand{\Tr}{\mathrm{tr}}

\newcommand{\op}{\mathsf{op}}
\newcommand{\coeffvec}{\vec \theta}

\newcommand{\vmu}{\vec{\mu}}
\newcommand{\vg}{\vec{g}}

\newcommand{\R}{\mathbb{R}}

\newcommand{\eps}{\varepsilon}

\newcommand{\poly}{\mathrm{poly}}
\newcommand{\polylog}{\mathrm{polylog}}

\renewcommand{\top}{\intercal}

\newcommand{\calE}{{\cal E}}
\newcommand{\calN}{{\cal N}}

\newcommand{\calS}{{\cal S}}
\newcommand{\calT}{{\cal T}}
\newcommand{\calU}{{\cal U}}

\newcommand{\Ind}{\mathds{1}}
\newcommand{\1}{\Ind}

\newcommand{\wt}{\widetilde}
\newcommand{\wh}{\widehat}

\newcommand{\bE}{\vec{E}}

\newcommand{\lambdamin}{\lambda_{\rm min}}

\newcommand{\x}{\vec x}
\newcommand{\z}{\vec z}

\newcommand{\citet}{\cite}
\newcommand{\citep}{\cite}

\newcommand{\lossdenoisingclipped}{L^{({\rm clip})}}
\newcommand{\lossdenoisingout}{L^{(out)}}

\newcommand{\vc}[1]{\mathrm{vec}(#1)}
\newcommand{\Id}{\boldsymbol{\mathrm{Id}}}
\newcommand{\bM}{\vec{M}}

\newcommand{\cov}{\vec{Q}}
\newcommand{\quadform}{\vec{\Lambda}}

\newcommand{\withincluster}{\Delta_{\rm in}}
\newcommand{\betweencluster}{\Delta_{\rm out}}
\newcommand{\withinclustermean}{\Delta^{(\vmu)}_{\rm in}}
\newcommand{\betweenclustermean}{\Delta^{(\vmu)}_{\rm out}}
\newcommand{\withinclustercov}{\Delta^{(\cov)}_{\rm in}}
\newcommand{\betweenclustercov}{\Delta^{(\cov)}_{\rm out}}

\newcommand{\numclust}{{n_{\rm c}}}

\newcommand{\radius}{R}

\newcommand{\ballt}{B}
\newcommand{\ball}[1]{B^{(#1)}}
\newcommand{\ballc}[1]{V^{(#1)}}

\newcommand{\hballt}{\wh{\ballt}}
\newcommand{\hballc}[1]{\wh{V}^{(#1)}}
\newcommand{\hball}[1]{\wh{B}^{(#1)}}

\newcommand{\htheta}{\widehat{\theta}}

\newcommand{\thetadiff}{\omega}

\newcommand{\lscore}{\wh{\vec s}}           
\newcommand{\tvdistance}{\mathrm{TV}}
\newcommand{\kldistance}{\mathrm{KL}}
\newcommand{\noisymixture}[1]{\mathcal M_{#1}}

\newcommand{\meanerr}{\upsilon_{\rm mean}}
\newcommand{\hatPi}{\wh{\vec \Pi}}
\newcommand{\resid}{\vec{\zeta}}
\newcommand{\bC}{\vec{C}}
\newcommand{\calW}{\mathcal{W}}
\newcommand{\hatmu}{\wh{\vmu}}
\newcommand{\covthres}{\Delta}
\newcommand{\Sfar}{S_{\rm far}}
\newcommand{\Sclose}{S_{\rm close}}

\newcommand{\w}{\vec{w}}
\newcommand{\hatQ}{\wh{\bQ}}
\newcommand{\hbeta}{\wh{\beta}}

\newcommand{\classify}{\mathsf{c}}

\newcommand{\covsep}{\Delta_{\rm covsep}}
\newcommand{\epscore}{\epsilon_{\mathrm{score} }}
\newcommand{\coverr}{\upsilon_{\rm cov}}
\newcommand{\precest}{\wh{\vec{K}}}

\newcommand{\bU}{\vec{U}}
\newcommand{\bLam}{\vec{\Lambda}}

\newcommand{\dpar}{D_p}

\newcommand{\score}{\vec s}

\newcommand{\epspart}{\eps_{\mathrm{part}}}

\title{Learning general Gaussian mixtures with efficient score matching}

\author{
Sitan Chen \thanks{Email: \texttt{sitan@seas.harvard.edu}.}\\
Harvard SEAS
\and
Vasilis Kontonis \thanks{Email: \texttt{vasilis@cs.utexas.edu}, supported by the NSF AI Institute for Foundations of Machine Learning (IFML).}\\
UT Austin 
\and
Kulin Shah \thanks{Email: \texttt{kulinshah@utexas.edu}, supported by the NSF AI Institute for Foundations of Machine Learning (IFML).
} \\
UT Austin
}

\date{April 29, 2024}

 \usepackage[T1]{fontenc}
\usepackage[bitstream-charter,cal=cmcal]{mathdesign}

\begin{document}

\setcounter{page}{0}
\thispagestyle{empty}

\maketitle

\begin{abstract}
We study the problem of learning mixtures of $k$ Gaussians in $d$ dimensions. 
We make no separation assumptions on the underlying mixture components: 
we only require that the covariance matrices have bounded condition number 
and that the means and covariances lie in a ball of bounded radius. We give an algorithm that draws $d^{\poly(k/\eps)}$ samples from the target mixture, 
runs in sample-polynomial time, and constructs a sampler whose output distribution is $\eps$-close from the unknown mixture in total variation.  Prior works for this problem either (i) required exponential runtime in the dimension $d$, (ii) placed strong assumptions on the instance (e.g., spherical covariances or clusterability), or (iii) had doubly exponential dependence on the number of components $k$.

Our approach departs from commonly used techniques for this problem like the method of moments. Instead, we leverage a recently developed reduction, based on diffusion models, from distribution learning to a supervised learning task called score matching. We give an algorithm for the latter by proving a structural result showing that the score function of a Gaussian mixture can be approximated by a piecewise-polynomial function, and there is an efficient algorithm for finding it. To our knowledge, this is the first example of diffusion models achieving a state-of-the-art theoretical guarantee for an unsupervised learning task.

\end{abstract}

\thispagestyle{empty}

\clearpage

\thispagestyle{empty}

\tableofcontents

\thispagestyle{empty}

\clearpage

\setcounter{page}{1}

\section{Introduction}

Gaussian mixture models (GMMs) are
one of the most well-studied models in statistics, 
with a history going back to the work of Pearson~\cite{pearson1894contributions}. 
Its computational study was initiated in the work of Dasgupta~\cite{Dasgupta:99}; since then,  it has been one of the prototypical non-convex learning problems that 
 has attracted significant attention from the theoretical computer science community
\cite{VempalaWang:02,KSV:05,BV:08,kalai2010efficiently,moitra2010settling, belkin2015polynomial,hopkins2018mixture,kothari2018robust,diakonikolas2020robustly,bakshi2020outlier,diakonikolas2020small,liu2022clustering,liu2023robustly,bakshi2022robustly,buhai2023beyond}.
 
\paragraph{Learning without separation}
We focus on learning even when parameter recovery is impossible, i.e., without assuming that the 
components of the mixture are separated.  In this setting, the learner has to produce a hypothesis that is close to the target GMM in total variation distance \cite{FOS:05focs,moitra2010settling, CDSS13, SOAJ14, DK14, DKKLMS16,  acharya2017sample, li2017robust, ashtiani2018nearly, diakonikolas2020small,bakshi2022robustly,buhai2023beyond}.

Statistically, this problem is essentially completely understood: in order to approximate the target mixture of  $k$ Gaussians in $\eps$ total variation distance, it is known that $\wt{\Theta}(k d^2/\eps^2)$ samples are sufficient and also necessary  \cite{ashtiani2018nearly}.  Even though statistically almost optimal, the algorithm of \cite{ashtiani2018nearly} has a runtime scaling exponentially in $\tilde{O}(kd^2)$. This exponential dependence on the dimension is due to the fact that their algorithm is based on brute-force enumeration.

Despite significant efforts, the computational aspects of the problem are still far from
well-understood.  The work \cite{SOAJ14} provided an algorithm for learning 
mixtures of spherical (i.e., with covariance matrices that are multiples of the identity $\Id$)
with  $\poly(d k /\eps)$ sample complexity and $\poly(d) (k/\eps)^{\poly(k)}$ runtime.
For spherical Gaussians, the runtime was more recently improved to quasi-polynomial in $k$:
in \cite{diakonikolas2020small}, a runtime and sample complexity of $\poly(d)(k/\eps)^{\log^2 k}$ was given.  

For GMMs with general covariance matrices, the focus of the present work, the best-known runtime is due to 
\cite{bakshi2022robustly} and is doubly exponential in the number of components $k$, i.e.,
$(d/\eps)^k {{(1/\eps)}^k}^{k^2}$.  To the best of our knowledge, this doubly 
exponential dependency on $k$ is implicit in all works on learning general GMMs using the method of moments \cite{moitra2010settling,bakshi2020outlier,diakonikolas2020robustly,liu2023robustly} (see \Cref{sec:doubly-exponential} for intuition for where this comes from). 

In particular, for any $k = \Omega(\sqrt{\log d})$, previously there was no algorithm that ran in time faster than the exponential-time algorithm of~\cite{ashtiani2018nearly}, even for constant $\epsilon$!

On the negative side, there is strong evidence in the form of statistical query (SQ) \cite{diakonikolas2017statistical} and lattice-based \cite{bruna2021continuous,gupte2022continuous} hardness that runtime which scales super-polynomially in the number of components $k$ is necessary.  More precisely, the SQ lower bound of \cite{diakonikolas2017statistical} implies that even to learn within constant
accuracy $\eps>0$, $d^{\Omega(k)}$ runtime is required.  Our work aims to bridge the gaps 
between the best-known upper and lower bounds for learning GMMs \--- we ask the following fundamental question.

\begin{center}
\emph{ What is the best possible runtime for learning general Gaussian mixture models with $k$ components?
Can we improve over the doubly exponential runtime of moment-based methods? 
}
\end{center}

\noindent We make significant progress towards answering this question. 
Under mild ``condition number'' bounds on the mixture components \--- and without assuming the components are separated \--- we give an algorithm that achieves runtime $d^{\poly(k)}$ for any constant accuracy $\eps > 0$.  Thus, for well-conditioned mixtures, our result 
improves {\em exponentially} over the best-known runtime of \cite{bakshi2022robustly} in the regime where $k = \Omega(\sqrt{\log d})$.

\paragraph{Diffusion models and learning}

Interestingly, our algorithm \emph{does not rely on matching moments} with
the target mixture. Instead, we draw inspiration from the recent literature on proving theoretical guarantees for diffusion models~\cite{debetal2021scorebased,BloMroRak22genmodel,chen2022improved, DeB22diffusion, leelutan22sgmpoly, liu2022let, Pid22sgm, WibYan22sgm, chen2023sampling, chen2023restoration, lee2023convergence,li2023towards,benton2023error,chen2023probability,benton2023linear, conforti2023score, wibisono2024optimal}, the state-of-the-art method in practice for audio and image generation~\cite{sohl2015deep,dhariwal2021diffusion,song2020score,ho2020denoising}. These works culminated in the key finding that for any distribution with bounded second moment, there is a reduction from distribution learning to a supervised learning task called \emph{score matching}. Roughly speaking, this task is defined as follows: given a sample from the target distribution that has been corrupted by some Gaussian noise, predict the noise that was used to generate the sample (see \Cref{sec:diffusion_basics} for an exposition of these concepts). Despite the striking level of generality with which this reduction holds, these works fell short of giving ``end-to-end'' learning guarantees as they didn't address how to actually perform score matching algorithmically.

Our main technical contribution is an algorithm for score matching for GMMs. This relies on a novel structural result showing that the score function of a GMM can be well-approximated by a piecewise polynomial, together with an efficient procedure to recover the polynomial pieces.

While diffusion models have achieved remarkable empirical successes~\cite{betker2023improving}, to our knowledge our guarantee marks the first example of an unsupervised learning problem where diffusion models can even yield improved \emph{theoretical} guarantees. Our techniques are a synthesis of this modern algorithmic technique on the one hand and classic ideas from theoretical computer science like low-degree approximation on the other. We leave it as an intriguing open question to identify other problems for which this marriage of toolkits could prove useful.

\subsection{Our results and techniques}
We first give the formal definition of the well-conditioned GMMs that we consider in this work.
Roughly, we require that the covariance matrices of the components are well-conditioned in the sense that their eigenvalues are upper and lower bounded and that the means and covariances lie within an $\ell_2$ ball 
of bounded radius.

\begin{definition}[Well-Conditioned Gaussian Mixture]
\label{def:well-conditioned-mixture}
Let $\normal_1,\ldots, \normal_k$ be $d$-dimensional Gaussian distributions 
with means $\vec \mu_1,\ldots, \vec \mu_k$ and covariances $\cov_1,\ldots, \cov_k$.  We denote by $\mathcal M$ the mixture of these distributions with weights $\lambda_1,\ldots, \lambda_k$.  We will say that $\mathcal M$ is $\tau$-well-conditioned if
for some $\alpha \leq 1 \leq \beta$ and $R>0$ with $(\beta/\alpha) \log R \leq \tau$, it holds that: for all $i$, $\alpha ~ \Id \preceq \cov_i 
\preceq \beta ~ \Id$ and $\|\vec \mu_i \|_2 + \|\cov_i - \Id\|_F \leq R$.
When we want to distinguish between parameters we will also say that $\mathcal M$ is 
$(\alpha, \beta, R)$-well-conditioned.
Moreover, we denote by $\lambdamin $ the minimum weight $\min_{i \in [k]} \lambda_i$.
\end{definition}

\noindent We now present our main result: an efficient algorithm for learning well-conditioned GMMs.

\begin{theorem}[Informal \--- Learning Gaussian mixtures, see \Cref{thm:generate-sample}]\label{thm:main_informal}  
Let $\mathcal M$ be a $\tau$-well-conditioned mixture of $k$ Gaussians in $d$
dimensions, and suppose $\lambdamin \ge 1/\poly(k)$.
There exists an algorithm that draws $N = d^{\poly(k\tau/\eps)}$ samples from $\mathcal M$, runs in sample-polynomial time, and constructs a sampling oracle
whose output distribution is $\eps$-close to $\mathcal M$ in total variation.  To generate a new sample the oracle requires $\poly(N, d)$ time.
\end{theorem}

\noindent 
To our knowledge, this is the first example of an unsupervised learning problem for which a diffusion-based sampler outperforms existing state-of-the-art theoretical approaches~\cite{moitra2010settling,bakshi2022robustly}.   In particular, when the number of components $k$ is super-constant, i.e., $k = \Omega (\sqrt{\log d})$, we obtain a quasipolynomial $2^{\poly(\log d)}$ runtime, improving over the exponential $2^{\poly(d)}$ runtime following from \cite{bakshi2022robustly}. Moreover, we remark that using moment methods for Gaussian mixtures, e.g.,~\cite{bakshi2022robustly}, results in a doubly exponential runtime in $k$ even for well-conditioned mixtures, see \Cref{sec:doubly-exponential}. 
Finally, our improvements hold for any $\epsilon = 1/\polylog(d)$.  In fact, prior to our work, nothing better than doubly 
exponential in $k$ was known even for constant accuracy $\epsilon = \Omega(1)$.  We leave investigating whether the dependency on $1/\epsilon$ can be improved as an interesting question for future work.

\paragraph{Learning mixtures of degenerate Gaussians.} 
As stated, \Cref{thm:main_informal} does not appear to give anything for mixtures with covariances that are not full rank. This includes, for instance, mixtures of linear regressions and mixtures of linear subspaces~\cite{chen2020learning,diakonikolas2020small}. It turns out that we can still give a learning guarantees in this case, though in \emph{Wasserstein distance} rather than total variation, see \Cref{remark:degenerate}.

\subsection{Related work}

\paragraph{Learning mixtures of Gaussians} 
A thorough literature review on learning Gaussian mixtures is well outside the scope of this work. In addition to the sampling of works~\cite{dasgupta1999learning,FOS:05focs,moitra2010settling, belkin2015polynomial, CDSS13, SOAJ14, DK14, DKKLMS16,  acharya2017sample, li2017robust, ashtiani2018nearly,diakonikolas2020small,bakshi2020outlier,diakonikolas2020robustly,bakshi2022robustly,buhai2023beyond} mentioned in the introduction which deal with parameter estimation or distribution learning, we also mention a related line of work on \emph{clustering} Gaussian mixtures. This is a setting where there is a large enough separation between components that one can reliably identify which component generated a given sample. Some representative works in this line include~\cite{vempala2004spectral,brubaker2008isotropic,regev2017learning,hopkins2018mixture,diakonikolas2018list,kothari2018robust,liu2022clustering}.

Similar in spirit to the present work is the interesting work of~\cite{yan2023learning} which also eschews the method of moments in favor of a variational method. Whereas we use diffusion models, they use a certain interacting particle system that approximates a Wasserstein gradient flow. They focus on the case of Gaussian mixtures with identity covariance components. While they prove that the gradient flow itself converges in an asymptotic sense and numerically demonstrate the effectiveness of their approach, they do not prove non-asymptotic, end-to-end learning guarantees like in the present work.

\paragraph{General theory for diffusion models} Several works have provided convergence guarantees for DDPMs and variants~\cite{debetal2021scorebased,BloMroRak22genmodel,chen2022improved, DeB22diffusion, leelutan22sgmpoly, liu2022let, Pid22sgm, WibYan22sgm, chen2023sampling, chen2023restoration, lee2023convergence,li2023towards,benton2023error,chen2023probability,benton2023linear}. These works assume the existence of an oracle for accurate score estimation and show that diffusion models can learn essentially any distribution over $\R^d$ (e.g. \cite{chen2023sampling,lee2023convergence} show this for arbitrary compactly supported distributions, and \cite{chen2022improved,benton2023linear} extended this to arbitrary distributions with finite second moment). Recently, \cite{koehler2023sampling} showed that Langevin diffusion with data-dependent initialization can also learn multimodal distributions like mixtures of Gaussians, provided one can perform score matching. In another sampling context, \cite{anari2023parallel,anari2024fast} gave fast parallel algorithms based on a similar diffusion-style sampler for various problems like Eulerian tours and asymmetric determinantal point processes.

\paragraph{End-to-end applications of diffusions} 
In this work, we use a diffusion process as a tool to obtain \emph{end-to-end efficient learning algorithms} and we are not making ``black-box'' assumptions about the computational or the statistical complexity of learning the score function.
The recent works \cite{shah2023learning, cui2023analysis} also consider learning Gaussian mixtures, specifically with well-separated identity covariance components, using diffusions and show in different settings that gradient descent can provably perform score matching. 
The results of \cite{shah2023learning, cui2023analysis} only apply to the special
case of learning spherical Gaussian mixtures --- a setting that is already known to admit
efficient learning algorithms. The focus of those works is mainly in understanding why gradient descent for score matching can achieve guarantees similar to the prior known results while our goal in this work is to provide new efficient algorithms for general mixtures that are not captured by prior works.

Several recent results use diffusion models to obtain new \emph{sampling} algorithms
with a focus on graphical models.  This is a different setting than the one considered in the present work: instead of being given samples from the target distribution, one is given a Hamiltonian describing some graphical model, or some combinatorial object such that one would like to sample certain structures defined on it.  
For example, \cite{el2022sampling,montanari2023posterior,alaoui2023sampling, montanari2023sampling, huang2024sampling} have used Eldan's stochastic localization~\cite{eldan2013thin,eldan2020taming} method to give sampling algorithms for certain distributions arising in statistical physics.
These works provide an algorithmic implementation for the drift in the diffusion process, which is defined by the score, using approximate message passing and natural gradient descent (see also \cite{celentano2022sudakov}). 

Finally, in a concurrent and independent work \cite{gatmiry2024learning} 
the authors give diffusion-based algorithms for the special case of learning spherical (identity covariance) Gaussian mixtures, qualitatively matching the best-known results by \cite{diakonikolas2020small}.  
Our focus here is different: we learn Gaussian
mixtures with \emph{general, well-conditioned} covariance matrices
and improve over the prior works \cite{moitra2010settling, bakshi2022robustly} yielding exponential savings in runtime 
when the number of components $k$ is not constant, i.e., $k = \Omega(\polylog d)$.

\paragraph{Statistical guarantees for score matching} 
Several recent works have investigated the \emph{statistical} complexity of score matching. \cite{koehler2022statistical} showed a connection between the statistical efficiency of score matching and functional inequalities satisfied by the data distribution. \cite{pabbaraju2024provable} studied score matching for learning log-polynomial distributions. Like in~\cite{koehler2022statistical}, they focus on the score function of the base distribution and not noisy versions thereof; as the authors note, in this case, score matching is computationally tractable as it is exactly an instance of polynomial regression, and their focus was on proving that the statistical efficiency of score matching here is comparable to that of maximum likelihood estimation.

Recently, \cite{wibisono2024optimal} established the optimal rate for score estimation of nonparametric distributions in high dimensions. \cite{chen2023score,oko2023diffusion} studied the sample complexity of score matching for nonparametric distributions specifically using a neural network. \cite{mei2023deep} bounded the sample complexity of learning certain graphical models using diffusion models by arguing that neural network layers can implement iterations of certain variational inference algorithms. We emphasize once more that these guarantees are all statistical in nature rather than algorithmic.

\section{Technical overview}

In this section, we provide an overview of our approach, sketches for the main
arguments, and pointers to the relevant sections for more details.

\subsection{Learning via DDPM}
\label{sec:diffusion_basics} 

Our algorithm is based on a denoising diffusion probabilistic model (DDPM) \cite{sohl2015deep, song2019generative, ho2020denoising}. Here we give a self-contained exposition of the basic tools from this literature (see \Cref{sec:diffusion_details} for details); readers who are familiar with diffusion models may safely skip to \Cref{inf-thm:learning-the-score} below.

The most common \cite{song2020score, montanari2023sampling} approach is to consider the Ornstein-Uhlenbeck process, which
given some distribution $q_0$ corresponds to the SDE $d \x_t = -\x_t d t + \sqrt{2} d \w_t$, with $\x_0 \sim q_0$.  The distribution $q_0$ here corresponds to the target distribution that we want to learn to generate samples from.  In what follows, we use $q_t$ to denote the law of the OU process at time $t$.  It holds that $q_t$ converges to the standard normal distribution and in particular
at time $t$ we have that 
\begin{equation}
\label{inf-eq:OU-equation}
  \x_t = e^{-t} \x_0 + ~ \sqrt{1- e^{-2t}} ~ \vec z_t,   ~~~\text{for}~~~ \vec x_0 \sim q_0, ~~ \vec z_t \sim \normal\,.
\end{equation}
Given some terminal timestep $T$ of the forward OU process with distribution $q_T$, the following reverse process perfectly transforms noisy distribution $q_T$ (which is close to standard Gaussian) to the data distribution $q_0$:
\[
    \mathrm{d} \x^\leftarrow_t = \{\x^\leftarrow_t + 2\nabla_{\x} \log q_{T-t}(\x^\leftarrow_t)\}\, \mathrm{d}t + \sqrt{2}\,\mathrm{d} \w_t \;\;\;\text{with} \;\;\; \x^\leftarrow_0 \sim q_T\,. \]
In this reverse process, the iterate $\x^\leftarrow_t$ is distributed according to $q_{T - t}$ for every $t\in[0, T]$, so that the final iterate $\x^\leftarrow_T$ is distributed according to the data distribution $q_0$.  To be able to generate samples using the reverse SDE we need
access to the \emph{score function} $\nabla_{\x} \log q_t(\x)$.
Given approximate oracle access to the score function of the target density $q_0$ (for us this is the mixture of Gaussians) at close enough noise levels, we can discretize
the reverse SDE that starts with a sample from the Gaussian noise and generates a sample whose distribution is close to the target density.
In particular, for timesteps $t_0,\ldots, t_N$, given estimates $\wh{\vec s}(\x, T-t_\ell)$ 
we will be using the following update rule to generate a sample (sometimes called the exponential integrator scheme as it replaces the time-dependent score term in the reverse SDE with the score
approximation at time-step $T-t_{\ell}$).  More precisely, at the $\ell$-th iteration,
we sample $\vec z_\ell \sim \normal(\vec 0, \Id)$ and update our guess as follows:
 \begin{equation}
 \label{eq:exponential-integrator}
  \vec y_{\ell+1} \gets ~ \rho_\ell ~ \vec y_{\ell} + 2 (\rho_\ell - 1)
        ~ \wh{\vec s}(\vec y_{\ell}, T-t_{\ell}) + \sqrt{\rho_{\ell}^2 - 1} ~ \vec z_\ell, 
\end{equation}
where $\rho_\ell$ is an appropriately chosen ``step-size'' parameter, see \Cref{alg:generate-sample} for more details.   Several recent works  (see, e.g., \cite{chen2023sampling,lee2023convergence,chen2022improved,benton2023linear}) have studied the convergence of the above (discretized) reverse SDE to the data distribution under black-box assumptions on the quality
    of the score estimates $\wh{\vec s}(\cdot, \cdot)$.  We will be using a recent result from \cite{benton2023linear} (see \Cref{lem:diffusion-convergence}) that places minimal assumptions on the data distribution and gives fast convergence rates.
    More precisely, for the case of well-conditioned Gaussian mixtures, it implies that if the score functions are approximated within $L_2$ error roughly 
    $\poly(\eps/\tau )$, then iterating \Cref{eq:exponential-integrator} will produce a sample within total variation distance $\eps$
    from the target Gaussian mixture after $\poly(d\tau /\eps)$ iterations.

\paragraph{Learning the score}
 We have now reduced the original sampling problem to roughly $N = \poly(d \tau/\eps ) $ regression problems to get the approximate 
 score functions at times $t_1, \ldots, t_N$.  
More precisely for every $t\in \{t_1,\ldots, t_N\}$ we would like to use 
some expressive enough class of functions $\mathcal G$ and 
solve the following minimization (score-matching) problem:
$ \min_{\vec g \in \mathcal G} \;\; \E_{\x_0, \vec z_t} [ \| \vec g(\x_t) - \nabla_{\x} \log q_t(\x_t) \|^2_2 ]$
where $\x_t$ is generated by adding the Gaussian noise $\vec z_t$ to
the sample $\x_0 \sim \mathcal M$, $\x_t = e^{-t} \x_0 + \sqrt{1- e^{-2t}} \vec z_t$.  Since we have sample access to the unknown mixture $\mathcal M$, we can
generate i.i.d.\ copies of $\x_t$ to solve the regression task.  However, the target score function at noise-level $t$ is not available (as it depends on the
density of the unknown mixture).  A standard workaround~\cite{hyvarinen2005estimation,vincent2011connection,ho2020denoising, song2020score} 
is the denoising 
approach where conditional on the observed $\x_t$ we try to predict the added
noise $\vec z_t$.  
It is a well-known consequence of Gaussian integration by parts (see e.g. Appendix A of~\cite{chen2023sampling} for a proof) that the following regression task is equivalent to the original score-matching
problem with the benefit that it does not require knowledge of the score function of
the distribution $q_t$ (that corresponds to the distribution of $\x_t$):
\begin{equation}
\label{eq:diffusion-loss-tech-overview}
    \min_{\vec g \in \mathcal G} \;\; L_t(\vec g) = \min_{\vec g\in\mathcal{G}} \;\; \E_{\x_0, \vec z_t} \Bigl[ \Big\| \vec g(\x_t) + \frac{\vec z_t}{ \sqrt{1 - \exp(-2t)} } \Big\|^2_2 \Bigr]
\end{equation}

Our main technical contribution is an efficient algorithm that uses the above denoising formulation of the score-matching problem and yields an approximation
to the score function $\lscore (\x_t)$.  
\begin{proposition}[Informal - Efficiently Learning the Score - \Cref{thm:learning-score-guarantee-fixed-t}]
\label{inf-thm:learning-the-score}
Let $\mathcal M$ be a $\tau$-well-conditioned mixture. 
Then, for any $\eps> 0$ and noise scale $t \geq \poly(\eps/\tau)$, there exists an algorithm that draws $d^{\poly(k \tau/\eps)}$ samples from $\mathcal M$, 
runs in sample-polynomial time, and returns a score function $\lscore(\cdot)$ such that with high probability it holds 
\(
    \E_{\x_t \sim \noisymixture{t}} \big[ \| \lscore(\x_t) - \nabla_{\x} \log q_t(\x) \|^2 \big] \leq \eps.
    \)
\end{proposition}
    
\noindent A detailed theorem statement and the details of the algorithm can be found in \Cref{thm:learning-score-guarantee-fixed-t}. The details of the proof of \Cref{inf-thm:learning-the-score} can be found
in \Cref{sec:learning-piecewise-polynomial}.
Combining the above efficient algorithm with the convergence rate of the reverse
SDE we are able to get our end-to-end efficient algorithm for sampling from the
mixture $\mathcal M$.  Our efficient algorithm in \Cref{inf-thm:learning-the-score} relies on a structural result showing that the score function of the 
mixture $\mathcal M$ can be approximated by a piecewise-polynomial function,
and an efficient algorithm to recover the partition of the piecewise polynomial approximation.  In the following sections, we describe the main ideas of each part.

\begin{remark}[Learning mixtures of low-dimensional (degenerate) Gaussians]\label{remark:degenerate}
    Here we briefly discuss how our techniques can also give a learning guarantee even when the covariances of the components are degenerate. The reason is that we can simply stop the reverse diffusion $\delta$ time steps early. Instead of approximately sampling from the original mixture $\mathcal{M}$, this would approximately sample in total variation from a slightly noisy version of $\mathcal{M}$, namely the distribution $\mathcal{M}_\delta$ given by starting at $\mathcal{M}$ and running the forward process for a small amount of time $\delta$. Given a component $\calN(\vmu_i,\cov_i)$ of $\mathcal{N}$, the corresponding component of $\mathcal{M}_\delta$ is given by $\calN(e^{-\delta}\vmu_i, e^{-2\delta}\cov_i + (1 - e^{-2\delta})\Id)$. In particular, the minimum singular value of the covariance is at least $1 - e^{-2\delta} = \Omega(\delta)$, and we can thus apply \Cref{thm:main_informal} to $\mathcal{M}_\delta$ instead of $\mathcal{M}$, incurring exponential dependence on $\poly(1/\delta)$. Moreover, the Wasserstein distance between $\mathcal{M}$ and $\mathcal{M}_\delta$ scales with $\delta(R + \poly(\beta/\alpha))$. Altogether, we find that we can sample from a distribution that is TV-close to a distribution which is Wasserstein-close to $\mathcal{M}$, even if $\mathcal{M}$ might have degenerate covariances.
\end{remark}

\subsection{Approximating the score function using piecewise polynomials}
We now present the key ideas behind our main technical result showing that a piecewise polynomial
approximation of the score function exists.  In the following discussion, we will be focusing on 
estimating the score function of the Gaussian mixture at a specific noise level $t$.  At noise level
$t$, each component of the mixture is rescaled by $e^{-t}$ and convolved with a mean-zero Gaussian with covariance  $(1 - e^{-2t}) \Id$ (see \Cref{inf-eq:OU-equation}).  Therefore, the score function at every noise level
corresponds to the score function of a Gaussian mixture with 
means $e^{-t} \vec \mu_i$ and covariances $e^{-2 t} \cov_i + (1- e^{-2t}) ~\Id $, where 
$\vec \mu_i$ and $\cov_i$ denote the parameters of $i^{\text{th}}$ component of the original target mixture $\mathcal M$. For simplicity, we assume that the minimum mixing weight of the mixture $\mc M$ is at least $\poly(1/k)$ in the following discussion. 
It turns out that the bottleneck is to approximate the score function of the original
mixture $\mathcal M$ and therefore, to keep the notation simple, for this presentation
we will focus on this problem. We will denote the score
function (i.e., the gradient of the log-density) of a mixture of Gaussians by $\vec s(\x; \mathcal M)$:
\begin{equation}
\label{inf-eq:score-function}
    \vec s(\x; \mathcal M) = -\sum_{i=1}^k w_i(\x) \underbrace{\cov_i^{-1}( \x - \vec{\mu}_i ) }_{\vec g_i(\x)} \hspace{5mm} \text{where} \hspace{4mm} w_i(\x) = \frac{ \lambda_i \normal(\vec \mu_i, \cov_i; \vec x)}
    {\sum_{j=1}^k \lambda_j \normal(\vec \mu_j, \cov_j; \x)}  
\end{equation}

\begin{proposition}[Informal - Efficient Piecewise Polynomial Approximation - \Cref{prop:piecewise-poly-approx-score}]
\label{inf-prop:cluster-poly-approx}
    Let $\mc M$ be a $\tau$-well-conditioned mixture of $k$ Gaussians.
    There exists a function $\classify(\cdot) : \R^d \mapsto [\numclust]$ and polynomials
    $p_1,\ldots, p_{\numclust}$ of degree at most $\ell = \poly(k\tau/\eps)$ such that 
    \( \E_{\x \sim \mathcal M}[\| s(\x; \mathcal M) - p_{\classify(\x)}(\x) \|^2 ] \leq \eps \ , \)
    Moreover, there exists an efficient algorithm that with high-probability 
    finds this piecewise polynomial approximation with $d^{\poly(\ell)}$ samples 
    and runtime.
\end{proposition}

\paragraph{Why piecewise polynomials?}
We first give some intuition behind the structure of the score function of a Gaussian mixture, and its piecewise polynomial approximation.  
We observe that the score function (see \Cref{inf-eq:score-function}) is a weighted combination of linear functions.  For example, for a mixture of two standard 
one-dimensional Gaussians with means at $-R$ and $R$, it behaves (approximately)
like the function $-\1\{x \leq 0\} (x+R) - \1\{x\geq 0\} (x-R)$, see the left figure in \Cref{fig:clustering-vs-polynomial-apx}.  
We observe that the total length of support of the mixture is roughly an interval of length $O(R)$
and the slope of the score function is approximately $O(R)$ close to the origin.
We would like to have a polynomial approximation 
of degree $\poly(\log R/\eps)$ for this instance but naively applying polynomial approximation results (see, e.g., Jackson's theorem, \Cref{lemma:multivariate-jackson}) would yield a degree $\poly(R/\eps)$ even for 1-dimensional mixtures.\footnote{When dealing with $d$ dimensional mixtures things are even worse since the effective support has a radius depending on the dimension $d$.}
Therefore, as we observe in \Cref{fig:clustering-vs-polynomial-apx}, two reasons prohibit us from applying polynomial approximation results in a black-box manner: (1) 
the total support is of radius $R$ and (2) there are regions (far from the mixture means) where the
slope of the score function is also large (also $R$).

For the case of two Gaussians, we see that the ``effective'' support is much smaller
(intervals of size roughly $\sqrt{\log(1/\eps)}$ around the means).  Moreover, by focusing on
the ``effective'' support we also avoid the area where the derivative of the score function is
large (close to the origin).  Thus one could hope to solve both issues discussed above by 
creating an interpolating polynomial by concentrating the nodes on the effective support.  
Such an approach would work when the support consisted of actual ``hard'' intervals (and not ``approximate'' intervals with Gaussian tails).  The main issue is a race condition 
between the value of the interpolating polynomial far from the interpolation nodes (roughly exponentially large in the degree) and the decay of the Gaussian density.
While this race condition can be solved in some special cases (such as for mixtures of two Gaussians
with very well-separated means on $-R$ and $+R$), in general when more Gaussians are present in 
the mixture, the mental image of a union of ``hard'' intervals is incorrect
and it is not clear that the tails will always be able to cancel out the large error of the polynomial far from the interpolation intervals.  

The above structure of the score function naturally leads to a piecewise polynomial approximation 
approach.  For the symmetric mixture of two Gaussians discussed above there is an obvious candidate
for the partition: we should perform polynomial approximation in $\poly(\log(R/\eps))$ sized intervals around $\pm R$ and output zero in the rest of the space. That would lead to the desired degree of
$\poly(\log(R/\eps)/\eps)$.
For the more complicated example of the right figure of \Cref{fig:clustering-vs-polynomial-apx} we could similarly try to split the instance in
an interval containing almost all the mass of two left components and one interval containing the 
three right components and perform polynomial approximation (and output zero out of those two intervals).
In both examples, by using the piecewise polynomial approximation we avoided both issues discussed earlier, i.e., using polynomial approximation over large intervals or approximating over intervals where the derivative of the score is large.

\paragraph{Clustering and polynomial approximation: a win-win analysis}
Piecewise polynomial regression is a computationally hard, non-convex problem 
when we search both for the polynomials and for the partition of the space.
Therefore, we have to make sure that we have an efficient algorithm to find the partition of the space and then apply polynomial regression inside each cell of the partition.  Our main algorithm
is enabled by a win-win argument in the sense that the areas where polynomial approximation requires
high degree (i.e., $\poly(R)$) can be easily avoided by a crude clustering algorithm and the 
areas where the clustering algorithm fails to separate between a set of components of the mixture are those where the polynomial approximation
is effective.

\begin{figure}
\centering
\includegraphics[width=.48\textwidth]{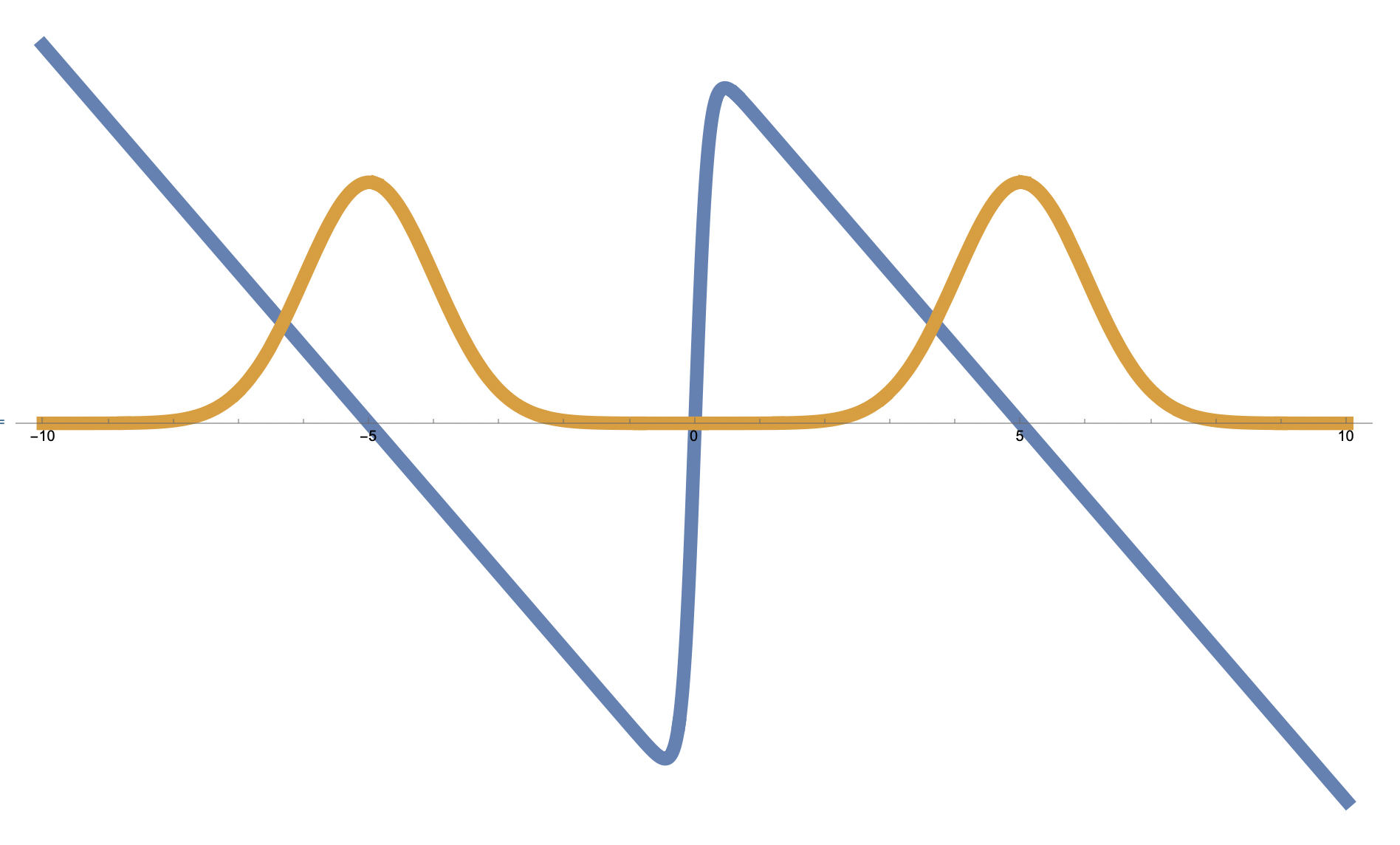}
~
\includegraphics[width=.48\textwidth]{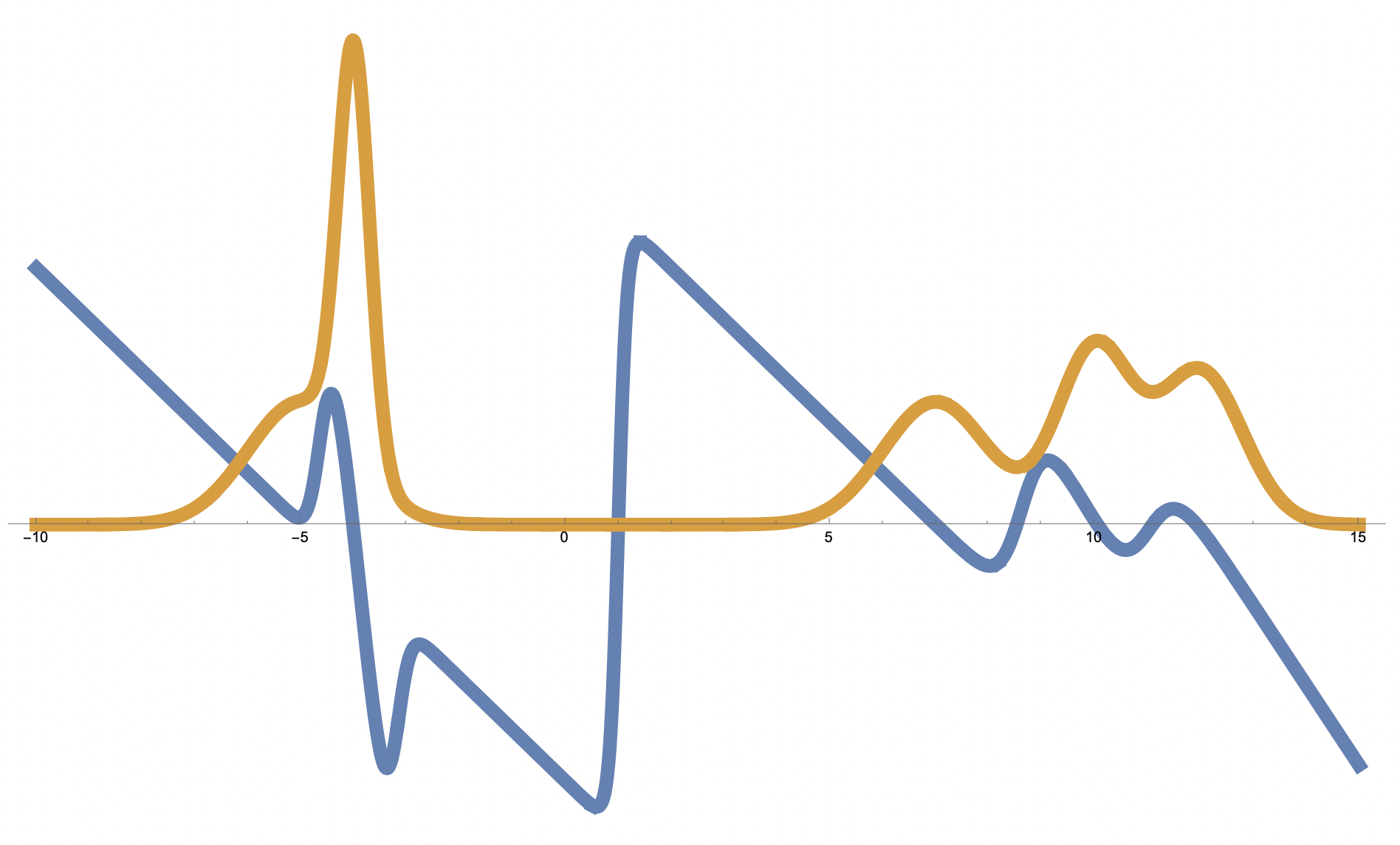}
\caption{When approximation is hard, clustering is easy. On the left figure,
we plot the density (gold) and score function (blue) of mixtures of two standard Gaussians with well-separated means (their distance is $R$). 
We observe that in that case, the score function
is (almost) a piecewise linear function with a large slope, i.e., roughly $R$, 
close to the origin.  In the right image, we have a mixture of $5$ Gaussians with different means and variances
that can be split into two clusters: a group of $2$ on the left and $3$ on the right.  Again the area where
the derivative of the score function (blue) is high, falls in between the two clusters (where the Gaussian density is exponentially small).  In both cases, a piecewise polynomial approximation yields the correct degree
that scaling with $(\log R)/\eps$ instead of $R/\eps$. Moreover, we expect that it is easy to cluster 
the points in the corresponding sub-mixtures that have much smaller effective support than the original
mixture.
}
\label{fig:clustering-vs-polynomial-apx}
\end{figure}

\subsection{Approximating the score given a crude partition}

As we observed in the previous examples,  the main difficulty in providing a polynomial approximation of the score function arises when it involves multiple Gaussians that are far apart.  We first make
more precise the notion of ``crude'' clustering \footnote{We use the terms ``clustering'' and ``partition'' function interchangeably.} that we require. 
\begin{definition}[$(\withincluster, \betweencluster)$-separated partition]
\label{inf-def:in-n-out}
Given a mixture of Gaussians 
$\normal_1= \normal(\vec \mu_1, \cov_1) ,\ldots, \normal_k = \normal(\vec \mu_k,\cov_k)$, 
we require that the clustering function $\classify(\x)$ assigns
each $\x \in \R^d$ to one of $n_c$ subsets $U_1,\ldots, U_{n_c}$ of $[k]$ 
that form a partition of the original $k$ components such that:
\begin{enumerate}
\item  If $\normal_i, \normal_j$ belong in different subsets $U_t$
and $U_{t'}$, they have to be \emph{at least} $\betweencluster = \poly( \tau k \log (1/\eps) )$ far in parameter distance, i.e.,
$\dpar(\normal_i, \normal_j) = \|\vec \mu_i - \vec \mu_j\|_2 + \|\cov_i - \cov_j\|_F \geq \betweencluster$.
\item  If $\normal_i, \normal_j$ belong in the same subset $U_t$, they have to be \emph{at most} $\withincluster = \poly( \tau k \log (1/\eps) )$ far in parameter distance, i.e.,
$\dpar(\normal_i, \normal_j) \leq \withincluster$.
\item  $\classify(\x)$ is consistent with the partition $U_1,\ldots, U_t$ with high-probability, i.e.,
for any $i \in U_t$, $\pr_{\x \sim \normal_i}[\classify(\x) \neq t] \leq \epspart$, where $\epspart$ is a small error parameter.
\end{enumerate}
\end{definition}

Given the above $(\withincluster, \betweencluster)$-partition, our proof consists of two steps: 
(i) show that we can reduce the original problem of approximating the score function of the whole mixture to approximating the score function of the sub-mixtures $U_t$ and (ii) 
providing low-degree approximations of the sub-mixture score functions.  We describe 
these steps in the next two paragraphs.

\paragraph{Simplifying the score}
As we discussed, the first obstacle in approximating the score function is that
it is a function over a domain of radius $\poly(R)$ (inducing a $\poly(R)$ dependency on the degree).
Fortunately, there is an additional structure connecting the weights $w_i(\x)$ and 
the linear terms $\vec g_i(\x)$.  We use this structure to prove that when $\x$ is sampled 
from some component $\normal_i$ then on expectation over the component $\normal_i$ we can remove
a term in the score function corresponding to a component $\normal_j$ that is far 
from $\normal_i$ without introducing large error, see \Cref{lem:score-removal-single}.
More precisely, we show that given a partition function $\classify(\cdot)$ that satisfies 
\Cref{inf-def:in-n-out}, for all $\x$ where $\classify(\x) = t$, 
we can ``simplify'' the score function by removing
the contribution of all components $\normal_j$ that do not belong in $U_t$.

Given a subset $U_t$ of indices of $[k]$, we denote by 
$\mathcal M(U_t)$ the submixture containing the components $\normal_i$ for $i \in U_t$ and by 
$\vec s(\x; \mathcal{M}(U_t))$ the score function containing only the contribution of
components from $U_t$, i.e., 
\[
\vec s(\x; \mathcal M(U_t)) = - \sum_{i \in U_t} \vec g_i(\x) \frac{ \lambda_i \normal_i(\x)}{\sum_{j \in U_t}\lambda_j \normal_j(\x)}
\]
We prove the following proposition showing that,
inside each cell $t$ of the partition given by $\classify(\cdot)$, we can replace the original score function $\vec s(\x; \mathcal M)$ by the score function of the 
sub-mixture $\vec s(\vec x;\mathcal M(U_t))$.
Each sub-mixture score function corresponding to $U_t$ contains components that are all 
$\withincluster$-close to each other, thus reducing the effective radius of the approximation domain to $\poly(\log R)$.
\begin{proposition}[Informal -- Score Simplification, see \Cref{prop:score-simplification}]
\label{inf-prop:score-simplification}
Fix $\eps > 0$. let $\mathcal M$ be a $\tau$-well-conditioned mixture of $k$ Gaussian distributions and satisfies $\|\vec \mu_i \|_2 + \|\cov_i - \Id\|_F \leq R$ for all the components.
Moreover, assume that $\classify$ satisfies \Cref{inf-def:in-n-out}.
Define the following piecewise approximation to the score function
\(
s(\x; \classify(\cdot)) = 
\sum_{t=1}^{\numclust} s(\x; \mathcal M(U_t)) ~ \1\{ \classify(\x) = t \} \,.
\)
It holds that
\(
\E_{\x \sim \mathcal M}
[\| s(\x; \mathcal M) - s (\x; \classify(\cdot)) \|_2^2 ] 
\leq \poly(k \tau R) \sqrt{\eps}  .
\)
\end{proposition}

\paragraph{Polynomial approximation of the simplified score} 
Recall from Eq.~\eqref{inf-eq:score-function} that the score function for any Gaussian mixture is a sum of the softmax function $w_i(\x)$ multiplied by a linear function $\cov_i^{-1}(\x - \vmu_i)$. A polynomial approximation of the softmax will provide a polynomial approximation for the simplified score. Note that we want to approximate the simplified score with the degree at most $\poly(k \tau / \eps)$ to obtain runtime of polynomial regression of $O(d^{\poly(k \tau / \eps)})$.

The degree of a polynomial approximation of a function generally depends on the domain of the approximation and smoothness of the function (in terms of the norm of its gradient), see \Cref{lemma:multivariate-jackson}. The softmax function is smooth and has a bounded gradient but the input to the softmax is $\{ \| \x - \vmu_i  \|^2_{\cov_i^{-1}} \}_{i=1}^{|U_t|}$ which can be as large as $\poly(d)$ and hence, the degree of the naive polynomial approximation could be $\poly(d/\eps)$.

To overcome this issue, we show that even though each input $\| \x - \vmu_i \|_{\cov_i^{-1}}^2$ is large, there exists a normalization of the softmax for which the inputs to the softmax are $\poly(\tau \withincluster)$. More precisely, we normalize the softmax such that $\{ \| \x - \vmu_i \|^2_{\cov_i^{-1}} - \| \x - \vmu_1 \|^2_{\cov_1^{-1}} - \iprod{ \cov_1, \cov_i^{-1} - \cov_1^{-1} } \}_{i=1}^{|U_t|}$ are the inputs to the softmax function and show that its norm is $\poly(\tau \withincluster)$ with high probability. Therefore, using multivariate Jackson's theorem (\Cref{lemma:multivariate-jackson}), we obtain the polynomial approximation for the softmax function and hence, for the simplified score function. 

\begin{lemma}[Informal - See \Cref{lem:poly-approx-score-bounded-interval}] Let $\mathcal M(U)$ be a $\tau$-well-conditioned mixture of $k$ Gaussian distributions restricted to the subset of components in $U$. Then, there exist a polynomial $p(\x; \mathcal M(U))$ of degree $\poly(\tau \withincluster / \eps)$ and coefficients bounded in magnitude by $d R \exp(\poly( \tau \withincluster / \eps ))$ such that for $\x \sim \mathcal M(U)$, with high probability, the polynomial satisfies
\(
    \| s(\x; \mathcal M(U)) - p(\x; \mathcal M(U))  \| \leq \eps.
    \)
\end{lemma}

\subsection{Crude clustering via PCA}  
We now describe our crude clustering algorithm for obtaining the partition satisfying
the assumptions of \Cref{inf-def:in-n-out}.  Our approach consists of two main steps: (1)  
approximately recover the span of the means and covariances using PCA on the second and 
fourth-order moment tensors of the mixture and (2) recover estimates of the parameters by 
brute forcing over the $k$-dimensional subspace recovered in the first step and using
pairwise log-likelihood tests to create the final partition function.

\paragraph{Obtaining estimates of means and covariances}

The algorithm operates in two phases. First, we obtain a crude estimate for the subspace spanned by the means, after which we brute-force within this low-dimensional subspace to find points close to each of the means. Second, we use these mean estimates to form an estimator for the subspace spanned by the covariances, after which we can similarly brute-force to find points close to each of the covariances.  
With roughly $d^{O(k)}$ runtime, we can construct a list of candidate parameters for the means
and covariances of the mixture containing crude (in the sense that they can be $\poly(k \tau)$-far) of the target parameters.

\begin{lemma}[Informal -- Recovering crude estimates of the parameters, see \Cref{lem:crude_param}]
\label{inf-lem:crude_param}
    There is an algorithm that returns a list $\calW$ such that for every $i\in[k]$, there exists $(\hatmu_i, \wh{\cov}_i) \in \calW$ for which $\norm{\vmu_i - \hatmu_i}^2 \lesssim \beta/\lambdamin$ and $\norm{\cov_i - \wh{\cov}_i}_F \lesssim k^{3/2}\beta/\lambdamin + k^2\alpha\log\radius$. Furthermore, $|\calW| \le (R/\sqrt{\beta})^{O(k^2)}\cdot d^{O(k)}$, and the algorithm runs in time $(\radius/\sqrt{\beta})^{O(k^2)}\cdot(\poly(d\radius/\beta) + d^{O(k)})$ and draws $\poly(d\radius/\beta)$ samples.
\end{lemma}

\noindent We use PCA on the covariance of the mixture $\vec M = \E_{\x \sim \mathcal M}[\x \x^\top]$ 
to obtain the subspace spanned by the means.  We observe that 
$\vec M = \sum_{i=1}^k \lambda_i \vec \mu_i \vec \mu_i^\top + \sum_{i=1}^k \lambda_i \cov_i$.
The main idea here is to think of $\vec M$ as approximately low-rank and treat the contribution
of the covariances as an error $\calE =  \sum_{i=1}^k \lambda_i \cov_i$.
Since the covariance matrices $\cov_i$ are well-conditioned (i.e., their eigenvalues are not bigger
than $\beta$ (see \Cref{def:well-conditioned-mixture}) we can show that if some $\vec \mu_i$
is larger than $\beta/\lambdamin$ then its contribution in $\vec M$ cannot be ``hidden'' by the error term $\calE$ and will have a large projection onto the subspace spanned by the top eigenvectors
of $\vec M$.  The proof of this claim follows by a standard argument for $k$-SVD and can be found in
\Cref{sec:crude-means}.

Finding estimates for the covariances is more complicated but similarly relies on recovering the subspace 
spanned by  the low-rank components of the (flattened) fourth-order tensor 
\begin{equation*}
    \vec \Psi = \E_{\x \sim \mathcal M}[\vc{ \vec x \vec x^\top} \vc{ \x \x^\top}]\,.    
\end{equation*}
The intuition behind our approach is that if the means of the mixture were all sufficiently close to zero, then the top-$k$ singular subspace of the matrix $\vec \Psi$ 
can be shown to contain points close to $\vc{\cov_1},\ldots,\vc{\cov_k}$. In general, if the means are arbitrary, then we can use the estimates $\hatmu_1,\ldots,\hatmu_k$ derived in the previous section to approximately ``recenter'' the mixture components near zero.  Since the means recovered in the previous step were already
crude $\poly(k)$ approximations of the true means a careful error analysis must be done so that this 
recentering does not introduce significantly more error (i.e., depending on the dimension $d$) in the covariance
estimates.  We refer to \Cref{sec:crude-parameters} and \Cref{alg:cov_span} for more details.

\noindent 

\paragraph{Clustering using the log-likelihood ratios}

We now present our main clustering guarantee, which leverages the estimates for the parameters we obtained previously. As those estimates are only crude approximations to the true parameters, we will obtain a commensurately crude clustering. 

Our algorithm starts by brute-forcing over mean-based and covariance-based partitions $\calS$ (resp. $\calT$).
$\calS$ (resp. $\calT$) partitions the mixture components into groups such that any two components in the same group have means (resp. covariances) that are not far, and any two components from two different groups have means (resp. covariances) that are not close. Their common refinement is a partition $\calU$ satisfying
the assumptions of \Cref{inf-def:in-n-out}: any two components in the same group have both means and covariances not too far, and any two components from two different groups either have means not too close or covariances not too close.

By brute-forcing over pairs of partitions of $[k]$ (of which there are at most $k^{2k}$) we may assume we have access to $\calS$ and $\calT$, and thus to $\calU$. Our goal is then to assign to every $\x\in\R^d$ an index into the partition $\calU$. For $\x$ which is sampled from the $i$-th component of the mixture which belongs to the $t$-th group in $\calU$, we would like our assignment to be $t$ with high probability. 
At a high level, the idea is as follows. It is not too hard to determine which group in $\calS$ a given point $\x$ should belong to, simply by checking which mean estimate $\hatmu_i$ is closest to $\x$ after projecting to the subspace spanned by $\hatmu_1,\ldots,\hatmu_k$. For each group in $\calS$, we can then effectively restrict our attention to components within that group and focus on clustering them according to their covariances. Roughly speaking, we accomplish this by comparing log-likelihoods of sampling $\x$ under $\calN(\hatmu_1,\wh{\cov}_1),\ldots,\calN(\hatmu_k,\wh{\cov}_k)$ and choosing the group in $\calT$ containing the component maximizing log-likelihood.  For more details, we refer to \Cref{sec:likelihood-clustering} and to \Cref{lem:main_cluster} for the formal clustering statement that we prove.

\subsection{Avoiding the doubly exponential dependency on \texorpdfstring{$k$}{k}}
\label{sec:doubly-exponential}

Here we provide some intuition for the origin of the doubly exponential dependence on $k$ which is implicit in existing works on learning mixtures of general Gaussians with the method of moments, and how our technique outlined above avoids this issue. Our starting point is the algorithm of~\cite{moitra2010settling}; in fact, for this discussion, it will suffice to consider the case of $d = 1$ and components of variance 1. 

Specialized to this case, in the analysis in~\cite{moitra2010settling}, the authors first proved that if all of the components have means with nonnegligible separation, say $\eta$, from each other, then one can learn the means by brute-forcing over a grid with sufficiently small granularity and finding a setting of parameters in this grid for which the corresponding mixture matches the first $O(k)$ moments with the target mixture to error $\eta^k$ (here we ignore constants in the exponent for simplicity).

Now what happens if the minimal separation $\eta$ is arbitrarily small? The authors noted that for means that are particularly close, one can simply ``merge them'': they are statistically close to a single component, and in a bounded number of samples one would not be able to tell the difference. Because the number of samples used by the algorithm outlined above is $(1/\eta)^k$, this implies that if there is some scale $\eta$ at which there is a \emph{gap} in the sense that all means are either $\eta^k$-close or $\eta$-far apart, then one can learn in the same amount of time/samples as in the $\eta$-separated case. 

The last question that remains is how to ensure such a scale exists. The idea is that if one looks at $k^2+1$ consecutive windows $\{[\eta^{k^i}, \eta^{k^{i-1}}]\}_{i=1,\ldots,k^2+1}$, by pigeonhole principle there must exist some window such that the separation between any pair of means lies outside this window. At that scale, one can apply the above reasoning to learn the means. This is the origin of the doubly exponential scaling in $k$ that is present in all existing algorithms for learning mixtures of general Gaussians, including the state-of-the-art guarantee of~\cite{bakshi2022robustly}.

It is instructive to contrast this with our approach. The main reason for the doubly exponential dependence in the above windowing argument was that one needed a scale at which the components break up into ``gapped clusters'' such that the separation within clusters is significantly smaller than the separation across clusters. For this clustering structure to exist, we need to go down potentially to a doubly exponentially small scale. In contrast, in our work, we make do with a very crude clustering for the purposes of our piecewise regression. We simply require that for components from different clusters, their parameter distance is sufficiently large, while for components from the same cluster, their parameter distance is \emph{not too large}. Crucially, we don't need to make any assumption about a gap between the intra- versus inter-cluster separations, ensuring we avoid the doubly exponential dependence on $k$.

\section{Diffusion models and other technical preliminaries}

In this section, we collect various technical ingredients. The bulk of this section is dedicated to an exposition of diffusion models in~\Cref{sec:diffusion_details}.

\subsection{Notation for mixture models}

Throughout the paper, we use either $q$ or $q_0$ to denote the data distribution on $\R^d$, i.e., the mixture of Gaussians with means $\vec \mu_1, \vec \mu_2, \ldots, \vec \mu_k$, covariances $\cov_1, \cov_2, \ldots, \cov_k$, and mixing weights $\lambda_1,\ldots,\lambda_k$ respectively. We will use $\mathcal N_i$ to denote the distribution for its $i$-th component, i.e. $\calN(\vec \mu_i,\cov_i)$. We use $p$ or $p_T$ to denote the learned distribution.

\begin{definition}
Let $\mathcal M = \frac{1}{k} \sum_{i=1}^k \normal(\vec \mu_i, \cov_i)$ be a $(\alpha, \beta, R)$-well-conditioned Gaussian mixture.
We say that a partition of $[k]$ into subsets $S_1,\ldots, S_m$ is $(\withincluster, \betweencluster)$-separated 
if for all $i, j \in S_\ell$ it holds that 
$\| \vmu_i - \vmu_j \| + \|\cov_i - \cov_j\|_F \leq \withincluster$ and for all $i \in S_\ell$, $j \in S_{\ell'}$
for $\ell \neq \ell'$ it holds $\| \vmu_i - \vmu_j \| + \|\cov_i - \cov_j\|_F \geq \betweencluster$.
We denote by $\mathcal M(S_i)$ the mixture distribution corresponding to the components of $S_i$, i.e., $\mathcal M(S_i) = \frac{1}{|S_i|} \sum_{j \in S_i} 
\normal(\vec \mu_j, \cov_j)$.
\end{definition}

Moreover, given a mixture $\mathcal M = \sum_{i=1}^k \lambda_i \mathcal D_i$ we denote by $\mathcal M^J$ the joint distribution over tuples
$(j, \x)$ where $j = i$ with probability $\lambda_i$ and, conditional on $j = i$, $\x$ is drawn from $\mathcal D_i$.

\subsection{Learning Gaussian mixtures via a denoising diffusion process}
\label{sec:diffusion_details}
We start by introducing some standard terminology and notation
on diffusion models.  We will be using the diffusion algorithmic template in a more or less
black box manner and therefore we try to keep the presentation short but still self-contained.
Throughout the paper, we use either $q$ or $q_0$ to denote the data distribution on $\R^d$. The two main components in diffusion models are the \emph{forward process} and the \emph{reverse process}. The forward process transforms samples from the data distribution into noise, for instance via the \emph{Ornstein-Uhlenbeck (OU) process}: 
\begin{equation*}
    \mathrm{d} \x_t = - \x_t \, \mathrm{d} t + \sqrt{2} \,\mathrm{d} \vec w_t \;\;\;\text{with} \;\;\; \x_0 \sim q_0\,,
\end{equation*}
where $(\w_t)_{t\ge 0}$ is a standard Brownian motion in $\R^d$.
We use $q_t$ to denote the law of the OU process at time $t$. Note that for $\x_t \sim q_t$,
\begin{equation}
\label{eq:Xt-density}
    \x_t = \exp(-t) \x_0 + \sqrt{1 - \exp(-2t)} \vec z_t \;\;\; \text{with} \;\; \x_0 \sim q_0, \;\; \vec z_t \sim \mc{N}(0, \Id)\,.
\end{equation}

The reverse process then transforms noise into samples, thus performing generative modeling. Ideally, this could be achieved by running the following stochastic differential equation for some choice of terminal time $T$:
\begin{equation}
    \mathrm{d} \x^\leftarrow_t = \{\x^\leftarrow_t + 2\nabla_{\x} \ln q_{T-t}(\x^\leftarrow_t)\}\, \mathrm{d}t + \sqrt{2}\,\mathrm{d} \w_t \;\;\;\text{with} \;\;\; \x^\leftarrow_0 \sim q_T\,,
\end{equation}
where now $\w_t$ is the reversed Brownian motion. In this reverse process, the iterate $\w^\leftarrow_t$ is distributed according to $q_{T - t}$ for every $t\in[0,T]$, so that the final iterate $\x^\leftarrow_T$ is distributed according to the data distribution $q_0$. 
The function $\nabla_{\x} \ln q_t$ is called the \emph{score function} and is required so that
we are able to run the reverse SDE and generate samples from the unknown distribution.  Ideally, we would like to have access to an approximate 
oracle $\wh {\vec s}(\x)$ such that for all $t \in [0, T]$ it is a good approximation to the score function $\nabla_{\x} \log q_t(\x)$:
\begin{equation}
    \E_{\x_t \sim q_t}[ \| \nabla_{\x} \ln q_t(\x_t) - \wh {\vec s_t}(\x_t) \|^2 ] 
    \leq \epscore \,.
    \end{equation}
To obtain such a function $\wh{\vec s}_t(\x)$, one would an expressive enough set of candidate functions $\mathcal G$ and then try to optimize the score matching loss:
\begin{equation*}
\min_{\vec g_t \in \mathcal G}
    \E_{\x_t \sim q_t}[ \| \nabla_{\x} \ln q_t(\x_t) - \vec g_t(\x) \|^2 ] 
    \leq \epscore \,.
\end{equation*}
However, as the density function of $q_t$ is unknown the above minimization 
problem cannot be solved directly.
A standard calculation (see e.g. Appendix A of~\cite{chen2023sampling}) shows that this is equivalent to minimizing the \emph{DDPM objective} in which one wants to predict the noise $\vec z_t$ from the noisy observation $\x_t$, i.e.
\begin{equation}
\label{eq:diffusion-loss}
    \min_{\vec g \in \mathcal G} \;\; L_t(\vec g_t) = \E_{\x_0, \vec z_t} \Biggl[ \Big\| \vec g_t(\x_t) + \frac{\vec z_t}{ \sqrt{1 - \exp(-2t)} } \Big\|^2 \Biggr]\,.
\end{equation}
In this work we focus specifically on the optimization problem~\eqref{eq:diffusion-loss} and show that it can be solved efficiently
when the underlying target density $q_0$ is a mixture of $k$ Gaussian distributions.

\begin{algorithm2e}
\DontPrintSemicolon
\caption{\textsc{GenerateSample}}
\label{alg:generate-sample}
    \KwIn{Score estimation error $\epscore$, confidence $\delta$,
    sequence of time steps $t_0, t_1,\ldots,t_N$ }
    \KwOut{A sample $\vec y_N \in \R^d$}

     \
    \For{$\ell \in\{0, \ldots, N-1\}$}{ 
       $\wh{\vec s}(\cdot, T-t_{\ell} ) \gets \textsc{LearnScore}(t_\ell, \epscore, \delta)$
       \Comment{Learn the score function at all time steps} \;
    }
 
    \For{$\ell \in\{0, \ldots, N-1\}$}{
        Set $\rho_\ell = e^{(t_{\ell+1} - t_\ell)/2}$ \; 
        Sample $\vec z_{\ell} \sim \normal(\vec 0, \Id)$ \;
        $ \vec y_{\ell+1} \gets ~ \rho_\ell ~ \vec y_{\ell} + 2 (\rho_\ell - 1)
        ~ \wh{\vec s}(\vec y_{\ell}, T-t_{\ell}) + \sqrt{\rho_{\ell}^2 - 1} ~ \vec z_\ell$ ~ \Comment{Run the (discretized) reverse SDE}
    }
    
        \Return{$\vec y_N$}
\end{algorithm2e}

We are now ready to present and prove our main result: an efficient algorithm for learning well-conditioned 
GMMs.  

\begin{theorem}[Efficient Sampler for GMMs] \label{thm:generate-sample}
Fix $\eps, \delta \in (0,1)$ and let $\mathcal M$ be an $(\alpha, \beta, R)$-well-conditioned mixture of $k$ Gaussians. Let $\tau = (\beta/\alpha) \, \log R$,
$\eps_{\text{score}} = \eps / \log ( R / (\alpha \eps) )$, $\delta = \alpha \eps / R$, and let
time sequence $t_1,\ldots, t_N$ be as defined in \Cref{lem:diffusion-convergence}.  Then with probability at least $1 - \delta_{\text{f}}$,
\Cref{alg:generate-sample} draws 
$M = d^{\poly(k \tau/(\lambdamin \eps))} \log \frac{1}{\delta_{\text{f}}}$ samples from $\mathcal M$,
runs in sample-polynomial time, and generates a sample
$\vec y_N$ whose distribution is $\eps$-close in total variation to $\mathcal M$.
\end{theorem}
\begin{proof}

We are going to
use the following result on the convergence of the discretized reverse SDE with the score 
approximation that we use in \Cref{alg:generate-sample}.

\begin{lemma}[Convergence given approximate scores, \citep{benton2023linear}]
\label{lem:diffusion-convergence}
Fix some $\delta \in(0,1)$, $T \geq 1$ and let $N$ be some even integer larger 
than $\log(1/\delta)$ and let $\kappa > 0$ be larger than a sufficiently large 
constant multiple of $(T + \log(1/\delta))/N$.
Set $t_0 = 0$, $t_{N/2} = T-1$,  $t_N = T-\delta$.
Moreover, set $t_1,\ldots, t_{N/2-1}$ equally spaced on $[0, T-1]$, i.e.,
$t_{\ell+1} - t_{\ell} = \kappa > 0$ for all $\ell \in\{0, \ldots, N/2-1\}$
and $T-t_{N/2+1},\ldots T-t_{N-1}$ exponentially decaying, i.e.,
$t_{N/2+\ell+1}- t_{N/2+\ell} = \kappa/(1+\kappa)^{\ell}$ for all $\ell \in\{0,\ldots, N/2-2\}$ and $\gamma_\ell \leq \kappa \min(1, T - t_\ell)$. Assume that the data distribution and the score function satisfy the following assumptions.
\begin{enumerate} 
\item $\sum_{\ell=0}^{N - 1} \gamma_\ell \E_{\x \sim q_{t_\ell}}[\| \nabla \log q_{T-t_\ell}(\x) - \wh{\vec s}(\x, T-t_\ell)\|_2^2] \leq \epscore^2$. 
\item The target distribution $q_0$ on $\R^d$ has finite second moment.  
\end{enumerate}
For any $t \in [0, T]$ denote by $q_t$ the distribution of $\exp(-t) \vec x_0 + \sqrt{1-\exp(-2t)}
\vec z_t$, where $\vec x_0 \sim q_0$ and $\vec z_t \sim \normal(\vec 0,\Id)$
and denote by $p_{t_N}$ the distribution of the output $\vec y_N$ of \Cref{alg:generate-sample}. 
It holds that 
\[
\kldistance(q_{\delta} \| p_{t_N}) \lesssim \epscore^2 +\kappa^2 d N 
+ \kappa d T + \kldistance(q_{T} \| \normal(\vec 0, \Id))\,. 
\]
\end{lemma}

We first show that the guarantee of \Cref{lem:diffusion-convergence} yields
a total variation bound between $p_{t_N}$ and the target Gaussian mixture  $\mathcal M$.  By Pinsker's inequality, we obtain that 
$\tvdistance(p_{t_N}, q_\delta) \lesssim \sqrt{\kldistance(q_\delta\|p_{t_N}) }$.
Moreover, by a triangle inequality, we obtain that 
$\tvdistance(q_0, p_{t_N}) \leq \tvdistance(p_{t_N}, q_\delta) + \tvdistance(q_\delta, q_0)$.  Therefore, we have to control $\tvdistance(q_0, q_\delta)$.  Using again Pinsker's inequality we obtain that
$\tvdistance(q_0, q_\delta) \lesssim \sqrt{\kldistance(q_0 \|q_\delta)}$.
To control the Kullback-Leibler divergence between the target
$q_0$ that corresponds to the well-conditioned mixture and 
$q_\delta$.  We observe that $q_\delta$ is also a Gaussian mixture
with parameters $\wh{\vec \mu}_i = \vec \mu_i \exp(-\delta)$ and $\wh{\cov}= \cov_i e^{-2 \delta} + (1-e^{-2 \delta}) ~ \Id$. We denote this mixture by $\mathcal M_\delta$  Since $\kldistance$ is convex we obtain
that 
\begin{equation*}
\kldistance(q_\delta \| q_0) \leq \sum_{i=1}^k \lambda_i 
\kldistance(\normal(\vec \mu_i, \cov_i)\|\normal(\wh{\vec \mu}_i, \wh{\cov}_i))
\leq 
\max_{i=1}^k \kldistance(\normal(\vec \mu_i, \cov_i)\|\normal(\wh{\vec \mu}_i, \wh{\cov}_i)) \,.
\end{equation*}
We can now use the following standard bound for the Kullback-Leibler distance
between two Normal distributions 
$\kldistance(\normal(\vec\mu_1, \cov_1)\|\normal(\vec \mu_2, \cov_2))
\lesssim \|\Id - \cov_2^{-1/2} \cov_1 \cov_2^{-1/2}\|_F^2 + \|\cov_2^{-1/2}(\vec \mu_1 - \vec \mu_2)\|_2^2 $.  We have that
\[
\| \cov_i^{-1/2} (\vec \mu_i - \vec \mu_i e^{-\delta})  \|_2^2
\leq \frac{R^2}{\alpha} (1- e^{-\delta})^2
\leq  \frac{R^2}{\alpha} ~ \delta^2 \,,
\]
where the last inequality follows by the fact that $\alpha \Id \preceq \cov_i $ and the fact that $\|\vec \mu_i\|_2 \leq R$
and the inequality $e^x \geq x + 1$.  Moreover, if $s_1,\ldots s_d$ are the eigenvalues of $\cov_1$, we have that
\[
 \|\Id - \cov_2^{-1/2} \cov_1 \cov_2^{-1/2}\|_F^2 
 = \sum_{i=1}^d \Big(1 - \frac{ s_i e^{-2 \delta} + (1 - e^{-2 \delta})}{s_i}\Big)^2
 =  (1-e^{-2\delta})^2 \sum_{i=1}^d \Big(\frac{ 1- s_i}{s_i}\Big)^2 
 \lesssim \frac{\delta^2}{\alpha^2} R^2 \,,
\]
where the last inequality follows by the assumption that 
$\|\cov_i - \Id\|_F^2 \lesssim R$ and the fact that $s_i \geq \alpha$ for all $i$.
Putting the above together, we obtain that 
$\tvdistance(p_{t_N}, \mathcal M) \lesssim \sqrt{\kldistance(q_\delta\|p_{t_N})}
+ \delta R /\alpha $.

Similarly, we have to control the convergence error of the forward OU process $\kldistance(q_T \| \normal(\vec 0, \Id))$.  Similarly to the above argument, 
by the convexity of the Kullback-Leibler, we obtain that it suffices to control
the KL divergence between any component of the mixture and the standard normal $\normal(\vec 0, \Id)$.  Using the same bound for the KL divergence as above, we have that 
\[
\kldistance(q_T \| \normal(\vec 0, \Id)
\leq \max_{i=1}^k (e^{-2 T} \|\vec \mu_i \|_2^2 
+ e^{-4 T} \|\cov_i - \Id\|_F^2)
\lesssim e^{-2 T} R^2 \,.
\]

To make the forward process converge to an $\eps$-approximate Gaussian, we take $T = \log (R/\eps)$.  We choose $\eps_{\text{score}} = \eps / \log ( R / (\alpha \eps) )$ and $\delta = \alpha \eps / R$. Additionally, we have $\gamma_\ell \leq \kappa$ for all $\ell$. Therefore, we have  
\begin{equation*}
    \sum_{\ell=0}^{N - 1} \gamma_\ell \E_{\x \sim q_{t_\ell}}[\| \nabla \log q_{T-t_\ell}(\x) - \wh{\vec s}(\x, T-t_\ell)\|_2^2] \leq \eps \; .
\end{equation*}
The above choice also yields $\kappa^2 d N \lesssim (\log^2(R/\alpha \eps) d) / N $ and $\kappa d T \lesssim (\log^2(R/\alpha \eps) d) / N$. Choosing $N = (\log^2(R/\alpha \eps) d) / \eps$ and combining all the terms in \Cref{lem:diffusion-convergence}, we obtain that $\tvdistance(p_{t_N}, \mathcal M) \leq \eps$. We obtain sample complexity and runtime of the algorithm by putting $\eps_{\text{score}} = \eps / \log ( R / (\alpha \eps) )$ and failure probability $\delta_f = \delta / N$ in \Cref{thm:learning-score-guarantee-fixed-t}. 
\end{proof}

\section{Obtaining crude estimates for the parameters}
\label{sec:crude-parameters}

In this section, we prove the next lemma showing that we can construct a list of candidates for the unknown parameters of the mixture, containing ``crude'' approximation to the true target parameters.

\begin{lemma}\label{lem:crude_param}
    There is an algorithm \textsc{CrudeEstimate}($q$) which returns a list $\calW$ such that for every $i\in[k]$, there exists $(\hatmu_i, \wh{\cov}_i) \in \calW$ for which $\norm{\vmu_i - \hatmu_i}^2 \lesssim \beta/\lambdamin$ and $\norm{\cov_i - \wh{\cov}_i}_F \lesssim k^{3/2}\beta/\lambdamin + k^2\alpha\log\radius$. Furthermore, $|\calW| \le (R/\sqrt{\beta})^{O(k^2)}\cdot d^{O(k)}$, and the algorithm runs in time $(\radius/\sqrt{\beta})^{O(k^2)}\cdot(\poly(d,1/\beta) + d^{O(k)})$ and draws $\poly(dR/\beta)$ samples.
\end{lemma}

\noindent The algorithm operates in two phases. First, we obtain a crude estimate for the subspace spanned by the means, after which we brute-force within this subspace to find points close to each of the means. Second, we use these mean estimates to form an estimator for the subspace spanned by the covariances, after which we can similarly brute-force to find points close to each of the covariances.

\subsection{Estimating the means}
\label{sec:crude-means}

This phase is straightforward: we simply take the top-$k$ singular subspace of the empirical second moment matrix (see \Cref{alg:mean_span} below). 

\begin{lemma}\label{lem:mean_net}
    There is an algorithm {\sc CrudeEstimateMeans}($q$) which returns a list $\calW$ such that for each $i\in[k]$, there exists $\hatmu_i\in\calW$ for which $\norm{\vmu_i - \hatmu_i}^2 \lesssim \beta/\lambdamin$. Furthermore, $| \mathcal W | \leq (\radius/\sqrt{\beta})^{O(k)}$, and the algorithm runs in time $\poly(d\radius/\beta) + (\radius/\sqrt{\beta})^{O(k)}$ and draws $\poly(d\radius/\beta)$ samples.
\end{lemma}

\begin{algorithm2e}
\DontPrintSemicolon
\caption{\textsc{CrudeEstimateMeans}($q$)}
\label{alg:mean_span}
    \KwIn{Sample access to $q$}
    \KwOut{List $\calW$ containing approximations to $\vmu_1,\ldots,\vmu_k$}
        Draw samples $\x_1,\ldots,\x_N$ from $q$ for $N \gets \poly(d\radius/\beta)$\;
        $\wh{\bM} \gets \frac{1}{N}\sum_i \x_i \x_i^\top$\;
        $\wh{V}\gets$ top-$k$ singular subspace of $\wh{\bM}$\;
        $\calW\gets$ a $\beta^{1/2}$-net over vectors in $\wh{V}$ with $L_2$ norm at most $2\radius$\;
        \Return{$\calW$}
\end{algorithm2e} 

The analysis (as well as subsequent parts of our proof) uses the following standard bound for $k$-SVD:

\begin{lemma}\label{lem:pca}
    Let $\bA = \sum^k_{i=1} \vec v_i \vec v_i^\top + \bE$ for $\norm{\bE}_{\rm op} \le \epsilon$. The top-$k$ singular subspace of $\bA$ contains vectors $\wh{\vec v}_1,\ldots,\wh{\vec v}_k$ for which $\norm{\vec v_i - \wh{\vec v}_i}^2 \le 2\epsilon$ for all $i\in[k]$.
\end{lemma}

\begin{proof}
    Define $\bA^* \triangleq \sum^k_{i=1} \vec v_i \vec  v_i^\top$. Let $\Pi^\perp$ denote the projector to the orthogonal complement of the top-$k$ singular subspace of $\bA$, and define $\vec r_i \triangleq \Pi^\perp \vec v_i$. Then
    \begin{equation*}
        \vec r_i^\top \bA^* \vec r_i = \sigma_{k+1}(\bA)\cdot \norm{\vec r_i}^2 - \vec r_i^\top \bE \vec r_i \le 2\epsilon \norm{\vec r_i}^2\,,
    \end{equation*}
    where in the last step we used Weyl's inequality to bound $\sigma_{k+1}(\bA)$.

    On the other hand,
    \begin{equation*}
        \vec r_i^\top \bA^* \vec r_i = \sum_j \iprod{\vec r_i, \vec v_j}^2 \ge \iprod{\vec r_i, \vec v_i}^2 = \norm{\vec r_i}^4\,,
    \end{equation*}
    so we conclude that $\norm{\vec r_i}^2 \le 2\epsilon$. If we define $\wh{\vec v}_i = \Pi \vec v_i$, where $\Pi$ is the projector to the top-$k$ singular subspace of $\bA$, then $\norm{\wh{\vec v}_i -\vec v_i}^2 = \norm{\vec r_i}^2 \le 2\epsilon$ as claimed.
\end{proof}

\noindent We now show that the empirical second moment matrix can be used to extract a rough approximation to the span of the means:

\begin{lemma}\label{lem:crude_means}
    For $\vec x\sim \mathcal M$, let $\bM \triangleq \E[\vec x \vec x^\top]$. Given $\wh{\bM}$ for which $\norm{\bM - \wh{\bM}}_{\rm op} \lesssim \beta$, let $\wh{V}$ denote the top-$k$ singular subspace of $\wh{\bM}$. Then for every $i\in[k]$, there exists $\wh{\vec \mu}_i \in \wh{V}$ for which $\norm{\wh{\vec \mu}_i - \vec \mu_i}^2 \lesssim \beta/\lambdamin$.
\end{lemma}

\begin{proof}
    Define $\calE = \sum_i \lambda_i \cov_i$ and $\bM^* \triangleq \sum_i \lambda_i \mu_i\mu_i^\top$. We have that
    \begin{equation*}
        \bM = \bM^* + \calE\,,
    \end{equation*}
    and $\norm{\calE}_{\rm op} \lesssim \beta$. 
    
    By \Cref{lem:pca}, where we take $\bA$ and $\bE$ therein to be $\wh{\bM}$ and $\calE + \wh{\bM} - \bM$, we find that $\wh{V}$ contains vectors $\vec \mu'_1, \ldots, \vec \mu'_k$ for which $\norm{\vec \mu'_i - \sqrt{\lambda}_i \vec \mu_i}^2 \lesssim \beta$. So if we take $\wh{\vec \mu}_i = \vec \mu'_i / \sqrt{\lambda_i}$, the claimed bound follows.
\end{proof}

\begin{proof}[Proof of \Cref{lem:mean_net}]
    By standard matrix concentration (see, e.g., \cite{vershynin-HDP-book})
    with $N = \poly(d\radius/\beta)$ samples (as set in \Cref{alg:mean_span}) we have that the matrix $\wh{\bM}$ constructed therein satisfies $\norm{\wh{\bM} - \bM}_{\rm op} \le \beta$, where $\bM \triangleq \E_{\mathcal{M}}[\x\x^\top] = \sum_i \lambda_i \vmu_i \vmu_i^\top + \sum_i \lambda_i \cov_i$. We have $\norm{\wh{\bM} - \sum_i \lambda_i \vmu_i \vmu_i^\top}_{\rm op} \le 2\beta$, so by \Cref{lem:crude_means}, the $\beta$-net constructed in \Cref{alg:mean_span} contains points which are $O(\beta/\lambdamin)$-close to each of the means $\vmu_i$ as claimed.
\end{proof}

\subsection{Estimating the covariances}

Next, we show how to recover a rough approximation to the span of the covariance matrices and, as a consequence, produce a net containing rough approximations to each of the covariance matrices. The algorithm is summarized in \Cref{alg:cov_span}.

\begin{lemma}\label{lem:cov_span}
    Suppose $\hatmu_1,\ldots,\hatmu_k\in\R^d$ satisfy $\norm{\vmu_i - \hatmu_i}^2 \le \meanerr$ for all $i\in[k]$. Then there is an algorithm {\sc CrudeEstimateCovariances}($q,\{\hatmu_i\}$) which returns a list $\calW$ such that for each $i\in[k]$, there exists $\wh{\cov}_i \in \calW$ for which $\norm{\cov_i - \wh{\cov}_i}_F \lesssim \beta^{1/2}\meanerr^{1/2} + k^{3/2}\meanerr + k^{5/2}\beta + k^2\alpha\log\radius$.
    Furthermore $|\calW| \le d^{O(k)}$, and the algorithm runs in time $\poly(d\radius/\beta) + d^{O(k)}$ and draws $\poly(d\radius/\beta)$ samples.
\end{lemma}

\noindent The intuition behind our approach is that if the means of the mixture were all sufficiently close to zero, then the top-$k$ singular subspace of the matrix $\E[\vc{\x\x^\top}\vc{\x\x^\top}^\top]$ can be shown to contain points close to $\vc{\cov_1},\ldots,\vc{\cov_k}$. In general, if the means are arbitrary, then we can use the estimates $\hatmu_1,\ldots,\hatmu_k$ derived in the previous section to approximately ``recenter'' the mixture components near zero. We now make this intuition precise.

\paragraph{Proof preliminaries.}
Define
\begin{equation*}
    \hatPi \triangleq \mathrm{span}(\hatmu_1,\ldots,\hatmu_k) \ \ \ \text{and} \ \ \ \hatPi^\perp \triangleq \Id - \hatPi\,.
\end{equation*}
Let $\vmu^{\parallel}_i \triangleq \hatPi \vmu_i$ and $\vmu^\perp_i \triangleq \hatPi^\perp \vmu_i$. Note that
\begin{equation*}
    \norm{\vmu^\perp_i}^2 = \norm{\hatPi^\perp(\vmu_i - \hatmu_i)}^2 \le \meanerr\,.
\end{equation*}
Also define 
 \begin{equation*}
     \resid_i \triangleq \vmu^\parallel_i - \hatmu_i
\end{equation*}
and note that
\begin{equation*}
    \norm{\resid_i}^2 = \norm{\hatPi(\vmu_i - \hatmu_i)}^2 \lesssim \meanerr\,.
\end{equation*}

Define $\covthres \ge 1$ by
\begin{equation}
    \covthres \triangleq C(\sqrt{\meanerr} + \sqrt{k\beta} + k^{1/4}\sqrt{\alpha\log\radius}) \label{eq:Qbound}
\end{equation}
for sufficiently large absolute constant $C > 0$. Given $i\in[k]$, define
\begin{equation*}
    \Sfar[i] \triangleq \{j\in[k]: \norm{\vec \mu_i - \vec \mu_j} \ge \covthres\} \ \ \ \text{and} \ \ \ \Sclose[i] \triangleq \{j\in[k]: \norm{\vec \mu_i - \vec \mu_j} \le \covthres\}\,.
\end{equation*}
The algorithm we give in this section (\Cref{alg:cov_span}) does not require knowledge of $\Sfar[i], \Sclose[i]$; these sets are only defined here for the purpose of analysis.

To approximately ``recenter'' the mixture components around zero, we will subtract from each sample the mean estimate which is closest to it in the subspace given by $\hatPi$. Formally, given $\x\sim\R^d$, define $\hatmu(\x)$ by
\begin{equation}
    \hatmu(\x) \triangleq \hatmu_i \ \ \ \text{for} \ \ \ i = \argmin_{j\in[k]} \norm{\hatmu_j - \hatPi \x}\,. \label{eq:hatmudef} 
\end{equation}
For every $i\in[k]$, define
\begin{equation*}
    \calK_i \triangleq \{ \x\in\R^d: \hatmu(\x) = \hatmu_i\}\,,
\end{equation*}
i.e. the set of points which are closest to $\hatmu_i$ in the subspace given by $\hatPi$.

Finally, given $\z\in\R^d$, define
\begin{equation*}
    \flatmat_{00}(\z) \triangleq \vc{\hatPi \z\z^\top \hatPi}\vc{\hatPi \z\z^\top \hatPi}^\top 
\end{equation*}
\begin{equation}
    \flatmat_{01}(\z) \triangleq \vc{\hatPi \z\z^\top \hatPi^\perp}\vc{\hatPi \z\z^\top \hatPi^\perp}^\top \label{eq:flatmat_def}
\end{equation}
\begin{equation*}
    \flatmat_{11}(\z) \triangleq \vc{\hatPi^\perp \z\z^\top \hatPi^\perp}\vc{\hatPi^\perp \z\z^\top \hatPi^\perp}^\top\,.
\end{equation*}
We will assemble an estimate for the span of the covariances out of the top-$k$ singular subspaces of empirical estimates of $\E_{\mathcal{M}}[\flatmat_{00}(\x - \hatmu(\x))], \E_{\mathcal{M}}[\flatmat_{01}(\x - \hatmu(\x))], \E_{\mathcal{M}}[\flatmat_{11}(\x - \hatmu(\x))]$.

For any $i\in[k]$ and $s\in\{00,01,11\}$, note that
\begin{align}
    \E_{\normal_i}[\flatmat_s(\x - \hatmu(\x))] &= \E_{\normal_i}[\flatmat_s(\x - \hatmu_i) \cdot \mathds{1}[\x\in \calK_i]] + \sum_{j\in \Sclose[i]\backslash i} \E_{\normal_i}[\flatmat_s(\x - \hatmu_{j})\cdot \mathds{1}[\x\in \calK_j]] \nonumber \\
    &\qquad\qquad\qquad\qquad\qquad\qquad\qquad\quad + \sum_{j\in \Sfar[i]} \E_{\normal_i}[\flatmat_s(\x - \hatmu_{j})\cdot \mathds{1}[\x\in \calK_j]] \nonumber \\
    &= \E_{\normal_i}[\flatmat_s(\x - \hatmu_i)] + \sum_{j\in \Sclose[i]\backslash i} \E_{\normal_i}[(\flatmat_s(\x - \hatmu_j) - \flatmat_s(\x - \hatmu_i))\cdot \mathds{1}[\x\in \calK_j]] \nonumber \\
    &\qquad\qquad\qquad\qquad\qquad + \sum_{j\in \Sfar[i]} \E_{\normal_i}[(\flatmat_s(\x - \hatmu_j) - \flatmat_s(\x - \hatmu_i))\cdot \mathds{1}[\x\in \calK_j]]\,, \label{eq:Ni_split}
\end{align}
where we used that $\calK_1,\ldots,\calK_m$ forms a partition of $\R^d$.

\paragraph{Constructing an approximation for $\sum_i \lambda_i \vc{\cov_i}\vc{\cov_i}^\top$.}

We will now argue that the two sums in Eq.~\eqref{eq:Ni_split} are negligible compared to the term $\E_{\normal_i}[\flatmat_s(\x - \hatmu_i)]$. This will allow us to construct a matrix that is close to $\sum_i \lambda_i \vc{\cov_i}\vc{\cov_i}^\top$.

In the expression $\E_{\normal_i}[\flatmat_s(\x - \hatmu(\x))]$ above, we are recentering $\x$ around $\hatmu(\x)$. We first show that the probability that a sample from the $i$-th component lands in $\calK_j$ for some $j\in\Sfar[i]$ is small, meaning that with high probability we are correctly recentering $\x$ around $\hatmu_j$ for some $j\in\Sclose[i]$.

\begin{lemma}\label{lem:clustermeans}
    For any $i\in[k]$,
    $ \Pr_{\normal_i}[\x \in \calK_j \ \text{for some} \ j\in \Sfar[i]] \le 1/\radius^8$.
\end{lemma}

\begin{proof}
    Note that $\Tr(\cov_i \hatPi) \le k\beta$ and $\norm{\cov^{1/2}_i\hatPi\cov^{1/2}_i}^2_F = \Tr(\cov_i \hatPi \cov_i \hatPi) \ge k\alpha^2$. 
    Therefore, for $\vec z\sim\calN(0,\Id)$, we may apply Hanson-Wright (\Cref{fact:HW} to control the tails of $\norm{\hatPi\cov_i^{1/2} \z}^2$). 
    \begin{fact}[Hanson-Wright]\label{fact:HW}
    Suppose $\bA\in\R^{d\times d}$ satisfies $\norm{\bA}^2_F / \norm{\bA}^2_{\rm op} \ge r$. Then for any $s > 0$,
    \begin{align}
        \Pr_{\x\sim\calN(0,\Id)}[\x^\top \bA \x - \tr(\bA) > s\norm{\bA}_F] &\le \exp(-\Omega(\min(s\sqrt{r}, s^2)) \\
        \Pr_{\x\sim\calN(0,\Id)}[\x^\top \bA \x - \tr(\bA) < -s\norm{\bA}_F] &\le \exp(-\Omega(\min(s\sqrt{r}, s^2))\,.
    \end{align}
    \end{fact}
    \noindent By taking $r$ in \Cref{fact:HW} to be $1$, we find that there is an absolute constant $C' > 0$ such that
    \begin{equation*}
        \Pr\bigl[\norm{\hatPi\cov_i^{1/2} \z}^2 > k\beta + C'\alpha\sqrt{k}\log \radius\bigr] \le 1/R^8\,.
    \end{equation*}
    
    Given $\x\sim \normal_i$, note that $\hatmu_i - \hatPi\x = \hatPi(\hatmu_i - \vmu_i) + \hatPi \cov_i^{1/2} \z$ for $\z \sim\calN(0,\Id)$. Thus, conditioned on the above event, 
    \begin{equation*}
        \norm{\hatmu_i - \hatPi \x} \le \sqrt{\meanerr} + \sqrt{k\beta} + k^{1/4}\sqrt{C'\alpha\log \radius}\,.
    \end{equation*}
    For any $j\in[k]$ and $\x\sim \normal_i$, note that $\hatmu_j - \hatPi \x = \hatPi(\hatmu_j - \hatmu_i) + \hatPi(\hatmu_i - \vmu_i) + \hatPi \cov_i^{1/2} \z$ for $\z\sim\calN(0,\Id)$. If $j\in \Sfar[i]$, we have
    \begin{equation*}
        \norm{\hatPi(\hatmu_j - \hatmu_i) + \hatPi(\hatmu_i - \vmu_i)} \ge \norm{\vmu_j - \vmu_i} - 3\sqrt{\meanerr} \ge \covthres - 3\sqrt{\meanerr}\,.
    \end{equation*}
    Thus, conditioned on the above event,
    \begin{equation*}
        \norm{\hatmu_j - \hatPi \x} \ge \covthres - 3\sqrt{\meanerr} - \sqrt{k\beta} - k^{1/4}\sqrt{C'\alpha\log \radius}\,.
    \end{equation*}
    By our choice of $\covthres$ in Eq.~\eqref{eq:Qbound}, if $C$ therein is a sufficiently large constant, the above is larger than $\norm{\hatmu_i - \hatPi \x}$ as desired.
\end{proof}

\noindent Next, we argue that the ``signal terms'' $\E_{\normal_i}[\flatmat_s(\x-\hatmu_i)]$ in Eq.~\eqref{eq:Ni_split} are well-approximated by the rank-one matrices $\vc{\hatPi\cov_i\hatPi}\vc{\hatPi\cov_i\hatPi}^\top$, $\vc{\hatPi\cov_i\hatPi^\perp}\vc{\hatPi\cov_i\hatPi^\perp}^\top$, and $\vc{\hatPi^\perp\cov_i\hatPi^\perp}\vc{\hatPi^\perp\cov_i\hatPi^\perp}^\top$.

\begin{lemma}\label{lem:signal_term}
    \begin{equation*}
        \norm{\E_{\normal_i}[\flatmat_{00}(\x - \hatmu_i)] - \vc{\hatPi(\cov_i+\resid_i\resid_i^\top)\hatPi}\vc{\hatPi(\cov_i+\resid_i\resid_i^\top)\hatPi}^\top}_{\rm op} \lesssim \beta^2 + \beta\meanerr
    \end{equation*}
    \begin{equation*}
        \norm{\E_{\normal_i}[\flatmat_{01}(\x - \hatmu_i)] - \vc{\hatPi(\cov_i + \resid_i(\vmu^\perp_i)^\top)\hatPi^\perp}\vc{\hatPi(\cov_i + \resid_i(\vmu^\perp_i)^\top)\hatPi^\perp}^\top}_{\rm op} \lesssim \beta^2 + \beta\meanerr
    \end{equation*}
    \begin{equation*}
        \norm{\E_{\normal_i}[\flatmat_{11}(\x - \hatmu_i)] - \vc{\hatPi^\perp(\cov_i + \vmu^\perp_i (\vmu^\perp_i)^\top) \hatPi^\perp}\vc{\hatPi^\perp(\cov_i + \vmu^\perp_i (\vmu^\perp_i)^\top)\hatPi^\perp}^\top}_{\rm op} \lesssim \beta^2 + \beta\meanerr\,.
    \end{equation*}
\end{lemma}

\begin{proof}
    We will be bounding the operator norm of matrices of 
    the form of $\E[\vc{\vec x \vec x^{\top}} \vc{\vec x \vec x^{\top}}]$ where $\x$ is a Gaussian vector.  
    To do so we take any test vector $\bA\in\R^{d^2}$ for which $\norm{\bA}_F = 1$; we will regard it interchangeably as a vector or as a $d\times d$ matrix. 
    We then bound $\vc{\vec A}^\top 
    \E_{\x}[\vc{\vec x \vec x^{\top}} \vc{\vec x \vec x^{\top}}] \vc{\vec A}$
    using the following simple lemma (that follows from Wicks' identity for the fourth Gaussian moments).
\begin{lemma}\label{lem:wicks_scalar}
    Let $\vec A \in \R^{d \times d}$ be any matrix 
    and $\cov$ be a covariance matrix. Then for $\vec x \sim \normal (\vec \mu, \cov)$, we have 
    \begin{align*}
        \E_{\vec x \sim \normal ( \vec \mu, \cov )}[ (\vec x^\top \vec A \vec x)^2 ] = 
        & \; \iprod{\bA, \cov}^2 + 2 \| \cov^{1/2} \bA \cov^{1/2} \|_F^2 + \| \cov^{1/2} \vec{A}^\top \vec \vmu \|^2 + \| \vec \cov^{1/2} \vec{A} \vmu \|^2 + ( \vec \mu^\top \vec A \vec \mu )^2 \\ 
        &+ 2 \vec \mu^\top \vec A \vec \mu \cdot \iprod{\bQ, \bA} + 2 \Tr(\cov^{1/2} \vec A \vec \mu \vec \mu^\top \vec{A}\cov^{1/2}).
    \end{align*}
    Moreover, if $\|\vec A\|_F \leq 1$ and 
    $\|\cov \|_{\rm op} \leq \beta$, then 
    \begin{equation}
        \Bigl|
        \E_{\vec x \sim \normal ( \vec \mu, \cov )}[ (\vec x^\top \vec A \vec x)^2 ] -
             \; \iprod{\bA, \bQ}^2 - 2(\vmu^\top \bA \vmu)\iprod{\bA, \bQ} - (\vmu^\top \bA \vmu)^2\Bigr|
             \lesssim \max( \beta^2, \beta\| \bA^\top \vmu\|^2, \beta\|\bA \vmu\|^2) \,. \label{eq:scalarwicks_opbound}
    \end{equation}
    
\end{lemma}

\begin{proof}
    Writing $\vec{x}\sim \calN(\vec{\mu}, \cov)$ as $\x = \cov^{1/2} \vg + \vmu$ for $\vg\sim\normal$, we have
    \begin{align*}
        \E_{\vec g \sim \normal}&[ (( \cov^{1/2} \vec g + \vec \mu )^\top \vec A ( \cov^{1/2} \vec g + \vec \mu ) )^2 ]  =  \\
        &\E_{\vec g \sim \normal}[ (\vec g^\top \cov^{1/2} \vec{A} \cov^{1/2} \vec g )^2 ] + \E_{\vec g \sim \normal}[ (\vec \mu^\top \vec{A} \cov^{1/2} \vec g)^2 ]  + \E_{\vec g \sim \normal}[ (\vec g^\top \vec Q^{1/2} \vec{A} \vec \mu )^2 ] + \E_{\vec g \sim \normal}[ (\vec \mu^\top \vec A \vec \mu)^2 ] \\
        &+ 2 \vec \mu^\top \vec{A} \vec \mu \E_{\vec g \sim \normal}[ (\vec g^\top \cov^{1/2} \vec A \cov^{1/2} \vec g ) ] + 2\E_{\vec g \sim \normal}[ (\vec g^\top \vec Q^{1/2} \vec{A} \vec \mu ) ( \vec \mu^\top \vec A \cov^{1/2} \vec g ) ].
    \end{align*}
    Using the definition of $\vec{B} = \vec Q^{1/2} \vec{A} \vec Q^{1/2}$, we have $$\E[ (\vec g^\top \vec{B} \vec g )^2 ] =  \sum_{i, j=1}^d \vec{B}_{i, i} \vec{B}_{j, j} \E[ \vec g_i^2 \vec g_j^2 ] + 2 \sum_{i, j=1}^d \vec{B}_{i ,j}^2 \E[ \vec g_i^2 \vec g_j^2 ] = \Tr(\vec{B})^2 + 2 \| \vec{B} \|_F^2. $$
    Using the fact that $\E[ \vec g^\top \vec{M} \vec g ] = \Tr(\vec{M})$ for any matrix $\vec{M}$, we obtain the result.
\end{proof}

    For the first claimed inequality, we apply Eq.~\eqref{eq:scalarwicks_opbound} from \Cref{lem:wicks_scalar} to $\hatPi \bA \hatPi$ and $\x'\sim \calN(\vmu_i - \hatmu_i, \cov_i)$ to get
    \begin{equation*}
        \Bigl|\E[(\x'^\top \hatPi \bA \hatPi \x')^2] - \iprod{\bA, \hatPi (\cov_i + \resid_i\resid_i^\top) \hatPi}^2\Bigr| \lesssim \beta^2 + \beta \norm{\hatPi \bA^\top \resid_i}^2 + \beta \norm{\hatPi \bA \resid_i}^2\,.
    \end{equation*}
    Note that $\norm{\hatPi \bA^\top \resid_i}^2 \le \meanerr$ and $\norm{\hatPi \bA \resid_i}^2 \le \meanerr$, so 
    \begin{equation*}
        \E[(\x'^\top \hatPi \bA \hatPi \x')^2] = \iprod{\bA, \hatPi (\cov_i+ \resid_i\resid_i^\top) \hatPi}^2 \pm O(\beta^2 + \beta\meanerr)\,.
    \end{equation*}
    Furthermore, $\bA^\top \flatmat_{00}(\x - \hatmu_i) \bA = (\x'^\top \hatPi \bA\hatPi \x')^2$ for $\x' = \x - \hatmu_i$, so because the above bound holds for all $\bA$ for which $\norm{\bA}_F = 1$, the first claimed inequality follows.
    
    The proof of the second inequality proceeds similarly. By Eq.~\eqref{eq:scalarwicks_opbound} applied to $\hatPi \bA \hatPi^\perp$ and $\x' \sim \calN(\vmu_i - \hatmu_i, \cov_i)$, we get
    \begin{equation*}
        \Bigl|\E[(\x'^\top \hatPi \bA \hatPi^\perp \x')^2] - \iprod{\bA, \hatPi (\cov_i + \resid_i (\vmu^\perp_i)^\top) \hatPi^\perp}^2\Bigr| \lesssim \beta^2 + \beta\norm{\hatPi \bA \vmu^\perp_i}^2 + \beta\norm{\hatPi^\perp \bA^\top \resid_i}^2\,.
    \end{equation*}
    Note that $\norm{\hatPi \bA \vmu^\perp_i}^2 \le \meanerr$ and $\norm{\hatPi^\perp \bA^\top \resid_i}^2 \le \meanerr$, so
    \begin{equation*}
        \E[(\x'\hatPi\bA\hatPi^\perp\x')^2] = \iprod{\bA, \hatPi (\cov_i + \resid_i(\vmu^\perp_i)^\top)\hatPi^\perp}^2 \pm O(\beta^2 + \beta\meanerr)\,.
    \end{equation*}
    Furthermore, $\bA^\top \flatmat_{01}(\x - \hatmu_i)\bA = (\x'^\top\hatPi\bA\hatPi^\perp\x')^2$ for $\x' = \x - \hatmu_i$, so because the above bound holds for all $\bA$ for which $\norm{\bA}_F = 1$, the second claimed inequality follows.

    For the third inequality, by Eq.~\eqref{eq:scalarwicks_opbound} applied to $\hatPi^\perp\bA\hatPi^\perp$ and $\x'\sim\calN(\vmu_i - \hatmu_i, \cov_i)$, we get
    \begin{equation*}
        \Bigl|\E[(\x'^\top \hatPi^\perp \bA \hatPi^\perp \x')^2] - \iprod{\bA, \hatPi^\perp(\cov_i + \vmu^\perp_i (\vmu^\perp_i)^\top)\hatPi^\perp}^2\Bigr| \lesssim \beta^2 + \beta\norm{\hatPi^\perp \bA \vmu^\perp_i}^2 + \beta\norm{\hatPi^\perp \bA^\top \vmu^\perp_i}^2\,.
    \end{equation*}
    Note that $\norm{\hatPi^\perp \bA \vmu^\perp_i}^2 \le \meanerr$ and $\norm{\hatPi^\perp \bA^\top \vmu^\perp_i}^2 \le \meanerr$,
    so
    \begin{equation*}
        \E[(\x'\hatPi^\perp\bA\hatPi^\perp\x')^2] = \iprod{\bA, \hatPi^\perp (\cov_i + \vmu^\perp_i (\vmu^\perp_i)^\top)\hatPi^\perp}^2 \pm O(\beta^2 + \beta\meanerr)\,.
    \end{equation*}
    Furthermore, $\bA^\top \flatmat_{11}(\x - \hatmu_i)\bA = (\x'^\top\hatPi^\perp \bA\hatPi^\perp\x')^2$ for $\x' = \x - \hatmu_i$, so because the above bound holds for all $\bA$ for which $\norm{\bA}_F = 1$, the third claimed inequality follows.
\end{proof}

\noindent Now if we can show that the remaining terms in Eq.~\eqref{eq:Ni_split} have small norm, then we can argue that we can read off a rough approximation of $\sum_i \lambda_i \vc{\cov_i}\vc{\cov_i}^\top$ from $\E_{\normal_i}[\flatmat_s(\x - \hatmu(\x))]$. In the following Lemma, we show the remaining terms in Eq.~\eqref{eq:Ni_split} are indeed bounded:

\begin{lemma}\label{lem:noise_term}
    Let $\hatmu,\hatmu'$ be any vectors from among $\hatmu_1,\ldots,\hatmu_k$. Suppose that either of the following holds:
    \begin{itemize}
        \item $j\in\Sfar[i]$, or
        \item $j\in \Sclose[i]$ and additionally $\hatmu, \hatmu'$ are centers of components in $\Sclose[i]$.
    \end{itemize}
    Then
    \begin{equation*}
        \norm{\E_{\normal_i}[(\flatmat_{00}(\x-\hatmu) - \flatmat_{00}(\x - \hatmu'))\cdot \mathds{1}[\x\in\calK_j]]}_{\rm op} \lesssim \beta^2 k^2 + \covthres^4
    \end{equation*}
    \begin{equation*}
        \norm{\E_{\normal_i}[(\flatmat_{01}(\x-\hatmu) - \flatmat_{01}(\x - \hatmu'))\cdot \mathds{1}[\x\in\calK_j]]}_{\rm op} \lesssim \beta^{3/2} \covthres + \beta \covthres^2
    \end{equation*}
    \begin{equation*}
        \E_{\normal_i}[(\flatmat_{11}(\x-\hatmu) - \flatmat_{11}(\x - \hatmu'))\cdot \mathds{1}[\x\in\calK_j]] = 0\,.
    \end{equation*}
\end{lemma}

\begin{proof}
    For $\x\sim \normal_i$, define $\wt{\x} \triangleq \hatPi \x - \hatmu$, $\wt{\x}' \triangleq \hatPi\x - \hatmu'$, and $\x^\perp \triangleq \hatPi^\perp \x$ so that $\x - \hatmu = \wt{\x} + \x^\perp$ and $\x - \hatmu' = \wt{\x}' + \x^\perp$.

    Let $\bA\in\R^{d^2}$ be a test vector which we regard interchangeably as a vector and as a $d\times d$ matrix, and which satisfies $\norm{\bA}_F = 1$.

    \vspace{0.3em}
    
    \noindent \underline{\em Proof for $\flatmat_{00}$}: We have
    \begin{equation*}
        \bigl|\bA^\top (\flatmat_{00}(\x - \hatmu) - \flatmat_{00}(\x - \hatmu')) \bA \cdot \mathds{1}[\x \in\calK_j]\bigr| = \bigl|(\wt{\x}^\top\bA \wt{\x})^2 - (\wt{\x}'^\top\bA \wt{\x}')^2\bigr|\cdot \mathds{1}[\x\in\calK_j]\,.
    \end{equation*}
    To bound the expectation of this over $\x\sim \normal_i$, it suffices to bound $\E_{\normal_i}[(\wt{\x}^\top \bA \wt{\x})^2 \cdot \mathds{1}[\x\in\calK_j]]$ and $\E_{\normal_i}[(\wt{\x}'^\top \bA \wt{\x}')^2 \cdot \mathds{1}[\x\in\calK_j]]$. These can be handled in the same way, so here we consider the former. 
    
    First suppose that $j\in\Sfar[i]$. By Cauchy-Schwarz,
    \begin{equation*}
        \E_{\normal_i}[(\wt{\x}^\top \bA \wt{\x})^2 \cdot \mathds{1}[\x\in\calK_j]] \le \E_{\normal_i}[(\wt{\x}^\top\bA \wt{\x})^4]^{1/2} \cdot \Pr[\x\in\calK_j]^{1/2}\,.
    \end{equation*}
    Note that
    \begin{equation*}
        \E_{\normal_i}[(\wt{\x}^\top\bA\wt{\x})^4]^{1/2} \le \E_{\normal_i}[\norm{\wt{\x}}^8]^{1/2} \lesssim \E_{\vec{h}\sim \calN(0,\hatPi\cov_i\hatPi)}[\norm{\vec{h}}^8]^{1/2} + \norm{\hatPi(\vmu_i - \hatmu)}^4 \lesssim \beta^2 k^2 + \radius^4\,.
    \end{equation*}
    The proof of the first part of the Lemma then follows by the fact that $\Pr[\x\in\calK_j]^{1/2} \le 1/R^4$ by \Cref{lem:clustermeans}, so we get an overall bound of $\beta^2 k^2 / \radius^4 + 1 \le \beta^2 k^2 + \covthres^4$ (as $\covthres, \radius \ge 1$ by assumption).

    Next, suppose that $j\in \Sclose[i]$ and additionally $\hatmu, \hatmu'$ are centers of components in $\Sclose[i]$. Then
    \begin{equation*}
        \E_{\normal_i}[(\wt{\x}^\top \bA \wt{\x})^2 \cdot \mathds{1}[\x\in\calK_j]] \le \E_{\normal_i}[(\wt{\x}^\top \bA \wt{\x})^2] \le \E_{\normal_i}[\norm{\wt{\x}}^4] \lesssim \E_{\vec{h}\sim\calN(0,\hatPi\cov_i\hatPi)}[\norm{\vec{h}}^4] + \norm{\hatPi(\vmu_i - \hatmu)}^4 \lesssim \beta^2 k^2 + \covthres^4\,,
    \end{equation*}
    thus establishing the third part of the Lemma.

    \vspace{0.3em}

    \noindent \underline{\em Proof for $\flatmat_{01}$}: We have
    \begin{equation*}
        \bA^\top (\flatmat_{01}(\x - \hatmu) - \flatmat_{01}(\x - \hatmu')) \bA \cdot \mathds{1}[\x \in\calK_j] = \bigl((\wt{\x}^\top\bA \x^\perp)^2 - (\wt{\x}'^\top\bA \x^\perp)^2\bigr)\cdot \mathds{1}[\x\in\calK_j]\,.
    \end{equation*}
    Note that the event that $\x\in \calK_j$ only depends on $\x^\perp$, so the expectation of the above over $\x\sim \normal_i$ is given by
    \begin{align}
        \MoveEqLeft \E_{\normal_i}\bigl[\bigl((\wt{\x}^\top\bA \x^\perp)^2 - (\wt{\x}'^\top\bA \x^\perp)^2\bigr)\cdot \mathds{1}[\x\in\calK_j]\bigr] \nonumber \\
        &= \E_{\wt{\x},\wt{\x}'}\bigl[\mathds{1}[\x\in\calK_j]\cdot \E_{\x^\perp}[(\wt{\x}^\top\bA\x^\perp)^2 - (\wt{\x}'^\top\bA\x^\perp)^2]\bigr] \nonumber \\
        &= \E_{\wt{\x},\wt{\x}'}\bigl[\mathds{1}[\x\in\calK_j]\cdot \iprod{\bA\hatPi^\perp \cov_i \hatPi^\perp\bA^\top, \wt{\x}\wt{\x}^\top - \wt{\x}'\wt{\x}'^\top}\bigr] \nonumber \\
        &\le \Pr[\x\in\calK_j]^{1/2} \cdot \E_{\wt{\x},\wt{\x}'}[\iprod{\bA\hatPi^\perp \bQ_i \hatPi^\perp \bA^\top, \wt{\x}\wt{\x}^\top - \wt{\x}'\wt{\x}'^\top}^2]^{1/2} \nonumber \\
        &\lesssim \Pr[\x\in\calK_j]^{1/2} \cdot \beta \E_{\wt{\x},\wt{\x'}}[\norm{\wt{\x}\wt{\x} - \wt{\x}'\wt{\x}'^\top}_F^2]^{1/2} \nonumber \\
        &= \Pr[\x\in\calK_j]^{1/2} \cdot\beta\E_{\vec{h}\sim\calN(0,\hatPi\cov_i \hatPi)}[\norm{(\vec{h} + \hatPi(\vmu_i - \hatmu))(\vec{h} + \hatPi(\vmu_i - \hatmu))^\top \nonumber \\
        &\qquad\qquad\qquad\qquad\qquad\qquad\qquad\qquad- (\vec{h} + \hatPi(\vmu_i - \hatmu'))(\vec{h} + \hatPi(\vmu_i - \hatmu'))^\top}_F^2]^{1/2} \nonumber \\
        &= \Pr[\x\in\calK_j]^{1/2} \cdot \beta \E_{\vec{h}}[\norm{\vec{h} \hatPi(\hatmu' - \hatmu)^\top + (\hatmu' - \hatmu)\hatPi\vec{h}^\top \nonumber \\
        &\qquad\qquad\qquad\qquad\qquad\qquad + \hatPi
         (\vmu_i - \hatmu)(\vmu_i - \hatmu)^\top \hatPi - \hatPi(\vmu_i - \hatmu')(\vmu_i - \hatmu')^\top \hatPi}^2_F]^{1/2} \nonumber \\
         &\lesssim \Pr[\x\in\calK_j]^{1/2} \cdot \beta \bigl(\beta^{1/2} \norm{\hatmu' - \hatmu} + \norm{\vmu_i - \hatmu}^2 + \norm{\vmu_i - \hatmu'}^2\bigr) \,,\label{eq:diffquads}
    \end{align}
    where in the second step we used that the covariance of $\x^\perp$ is $\hatPi^\perp \cov_i \hatPi^\perp$.

    Suppose that $j\in\Sfar[i]$. Then by \Cref{lem:clustermeans}, the above can be upper bounded by $\beta^{3/2}/\radius^3 + \beta/\radius^2 \le \beta^{3/2}\covthres + \beta\covthres^2$ (as $\covthres, \radius\ge 1$ by assumption), completing the proof of the second part of the Lemma.

    Next, suppose that $j\in \Sclose[i]$ and additionally $\hatmu, \hatmu'$ are centers of components in $\Sclose[i]$. Then Eq.~\eqref{eq:diffquads} can be upper bounded by $\beta^{3/2} \covthres + \beta \covthres^2$, completing the proof of the fourth part of the Lemma.

    \vspace{0.3em}

    \underline{\em Proof for $\flatmat_{11}$}: We have
    \begin{equation*}
        \bA^\top(\flatmat_{11}(\x - \hatmu) - \flatmat_{11}(\x - \hatmu'))\bA \cdot \mathds{1}[\x\in\calK_j] = \bigl(({\x^\perp}^\top \bA \x^\perp)^2 - ({\x^\perp}^\top \bA \x^\perp)^2\bigr)\cdot\mathds{1}[\x\in\calK_j] = 0\,.
    \end{equation*}
    As this holds for all $\bA$, the last part of the Lemma follows.
\end{proof}

\noindent By combining Eq.~\eqref{eq:Ni_split} with \Cref{lem:signal_term} and~\Cref{lem:noise_term}, we conclude the following:

\begin{corollary}\label{cor:Ni_approx}
    \begin{equation*}
        \norm{\E_{\normal_i}[\flatmat_{00}(\x-\hatmu(\x))] - \vc{\hatPi(\cov_i+\resid_i\resid_i^\top)\hatPi}\vc{\hatPi(\cov_i+\resid_i\resid_i^\top)\hatPi}^\top}_{\rm op} \lesssim \beta\meanerr  + \beta^2k^3 + k\covthres^4
    \end{equation*}
    \begin{equation*}
        \norm{\E_{\normal_i}[\flatmat_{01}(\x-\hatmu(\x))] - \vc{\hatPi(\cov_i + \resid_i(\vmu^\perp_i)^\top)\hatPi^\perp}\vc{\hatPi(\cov_i + \resid_i(\vmu^\perp_i)^\top)\hatPi^\perp}^\top}_{\rm op} \lesssim \beta^2 + \beta\meanerr + k\beta^{3/2}\covthres + k\beta\covthres^2
    \end{equation*}
    \begin{equation*}
        \norm{\E_{\normal_i}[\flatmat_{11}(\x-\hatmu(\x))] - \vc{\hatPi^\perp(\cov_i + \vmu^\perp_i (\vmu^\perp_i)^\top) \hatPi^\perp}\vc{\hatPi^\perp(\cov_i + \vmu^\perp_i (\vmu^\perp_i)^\top)\hatPi^\perp}^\top}_{\rm op} \lesssim \beta^2 + \beta\meanerr\,.
    \end{equation*}
\end{corollary}

\noindent Using \Cref{cor:Ni_approx} and \Cref{lem:pca}, we are now ready to state our algorithm and prove the main guarantee of this section.

\begin{algorithm2e}
\DontPrintSemicolon
\caption{\textsc{CrudeEstimateCovariances}($q, \{\hatmu_i\}$)}
\label{alg:cov_span}
    \KwIn{Sample access to $q$, estimates $\hatmu_1,\ldots,\hatmu_k$}
    \KwOut{List $\calW$ containing approximations to $\cov_1,\ldots,\cov_k$}
        $\hatPi\gets $ span of $\hatmu_1,\ldots,\hatmu_k$\;
        Define the functions $\flatmat_s$ from Eq.~\eqref{eq:flatmat_def} and $\hatmu(\cdot)$ from Eq.~\eqref{eq:hatmudef} using $\hatmu_1,\ldots,\hatmu_k$.\;
        Initialize $\calW$ to the empty set.\;
        Draw samples $\x_1,\ldots,\x_N$ from $q$ for $N \gets \poly(d\radius/\beta)$.\;
        \For{$s\in\{00,01,11\}$}{
            $\wh{\bC}_s \gets \frac{1}{N}\sum^N_{j=1} \flatmat_s(\x_j - \hatmu(\x_j))$\; \label{step:Cs}
            $\wh{V}_s\gets$ top-$k$ singular subspace of $\wh{\bC}_s$\;
            $\calW_s\gets$ a $\beta$-net over vectors in $\wh{V}_s$ with $L_2$ norm at most $\beta\sqrt{d}$\;\label{step:betanet}
        }
        \For{$\wh{\cov}^{00}\in\calW_{00}$, $\wh{\cov}^{01}\in\calW_{01}$, $\wh{\cov}^{11}\in\calW_{11}$}{
            Add $\wh{\cov}^{00} + \wh{\cov}^{01} + (\wh{\cov}^{01})^\top + \wh{\cov}^{11}$ to $\calW$.\;
        }
        \Return{$\calW$}
\end{algorithm2e} 

\begin{proof}[Proof of \Cref{lem:cov_span}]
    Consider the matrix $\bC_{00} \triangleq \E_{\mathcal{M}}[\flatmat_{00}(\x - \hatmu(\x))] = \sum_i \lambda_i \E_{\normal_i}[\flatmat_{00}(\x-\hatmu(\x))]$. By standard matrix concentration, for $N = \poly(d\radius/\beta)$ given in \Cref{alg:cov_span}, we have that the matrix $\wh{\bC}_{00}$ constructed in Step~\Cref{step:Cs} of \Cref{alg:cov_span} satisfies $\norm{\wh{\bC}_{00} - \bC_{00}}_{\rm op} \le\beta$. Therefore, by triangle inequality and \Cref{cor:Ni_approx},
    \begin{equation*}
        \norm{\wh{\bC}_{00} - \sum_i \lambda_i \vc{\hatPi(\cov_i+\resid_i\resid_i^\top)\hatPi}\vc{\hatPi(\cov_i+\resid_i\resid_i^\top)\hatPi}^\top}_{\rm op} \lesssim \beta\meanerr + \beta^2 k^3 + k\covthres^4\,.
    \end{equation*}
    By \Cref{lem:pca}, this means that the top-$k$ singular subspace of $\wh{\bC}_{00}$ contains $d^2$-dimensional vectors $\wh{\cov}^{00}_1,\ldots,\wh{\cov}^{00}_k$ which, regarded as $d\times d$ matrices, satisfy
    \begin{equation*}
        \norm{\wh{\cov}^{00}_i - \hatPi(\cov_i + \resid_i\resid_i^\top)\hatPi}^2_F \lesssim \beta\meanerr + \beta^2 k^3 + k\covthres^4
    \end{equation*}
    for all $i\in[k]$.
    
    In an entirely analogous fashion, we can show that the top-$k$ singular subspace of $\wh{\bC}_{01}$ contains $d^2$-dimensional vectors $\wh{\cov}^{01}_1,\ldots,\wh{\cov}^{01}_k$ satisfying
    \begin{equation*}
        \norm{\wh{\cov}^{01}_i - \hatPi(\cov_i + \resid_i (\vmu^\perp_i)^\top)\hatPi^\perp}^2_F \lesssim \beta^2 + \beta\meanerr + k\beta^{3/2}\covthres + k\beta\covthres^2
    \end{equation*}
    Likewise, the top-$k$ singular subspace of $\wh{\bC}_{11}$ contains $d^2$-dimensional vectors $\wh{\cov}^{11}_1, \ldots, \wh{\cov}^{11}_k$ satisfying
    \begin{equation*}
        \norm{\wh{\cov}^{11}_i - \hatPi^\perp(\cov_i + \vmu^\perp_i (\vmu^\perp_i)^\top)\hatPi^\perp}^2_F \lesssim \beta^2 + \beta\meanerr\,.
    \end{equation*}
    Finally, note that
    \begin{equation*}
        \norm{\hatPi \resid_i \resid_i^\top \hatPi}_F, \norm{\hatPi \resid_i (\vmu^\perp_i)^\top \hatPi^\perp}_F, \norm{\hatPi^\perp \vmu^\perp_i (\vmu^\perp_i)^\top \hatPi^\perp}_F \le \meanerr\,.
    \end{equation*}
    Combining all of these bounds we find that
    \begin{equation*}
        \norm{\wh{\cov}^{00}_i + \wh{\cov}^{01}_i + (\wh{\cov}^{01})^\top + \wh{\cov}^{11}_i - \cov_i}_F \lesssim \beta^{1/2}\meanerr^{1/2} + k^{3/2}(\beta + \covthres^2) \lesssim \beta^{1/2}\meanerr^{1/2} + k^{3/2}\meanerr + k^{5/2}\beta + k^2\alpha\log\radius\,.
    \end{equation*}
    The claim then follows from the fact that $\calW_{00}, \calW_{01}, \calW_{11}$ in Step~\Cref{step:betanet} contain approximations to $\wh{\cov}^{00}_i, \wh{\cov}^{01}_i, \wh{\cov}^{11}_i$ that are $\beta$-close in operator norm. Finally, note that the size of $\calW$ is bounded by $d^{O(k)}$, by standard bounds on epsilon-nets.
\end{proof}

\subsection{Putting everything together}

It is straightforward to combine the results of the previous two sections to derive the proof of \Cref{lem:crude_param}. First, for completeness, we provide the pseudocode for the algorithm:

\begin{algorithm2e}
\DontPrintSemicolon
\caption{\textsc{CrudeEstimate}($q$)}
\label{alg:estimate_params}
    \KwIn{Sample access to $q$}
    \KwOut{List $\calW$ containing approximations to $(\vmu_1,\cov_1),\ldots,(\vmu_k,\cov_k)$}
        $\calW \gets \emptyset$\;
        $\calW^{(\vmu)} \gets $\textsc{CrudeEstimateMeans}($q$)\;
        \For{$\hatmu_1,\ldots,\hatmu_k\in \calW^{(\vmu)}$}{ \label{step:loopnet1}
            $\calW^{(\cov)} \gets $\textsc{CrudeEstimateCovariances}($q, \{\hatmu_i\}$)\;
            \For{$i\in[k]$, $\wh{\cov}\in\calW^{(\cov)}$}{ \label{step:loopnet2}
                Insert $(\hatmu_i, \wh{\cov})$ into $\calW$\;
            }
        }
        \Return{$\calW$}
\end{algorithm2e} 

\begin{proof}[Proof of \Cref{lem:crude_param}]
    By \Cref{lem:mean_net}, in some iteration of Line~\Cref{step:loopnet1} of \Cref{alg:estimate_params}, we get $\hatmu_1,\ldots,\hatmu_k$ which satisfy $\norm{\hatmu_i - \vmu_i}^2 \le \meanerr$ for $\meanerr = O(\beta/\lambdamin)$. Substituting this into \Cref{lem:cov_span}, we conclude that for each $i\in[k]$, in some iteration of Line~\Cref{step:loopnet2} of \Cref{alg:estimate_params}, we get $\wh{\cov}$ satisfying $\norm{\wh{\cov} - \cov_i}_F \lesssim k^{3/2}\beta/\lambdamin + k^2\alpha\log\radius$, where we used that $\lambdamin \le 1/k$ to simplify the bound in \Cref{lem:cov_span}.

    For the bound on $|\calW|$, note that there are $(\radius/\sqrt{\beta})^{O(k^2)}$ iterations of the outer loop, within each of which there are $d^{O(k)}$ iterations of the inner loop, so $|\calW| = (\radius/\sqrt{\beta})^{O(k^2)}\cdot d^{O(k)}$ as claimed. For the runtime, \textsc{CrudeEstimateMeans} is called exactly once, and \textsc{CrudeEstimateCovariances} is called $(\radius/\sqrt{\beta})^{O(k^2)}$ times, so the overall runtime of the algorithm is $(\radius/\sqrt{\beta})^{O(k^2)}\cdot(\poly(d,1/\beta) + d^{O(k)})$.
\end{proof}

\section{Clustering via likelihood ratio estimates}
\label{sec:likelihood-clustering}

\renewcommand{\covsep}{\betweenclustercov}

In this section we present our main clustering guarantee, which leverages the estimates for the parameters we obtained from the previous section. As those estimates are only crude approximations to the true parameters, we will obtain a commensurately crude clustering. First, we formalize the notion of ``clusters'' and what it means to give an accurate clustering:

\begin{definition}
    Let $\calS = \{S_1,\ldots,S_m\}$ and $\calT = \{T_1,\ldots,T_n\}$ be partitions of $[k]$.
    
    $(\calS, \calT)$ is a {\em $(\withinclustermean, \withinclustercov, \betweenclustermean,\betweenclustercov)$-separated partition pair} if:
    \begin{itemize}
        \item For all $a\in[m]$ and $i,i' \in S_a$, we have that $\norm{\vmu_i - \vmu_{i'}} \le \withinclustermean$. 
        \item For all distinct $a,a' \in [m]$ and $i\in S_a, i' \in S_{a'}$, we have that $\norm{\vmu_i - \vmu_{i'}} \ge \betweenclustermean$.
        \item For all $b\in[n]$ and $i,i' \in T_b$, we have that $\norm{\cov_i - \cov_{i'}}_F \le \withinclustercov$. 
        \item For all distinct $b,b' \in [n]$ and $i\in T_b, i' \in T_{b'}$, we have that $\norm{\cov_i - \cov_{i'}}_F \ge \betweenclustercov$.
    \end{itemize}
\end{definition}

Roughly speaking, $\calS$ (resp. $\calT$) partitions the mixture components into groups such that any two components in the same group have means (resp. covariances) that are not far, and any two components from two different groups have means (resp. covariances) that are not close. Their common refinement is a partition $\calU$ such that any two components in the same group have both means and covariances not too far, and any two components from two different groups either have means not too close or covariances not too close.

By brute-forcing over pairs of partitions of $[k]$ (of which there are at most $k^{2k}$), we may assume we have access to $\calS$ and $\calT$, and thus to $\calU$. Our goal is then to assign to every $\x\in\R^d$ an index into the partition $\calU$. For $\x$ sampled from the $i$-th component of the mixture which belongs to the $t$-th group in $\calU$, we would like our assignment to be $t$ with high probability. The main result of this section is to show that this is indeed possible:

\begin{proposition}\label{lem:main_cluster}
    Suppose $\hatmu_1,\ldots,\hatmu_k \in \R^d$ and $\hatQ_1,\ldots,\hatQ_k\in \R^{d\times d}$ satisfy $\norm{\vmu_i - \hatmu_i}^2 \le \meanerr$ and $\norm{\cov_i - \hatQ_i}_F \le \coverr$.

    Let $(\calS = \{S_1,\ldots,S_m\}, \calT = \{T_1,\ldots,T_n\})$ denote a $(\withinclustermean, \withinclustercov,\betweenclustermean,\betweenclustercov)$-separated partition of $[k]$, where
    \begin{equation}
        \betweenclustercov \ge \max(5(\beta/\alpha)^3\coverr, c\alpha) \,, \qquad \betweenclustermean \ge \max(6\sqrt{\meanerr}, 6\sqrt{\beta k})\,, \qquad \sqrt{\meanerr} + \withinclustermean \le c\betweenclustercov\sqrt{\alpha}/\beta\,.\label{eq:sep_conds}
    \end{equation}
    for sufficiently small constant $c > 0$. Let $\{U_1,\ldots,U_{\numclust}\}$ denote the common refinement of $\calS$ and $\calT$. 

    Then there is an explicit deterministic function $\classify: \R^d\to [\numclust]$ using $\calS$, $\calT$, and $\{\hatmu_i, \wh{\cov}_i\}$, such that for any $t\in[\numclust]$ and $i\in U_t$,
    \begin{equation*}
        \Pr_{\normal_i}[\classify(\x) \neq t] \le k^3\exp\Bigl(-\Omega\Bigl(\frac{(\betweenclustermean)^2}{\alpha\sqrt{k}} \wedge \frac{\alpha^6 (\betweenclustercov)^2}{\beta^6\coverr^2} \wedge \frac{\alpha^2\betweenclustercov}{\beta^3}\Bigr)\Bigr)
    \end{equation*}
\end{proposition}

\noindent At a high level, the idea is as follows. It is not too hard to determine which group in $\calS$ a given point $\x$ should belong to, simply by checking which mean estimate $\hatmu_i$ is closest to $\x$ after projecting to the subspace spanned by $\hatmu_1,\ldots,\hatmu_k$. For each group in $\calS$, we can then effectively restrict our attention to components within that group and focus on clustering them according to their covariances. Roughly speaking, we accomplish this by comparing log-likelihoods of sampling $\x$ under $\calN(\hatmu_1,\wh{\cov}_1),\ldots,\calN(\hatmu_k,\wh{\cov}_k)$ and choosing the group in $\calT$ containing the component maximizing log-likelihood.

\subsection{Proof preliminaries}

First, we need the following basic lemma which implies that given estimates $\hatQ_1,\ldots,\hatQ_k$ for the covariances of the components, we can produce estimates $\precest_1,\ldots,\precest_k$ for the \emph{inverse} covariances:

\begin{lemma}\label{lem:invcov_from_cov}
    If $\wh{\cov}\in\R^{d\times d}$ is a psd matrix satisfying $\norm{\cov - \wh{\cov}}_F \le \coverr$, and $\alpha\,\Id \preceq \cov \preceq \beta\,\Id$, then $\norm{\cov'^{-1} - \cov^{-1}}_F \le 4\coverr/\alpha^2$ for $\cov'\in\R^{d\times d}$ defined as follows. Let $\wh{\cov}$ have singular value decomposition $\bU\bLam\bU^\top$, and define $\cov' \triangleq \bU\bLam' \bU^\top$, where $\bLam'$ is given by replacing every diagonal entry of $\bLam$ less than $\alpha/2$ with $\alpha/2$.
\end{lemma}

\begin{proof}
    Note that there are at most $4\coverr^2/\alpha^2$ diagonal entries of $\Lambda$ less than $\alpha/2$, or else we would violate the assumption that $\norm{\cov - \hatQ}_F \le \coverr$. So $\norm{\cov' - \wh{\cov}}_F \le \coverr$ and thus $\norm{\cov' - \cov}_F \le 2\coverr$. Finally, note that $\norm{\cov'^{-1}}_{\rm op} = \sigma_{\min}(\cov')^{-1}\le 2/\alpha$. We have
    \begin{align}
        \norm{\cov'^{-1} - \cov^{-1}}_F &= \norm{\cov'^{-1}(\cov' - \cov)\cov^{-1}}_F \le 4\coverr/\alpha^2\,. \qedhere
    \end{align}
\end{proof}

\noindent Given $i,j\in[k]$ and $\x, \hatmu \in \R^d$, define
\begin{equation*}
    \quadform_{ij}(\x;\hatmu) = (\x - \hatmu)^\top \precest_j (\x - \hatmu) - \iprod{\bQ_i, \precest_j}\,.
\end{equation*}
Note that for any $\vmu, \bQ$,
\begin{equation*}
    \E_{x\sim\calN(\vmu,\bQ)}[\quadform_{ii}(\x;\hatmu) - \quadform_{ij}(x;\hatmu)] = \iprod{(\vmu - \hatmu)(\vmu - \hatmu)^\top + \bQ - \bQ_i, \precest_i - \precest_j}\,.
\end{equation*}
Provided $\vmu$ and $\hatmu$ are close, if $\cov = \cov_i$ then this quantity is close to zero, but if $\cov = \cov_j$ then this quantity scales as 
\begin{equation*}
    \iprod{\cov_j - \cov_i, \precest_i - \precest_j} \approx \iprod{\cov_j - \cov_i, \cov^{-1}_i - \cov^{-1}_j} = \Tr(\cov_j \cov^{-1}_i) + \Tr(\cov_i \cov^{-1}_j) - 2d\,,
\end{equation*}
which can be quite large in comparison. Motivated by this, we will use $\quadform_{ii}(\x;\hatmu) - \quadform_{ij}(\x;\hatmu)$ to cluster the samples according to the covariances of the components generating them.

\subsection{Properties of \texorpdfstring{$\quadform_{ij}$}{Lamij}}

\begin{lemma}\label{lem:farij}
    Suppose $\covsep \ge 5(\beta/\alpha)^3\coverr$. Let $i,j\in[k]$. Suppose $\hatmu \in \R^d$ satisfies 
    \begin{equation}
        \norm{\hatmu - \vmu_j} \le c\covsep \sqrt{\alpha} / \beta \label{eq:hatmuclose}
    \end{equation}
    for some $c > 0$.

    If $\norm{\bQ_j - \bQ_i}_F \ge \covsep$, then for any $c' > 0$, with probability at least $1 - \exp(-\Omega(c'^2(\alpha^4/\beta^6) \cdot \norm{\cov_j - \cov_i}^2_F \cdot \min(1,\alpha^2/\coverr^2)))$ over $\x\sim \normal_j$,
    \begin{equation*}
        \quadform_{ii}(\x;\hatmu) - \quadform_{ij}(\x;\hatmu) \ge \iprod{\cov_j - \cov_i, \cov^{-1}_i - \cov^{-1}_j} - E \,,
    \end{equation*}
    where
    \begin{equation*}
        E \triangleq (c^2 + 2c')\norm{\cov_j - \cov_i}^2_F/\beta^2 + (4\coverr/\alpha^2)\cdot \norm{\cov_j - \cov_i}_F \,.
    \end{equation*}
\end{lemma}

\begin{proof}
    Define $\bB \triangleq \bQ^{1/2}_j (\precest_i - \precest_j)\bQ_j^{1/2}$ and $\w \triangleq \bQ^{1/2}_j(\precest_i - \precest_j)(\vmu_j - \hatmu)$. Then for $\x\sim \normal_j$, writing this as $\x = \vmu_j + \bQ^{1/2}_j \z$ for $\z\sim\calN(0,\Id)$, we see that the quantity $\quadform_{ii}(\x) - \quadform_{ij}(\x)$ is distributed as
    \begin{equation}
        \z^\top \bB \z - 2\iprod{\z, \w} + \iprod{(\vmu_j - \hatmu)(\vmu_j - \hatmu)^\top - \bQ_i, \precest_i - \precest_j}\,. \label{eq:decomp_lam}
    \end{equation}
    
    \vspace{0.3em}
    
    \noindent\underline{\em Controlling $\z^\top \bB \z$}: We would like to apply \Cref{fact:HW}. Note that
    \begin{align}
        \norm{\bB}_F &\ge \norm{\cov^{1/2}_j (\cov^{-1}_i - \cov^{-1}_j) \cov^{1/2}_j}_F - 4\beta\coverr/\alpha^2 \nonumber \\
        &= \norm{\cov^{1/2}_j{\cov}^{-1}_i ({\cov}_j - {\cov}_i) {\cov}^{-1}_j\cov^{1/2}_j}_F - 4\beta\coverr/\alpha^2 \nonumber \\
        &\ge (\alpha/\beta^2)\cdot \norm{\cov_j - \cov_i}_F - 4\beta\coverr/\alpha^2 \gtrsim (\alpha/\beta^2)\cdot\norm{\cov_j - \cov_i}_F\,, \nonumber
    \end{align}
    where in the last step we used the fact that $\covsep$ satisfies $\covsep \ge 5(\beta/\alpha)^3\coverr$ by hypothesis. Furthermore, $\norm{\bB}_{\rm op} \lesssim (\beta/\alpha)\cdot (\coverr/\alpha + 1)$, so 
    $\norm{\bB}_F / \norm{\bB}_{\rm op} \gtrsim (\alpha^2/\beta^3)\cdot\norm{\cov_j - \cov_i}_F \cdot \min(1,\alpha/\coverr)$.

    Additionally,
    \begin{align}
        \norm{\bB}_F &\le \norm{\cov^{1/2}_j\cov^{-1}_i(\cov_j - \cov_i)\cov^{-1}_j\cov^{1/2}_j}_F + \beta\coverr / \alpha^2 \nonumber \\
        &\le (\beta/\alpha^2)\cdot\norm{\cov_j - \cov_i}_F + \beta\coverr / \alpha^2 \nonumber \\
        &\lesssim (\beta/\alpha^2)\cdot\norm{\cov_j - \cov_i}_F\,, \nonumber
    \end{align}
    where in the last step we used the assumption that $\covsep \ge \coverr$.

    By \Cref{fact:HW}, for any $s > 0$, we have
    \begin{align}
        \Pr_{\z\sim\calN(0,\Id)}\Bigl[\z^\top \bB \z - &\tr(\cov_j(\precest_i - \precest_j)) \le -s(\beta/\alpha^2)\cdot\norm{\cov_j - \cov_i}_F \Bigr] \nonumber
        \\ 
        &\le \exp(-\Omega(\min(s(\alpha^2/\beta^3)\cdot\norm{\cov_j - \cov_i}_F \cdot \min(1,\alpha/\coverr), s^2)))\,. \label{eq:applyHWpre}
    \end{align}
    We will take
    \begin{equation*}
        s = c'(\alpha^2/\beta^3) \cdot \norm{\cov_j - \cov_i}_F \cdot \min(1,\alpha/\coverr)
    \end{equation*}
    for arbitrarily small constant $c' > 0$. By this choice of $s$, we have $s(\beta/\alpha^2)\cdot\norm{\cov_j - \cov_i}_F \le c'\norm{\cov_j - \cov_i}^2_dF / \beta^2$. Additionally, $s^2$ is the dominant term in the exponent in Eq.~\eqref{eq:applyHWpre}. Summarizing,
    \begin{equation}
        \Pr_{\z\sim\calN(0,\Id)}\Bigl[\z^\top \bB \z - \tr(\cov_j(\precest_i - \precest_j)) \le -c'\norm{\cov_j - \cov_i}^2_F/\beta^2 \Bigr] \le \exp(-\Omega(s^2))\,. \label{eq:applyHW}
    \end{equation}

    \noindent\underline{\em Controlling $\iprod{\z,\w}$}: Note that $\norm{\precest_i}_{\rm op}, \norm{\precest_j}_{\rm op} \lesssim 1/\alpha$, so $\norm{\w}\lesssim \covsep/\sqrt{\alpha\beta}$ by Eq.~\eqref{eq:hatmuclose}. Note that because $\covsep \gtrsim \beta \ge \alpha^{5/2} / \beta^{3/2}$, we have that $s\covsep/\sqrt{\alpha\beta} \le c'(\covsep)^2 / \beta^2 \le c'\norm{\cov_j - \cov_i}^2_F / \beta^2$. By standard Gaussian tail bounds, we conclude that
    \begin{equation}
        \Pr[|\iprod{\z,\w}| \ge c'\norm{\cov_j - \cov_i}^2 / \beta^2] \le \exp(-\Omega(s^2))\,. \label{eq:gw}
    \end{equation}

    \noindent\underline{\em Controlling $\iprod{(\vmu_j - \hatmu)(\vmu_j - \hatmu)^\top, \precest_i - \precest_j}$}: As $\norm{\precest_i}_{\rm op}, \norm{\precest_j}_{\rm op} \lesssim 1/\alpha$, by Eq.~\eqref{eq:hatmuclose} we have that
    \begin{equation}
        |\iprod{(\vmu_j - \hatmu)(\vmu_j - \hatmu)^\top, \precest_i - \precest_j}| \le c^2(\covsep)^2/\beta^2\,. \label{eq:mumuKK}
    \end{equation}

    \noindent\underline{\em Putting things together}: 
    Conditioned on the events of Eq.~\eqref{eq:applyHW} and~\eqref{eq:gw} not holding, and also using the bound on the constant term in Eq.~\eqref{eq:mumuKK}, we see from the decomposition of $\quadform_{ii}(\x;\hatmu) - \quadform_{ij}(\x;\hatmu)$ in Eq.~\eqref{eq:decomp_lam} that
    \begin{equation}
        \Pr_{\x\sim \normal_j}\Bigl[\quadform_{ii}(\x;\hatmu) - \quadform_{ij}(\x;\hatmu) - \iprod{\cov_j - \cov_i, \precest_i - \precest_j} \le -(c^2 + 2c')\norm{\cov_j - \cov_i}^2_F/\beta^2 \Bigr] \lesssim \exp(-\Omega(s^2))\,.\label{eq:applyHW2}
    \end{equation}
   
    It remains to bound $\iprod{\cov_j - \cov_i, \precest_i - \precest_j}$. We have
    \begin{equation}
        \iprod{\cov_j - \cov_i, \precest_i - \precest_j} \ge \iprod{\cov_j - \cov_i, \cov^{-1}_i - \cov^{-1}_j} - (4\coverr/ \alpha^2) \cdot \norm{\cov_j - \cov_i}_F\,. \label{eq:covcovprecestprecest}
    \end{equation}
    Combining this with Eq.~\eqref{eq:applyHW2},
    we obtain the desired bound.
\end{proof}

\begin{lemma}\label{lem:nearii}
    Let $i\in[k]$. As in \Cref{lem:farij}, suppose $\hatmu\in\R^d$ satisfies
    \begin{equation}
        \norm{\hatmu - \vmu_i} \le c\covsep\sqrt{\alpha}/\beta \label{eq:hatmuclose2}
    \end{equation}
    for sufficiently small absolute constant $c > 0$.

    For any $s \ge 1$,
    with probability at least $1 - O(k)\cdot \exp(-\Omega(s))$ over $\x\sim \normal_i$, we have that for all $j\in[k]$,
    \begin{equation*}
        \quadform_{ii}(\x;\hatmu) - \quadform_{ij}(\x;\hatmu) \le (s\beta/\alpha^2)\cdot\{\norm{\cov_j - \cov_i}_F \vee \coverr\} + c^2(\covsep)^2 / \beta^2 + c\covsep\sqrt{s/\alpha\beta} \,.
    \end{equation*}
\end{lemma}

\begin{proof}
    Define $\bB \triangleq \cov^{1/2}_i (\precest_i - \precest_j) \cov^{1/2}_i$ and $\w\triangleq \cov^{1/2}_i (\precest_i - \precest_j)(\vmu_i - \hatmu)$ (note these are slightly different from $\bB$ defined in \Cref{lem:farij} as $\x$ is sampled from $\normal_i$ instead of $\normal_j$). Then for $\x\sim \normal_i$, writing this as $\x = \vmu_i + \cov^{1/2}_i \z$ for $\z\sim\calN(0,\Id)$, we see that the quantity $\quadform_{ii}(\x;\hatmu) - \quadform_{ij}(\x;\hatmu)$ is distributed as
    \begin{equation}
        \z^\top \bB \z - 2\iprod{\z, \w} + \iprod{(\vmu_i - \hatmu)(\vmu_i - \hatmu)^\top - \bQ_i, \precest_i - \precest_j}\,. \label{eq:decomp_lam2}
    \end{equation}
    
    \noindent \underline{\em Controlling $\z^\top \bB\z$}: Note that
    \begin{align*}
        \norm{\bB}_F &\le \norm{\cov^{1/2}_i \cov^{-1}_i (\cov_j - \cov_i)\cov^{-1}_j \cov^{1/2}_i}_F + 4\beta\coverr/\alpha^2 \\
        &\lesssim (\beta / \alpha^2)\cdot\{\norm{\cov_j - \cov_i}_F \vee \coverr\}
    \end{align*}
    By \Cref{fact:HW}, we have
    \begin{equation}
        \Pr_{\z\sim\calN(0,\Id)}\Bigl[|\z^\top \bB \z - \Tr(\cov_i(\precest_i - \precest_j))| \le  (s\beta/\alpha^2)\cdot\{\norm{\cov_j - \cov_i}_F \vee \coverr\}\Bigr] \ge 1 - 2\exp(-\Omega(s))\,.\label{eq:gBg_bound}
    \end{equation}

    \noindent \underline{\em Controlling $|\iprod{\z,\w}|$}: Note that $\norm{\precest_i}_{\rm op}, \norm{\precest_j}_{\rm op} \lesssim 1/\alpha$, so $\norm{\w} \le c\covsep /\sqrt{\alpha\beta}$ by Eq.~\eqref{eq:hatmuclose2}. By standard Gaussian tail bounds, we conclude that with probability at least $1 - \exp(-\Omega(s))$,
    \begin{equation}
        |\iprod{\z,\w}| \le c\covsep\sqrt{s/\alpha\beta}\,. \label{eq:gwbound}
    \end{equation}
    
    \noindent \underline{\em Controlling $\iprod{(\vmu_i - \hatmu)(\vmu_i - \hatmu)^\top, \precest_i - \precest_j}$}: As $\norm{\precest_i}_{\rm op}, \norm{\precest_j}_{\rm op} \lesssim 1/\alpha$, by Eq.~\eqref{eq:hatmuclose2} we have that
    \begin{equation}
        |\iprod{(\vmu_i - \hatmu)(\vmu_i - \hatmu)^\top, \precest_i - \precest_j}| \le c^2(\covsep)^2/\beta^2\,. \label{eq:constanttermbound}
    \end{equation}

    \noindent \underline{\em Putting things together}: Conditioned on the events of Eq.~\eqref{eq:gBg_bound} and~\eqref{eq:gwbound} holding, and also using the bound on the constant term in Eq.~\eqref{eq:constanttermbound}, we see from the decomposition of $\quadform_{ii}(\x;\hatmu) - \quadform_{ij}(\x;\hatmu)$ that
    \begin{equation*}
        \Pr_{\x\sim \normal_i}\Bigl[|\quadform_{ii}(\x;\hatmu) - \quadform_{ij}(\x;\hatmu)| > (s\beta/\alpha^2)\cdot\{\norm{\cov_j - \cov_i}_F \vee \coverr\} + c^2(\covsep)^2 / \beta^2 + c\sqrt{s/\alpha\beta}\Bigr] \lesssim \exp(-\Omega(s))\,.
    \end{equation*}
    The claimed bound follows by a union bound.
\end{proof}

\subsection{Formally defining the clustering} 

\noindent We are now ready to define our clustering function.

\noindent Let $(\calS = \{S_1,\ldots,S_m\}, \calT = \{T_1,\ldots,T_n\})$ denote a $(\withinclustermean, \withinclustercov,\betweenclustermean,\betweenclustercov)$-separated partition of $[k]$. First, define
\begin{equation*}
    \classify^{(\vmu)}(\x) \triangleq a\in[m] \ \text{for which} \ \argmin_{i\in[k]} \norm{\hatmu_i - \hatPi \x} \in S_a\,,
\end{equation*}
where $\hatPi$ is the projector to the span of $\hatmu_1,\ldots,\hatmu_k$.

The following is a slight modification of \Cref{lem:clustermeans}:

\begin{lemma}\label{lem:meancorrect}
    Suppose that $\betweenclustermean \ge \max(6\sqrt{\meanerr}, 6\sqrt{k\beta})$. Then
    for any $i\in S_a$ and $a' \neq a$,
    \begin{equation*}
        \Pr_{\normal_i}[\classify^{(\vmu)}(\x) = a'] \le \exp\Bigl(-\Omega\Bigl(\frac{1}{\alpha\sqrt{k}}\min_{i' \in S_{a'}} \norm{\vmu_i - \vmu_{i'}}^2\Bigr)\Bigr)\,.
    \end{equation*}
    Equivalently,
    \begin{equation*}
        \Pr_{\normal_i}[\hatmu(\x) \in \{\hatmu_{i'}: i' \in S_{a'}\}] \le \exp\Bigl(-\Omega\Bigl(\frac{1}{\alpha\sqrt{k}}\min_{i' \in S_{a'}} \norm{\vmu_i - \vmu_{i'}}^2\Bigr)\Bigr)\,.
    \end{equation*}
\end{lemma}

\begin{proof}
    Note that $\Tr(\cov_i\hatPi) \le k\beta$ and $\norm{\cov^{1/2}_i \hatPi \cov^{1/2}_i}^2_F \ge k\alpha^2$, so for $\z\sim\calN(0,\Id)$, by \Cref{fact:HW} with $r$ therein taken to be $1$, for all $s > 0$ we have
    \begin{equation*}
        \Pr[\norm{\hatPi\cov^{1/2}_i\z}^2 > k\beta + s\alpha\sqrt{k}] \le \exp(-\Omega(s))\,.
    \end{equation*}

    Given $\x\sim \normal_i$, note that $\hatmu_i - \hatPi\x = \hatPi(\hatmu_i - \vmu_i) + \hatPi \cov^{1/2}_i \z$ for $\z\sim\calN(0,\Id)$. Thus, conditioned on the above event,
    \begin{equation*}
        \norm{\hatmu_i - \hatPi\x} \le \sqrt{\meanerr} + \sqrt{k\beta} + k^{1/4}\sqrt{\alpha s}\,.
    \end{equation*}
    Next, for any $i'\not\in S_a$ and $\x\sim \normal_i$, note that $\hatmu_{i'} - \hatPi\x = \hatPi(\hatmu_{i'} - \hatmu_i) + \hatPi(\hatmu_i - \vmu_i) + \hatPi\cov^{1/2}_i\z$ for $\z\sim\calN(0,\Id)$. We have
    \begin{equation*}
        \norm{\hatPi(\hatmu_{i'} - \hatmu_i) + \hatPi(\hatmu_i - \vmu_i)} \ge \norm{\vmu_{i'} - \vmu_i} - 3\sqrt{\meanerr} \ge \frac{1}{2}\norm{\vmu_{i'} - \vmu_i}\,,
    \end{equation*}
    where in the last step we used that $\betweenclustermean\ge 6\sqrt{\meanerr}$. Thus, conditioned on the above event,
    \begin{equation*}
        \norm{\hatmu_{i'} - \hatPi\x} \ge \frac{1}{2}\norm{\vmu_{i'} - \vmu_i} - \sqrt{k\beta} - k^{1/4}\sqrt{\alpha s}\,.
    \end{equation*}
    Provided that $s > (\frac{1}{2}\norm{\vmu_{i'} - \vmu_i} - \sqrt{\meanerr} - \sqrt{k\beta})^2 / \alpha\sqrt{k}$, we have that $\norm{\hatmu_i - \hatPi\x} < \norm{\hatmu_{i'} - \hatPi\x}$. As $\sqrt{\meanerr} \le \frac{1}{6}\betweenclustermean$ and $\sqrt{k\beta} \le \frac{1}{6}\betweenclustermean$, it suffices to take $s = \frac{\norm{\vmu_{i'} - \vmu_i}^2}{36\alpha\sqrt{k}}$.

    The second part of the Lemma follows by definition of $\hatmu(\x)$.
\end{proof}

Define $\classify^{(\cov)}(\x)$ as follows. First note that we can't directly use $\quadform_{ii} - \quadform_{ij}$ as it has a term $\iprod{\cov_i, \precest_i - \precest_j}$ which depends on the true covariance $\cov_i$. Likewise, the lower and upper bounds on $\quadform_{ii} - \quadform_{ij}$ in \Cref{lem:farij} and~\Cref{lem:nearii} depend on the true covariances $\cov_i, \cov_j$.

Instead, we will brute force over guesses for these quantities. Henceforth, suppose we have access to numbers $\{t_{ij}\}$ satisfying
\begin{equation*}
    \bigl|t_{ij} - (\iprod{\cov_i, \precest_i - \precest_j} + \iprod{\cov_j - \cov_i, \cov_i^{-1} - \cov_j^{-1}} - E)\bigr| \le \eta
\end{equation*}
for sufficiently small parameter $\eta$, where $E$ is the error term from \Cref{lem:farij}. Because
\begin{equation*}
    |\iprod{\cov_i, \precest_i - \precest_j} + \iprod{\cov_j - \cov_i, \cov_i^{-1} - \cov_j^{-1}} - E| \lesssim \beta d/\alpha + \coverr\beta\sqrt{d}/\alpha^2 \lesssim \beta d / \alpha\,,
\end{equation*}
we can produce these numbers by brute-forcing over a grid of size $(\beta d / \alpha\eta)^{O(k^2)}$. We will eventually take
\begin{equation}
    \eta = \frac{\covsep}{100\beta^2}\,.\label{eq:etabound}
\end{equation}

With these $\{t_{ij}\}$ in hand, given an index $\ell\in[n]$ into the partition $\{T_1,\ldots,T_n\}$, we define $\classify^{(\cov)}(\x) = b$ if there exists some $i\in T_b$ such that
\begin{equation*}
    (\x - \hatmu(\x))^\top (\precest_i - \precest_j) (\x - \hatmu(\x)) < t_{ij} - \eta
\end{equation*}
for all $j\not\in T_b$.
If there exist multiple such $b$ for which this is the case, then choose one arbitrarily. If no such $b$ exists, then set $\classify^{(\cov)}(\x)$ to be $0$.

\begin{corollary}\label{cor:covcorrect1}
    For any $i\in S_a \cap T_b$ and nonzero $b' \neq b$, we have that
    \begin{equation*}
        \Pr_{\normal_i}[\classify^{(\cov)}(\x) = b' \mid \classify^{(\vmu)}(\x) = a] \le 2k^2 \exp(-\Omega(c'^2(\alpha^4/\beta^6)\cdot \min_{j\in T_{b'}}\norm{\cov_j - \cov_i}^2_F \cdot \min(1,\alpha^2/\coverr^2)))\,.
    \end{equation*}
\end{corollary}

\begin{proof}
    We can rewrite the conditional probability as
    \begin{equation*}
        \Pr_{\normal_i}[\classify^{(\vmu)}(\x) = a]^{-1} \cdot \Pr_{\normal_i}[\classify^{(\vmu)}(\x) = a \ \text{and} \ \classify^{(\cov)}(\x) = b'] \le 2\Pr_{\normal_i}[\classify^{(\vmu)}(\x) = a \ \text{and} \ \classify^{(\cov)}(\x) = b']\,,
    \end{equation*}
    where we used \Cref{lem:meancorrect} and the fact $k\cdot \exp(-\Omega((\betweenclustermean)^2 / \alpha\sqrt{k})) \le 1/2$. Note that 
    \begin{align}
        \Pr_{\normal_i}[\classify^{(\vmu)}(\x) = a \ \text{and} \ \classify^{(\cov)}(\x) = b'] &= \sum_{i' \in S_a} \Pr_{\normal_i}[\hatmu(\x) = \hatmu_{i'} \ \text{and} \ \classify^{(\cov)}(\x) = b'] \\
        &\le \sum_{i'\in S_a} \sum_{j \in T_{b'}} \Pr_{\normal_i}[(\x - \hatmu_{i'})^\top(\precest_{j} - \precest_j)(\x - \hatmu_{i'}) < t_{jj'} - \eta  \ \forall \ j'\in[k]] \\
        &\le \sum_{i'\in S_a} \sum_{j \in T_{b'}} \Pr_{\normal_i}[(\x - \hatmu_{i'})^\top(\precest_{j} - \precest_i)(\x - \hatmu_{i'}) < t_{ji} - \eta] \\
        &\le k^2\exp(-\Omega(c'^2(\alpha^4/\beta^6)\cdot \min_{j \in T_{b'}} \norm{\cov_{j} - \cov_i}^2_F \cdot \min(1,\alpha^2/\coverr^2)))\,,
    \end{align}
    where in the last step we used \Cref{lem:farij}.
\end{proof}

\begin{corollary}\label{cor:covcorrect2}
    Suppose that
    \begin{equation}
        \covsep \ge C\max(\coverr\beta^2/\alpha^2, c^{2/3}(\covsep)^{2/3}\alpha^{1/3}, (\beta/\alpha)^3\coverr) \label{eq:covcorrect2_condition}
    \end{equation}
    for sufficiently large absolute constant $C > 0$.
    Then for any $i\in S_a \cap T_b$, we have that
    \begin{equation*}
        \Pr_{\normal_i}[\classify^{(\cov)}(\x) = 0 \mid \classify^{(\vmu)}(\x) = a] \le 2k^3\exp(-\Omega(\alpha^2\covsep/\beta^3))\,.
    \end{equation*}
\end{corollary}

\begin{proof}
    We can rewrite the conditional probability as
    \begin{equation*}
        \Pr_{\normal_i}[\classify^{(\vmu)}(\x) = a]^{-1} \cdot \Pr_{\normal_i}[\classify^{(\vmu)}(\x) = a \ \text{and} \ \classify^{(\cov)}(\x) = 0] \le 2\Pr_{\normal_i}[\classify^{(\vmu)}(\x) = a \ \text{and} \ \classify^{(\cov)}(\x) = 0]\,,
    \end{equation*}
    where we used \Cref{lem:meancorrect} and the fact $k\cdot \exp(-\Omega((\betweenclustermean)^2 / \alpha\sqrt{k})) \le 1/2$. Note that 
    \begin{align}
        \Pr_{\normal_i}[\classify^{(\vmu)}(\x) = a \ \text{and} \ \classify^{(\cov)}(\x) = 0] &= \sum_{i' \in S_a} \Pr_{\normal_i}[\hatmu(\x) = \hatmu_{i'} \ \text{and} \ \classify^{(\cov)}(\x) = 0] \\
        &\le \sum_{i'\in S_a} \sum_{j\not\in T_b} \Pr_{\normal_i}[(\x - \hatmu_{i'})^\top(\precest_{i} - \precest_{j})(\x - \hatmu_{i'}) \ge t_{ij} - \eta] \label{eq:sumprobs}
    \end{align}
    We wish to apply \Cref{lem:nearii} here. Consider any $j\not\in T_b$. Note that
    \begin{align*}
        t_{ij} - \eta - \iprod{\cov_i, \precest_i - \precest_j} &\ge \iprod{\cov_j - \cov_i, \cov^{-1}_i - \cov^{-1}_j} - 2\eta - E \\
        &\ge \tr((\cov_j - \cov_i)\cov^{-1}_i(\cov_j - \cov_i)\cov^{-1}_j) - 2\eta - E \\
        &\ge (1/\beta^2)\cdot\norm{\cov_j - \cov_i}^2_F - 2\eta - E\,.
    \intertext{
    In \Cref{lem:nearii}, take $s = (\alpha^2/\beta^3)\cdot\norm{\cov_j - \cov_i}_F$. Then we can bound the above by}
        &\ge (s\beta/\alpha^2)\cdot \{\norm{\cov_j - \cov_i}_F \vee \coverr\} + c^2(\covsep)^2/\beta^2 + c\covsep\sqrt{s/\alpha\beta}\,.
    \end{align*}
    By \Cref{lem:nearii}, this happens with probability at most $O(k) \cdot \exp(-\Omega(s))$. There are at most $k^2$ terms in the sum in Eq.~\eqref{eq:sumprobs}, so the claimed bound follows by a union bound.
\end{proof}

\noindent We can now immediately conclude the proof of the main result of this section:

\begin{proof}[Proof of \Cref{lem:main_cluster}]
     Define $\classify(\x)$ as follows. Let $a = \classify^{(\vmu)}(\x)$ and $b = \classify^{(\cov)}(\x))$. If $b = 0$, or $S_a$ and $T_b$ do not intersect, then define $\classify(\x)$ arbitrarily. Otherwise, if they do intersect, let $U_t$ denote the element of the common refinement of $\calS$ and $\calT$ corresponding to $S_a \cap T_b$, and define $\classify(\x) = t$.

     The bound on the misclassification error then follows from \Cref{lem:meancorrect}, \Cref{cor:covcorrect1}, and \Cref{cor:covcorrect2}, noting that the condition of Eq.~\eqref{eq:sep_conds} ensures that the hypotheses of these components are met.
\end{proof}

\noindent For convenience, we summarize $\classify(\x)$ in \Cref{alg:classify} below.

\begin{algorithm2e}
\DontPrintSemicolon
\caption{\textsc{Clustering}}
\label{alg:classify}
    \KwIn{Partitions $\calS = \{S_1,\ldots,S_m\}, \calT = \{T_1,\ldots,T_n\}$ of $[k]$, estimates $\{(\hatmu_i, \wh{\cov}_i)\}$, thresholds $\{t_{ij}\}$}
    \KwOut{Clustering function $\classify: \R^d\to[n_c]$}
        $\eta \gets \betweenclustercov/100\beta^2$.\;
        Let $U_1,\ldots,U_{\numclust}$ denote the common refinement of the partitions $\calS, \calT$.\;
        Let $\hatPi$ denote the projector to the span of $\hatmu_1,\ldots,\hatmu_k$.\;
        Define $c^{(\vmu)}(\x)$ to be the index $a$ of the piece $S_a$ of $\calS$ containing $\argmin_{i\in[k]} \norm{\hatmu_i - \hatPi\x}$.\;
        Define $\hatmu(\x)$ to be $\hatmu_i$ for $i = \argmin_{j\in[k]} \norm{\hatmu_j - \hatPi \x}$.\;
        Define $c^{(\cov)}(\x)$ to be the index $b$ if there exists $i\in T_b$ such that $(\x - \hatmu(\x))^\top(\precest_i - \precest_j)(\x - \hatmu(\x)) < t_{ij} - \eta$ for all $j\not\in T_b$.\;
        \eIf{$b = 0$ or $S_a\cap T_b =  \emptyset$}{
            Define $\classify(\x)$ arbitrarily.\;
        }{
            Let $U_t$ denote an element of the common refinement corresponding to $S_a \cap T_b$.\;
            \Return $\classify(\x) = t$.\;
        }
\end{algorithm2e}

\section{Score simplification}

The main difficulty in providing a polynomial approximation of the score function arises when it involves multiple Gaussians that are far apart.  
Without further structural assumptions about the function and/or the underlying measure, the degree of the polynomial
approximation depends on (1) the smoothness properties of the target function (e.g., Lipschitz constant or higher-order derivative bounds) and (2) the radius of the support over which the polynomial is guaranteed to be close to the target function.  

Recall that the score function of a mixture $\mathcal M$ of $k$ Gaussian distributions with means $\vec \mu_1,\ldots, \vec \mu_k$ and covariances
$\cov_1,\ldots, \cov_k$ is given by
\begin{equation*}
\vec s(\x; \mathcal M)
    = -\sum_{i=1}^k w_i(\x) \cov_i^{-1}( \x - \vec{\mu}_i ) \hspace{5mm} \text{where} \hspace{4mm} w_i(\x) = 
    \frac{\lambda_i \normal(\vec \mu_i, \cov_i;\x)}{\sum_{j=1}^k \lambda_j
    \normal(\vec \mu_i, \cov_i;\x)} \,.
\end{equation*}
For simplicity, in what follows we will denote by $\normal_i$ the $i$-th component
of the above mixture, $\normal_i = \normal(\vec \mu_i, \cov_i)$.
For Gaussian mixtures, the effective support of the score function is roughly proportional to the radius of the parameter space which scales with the dimension and the parameter distance $\poly(d,R)$.  This is the case as we consider a mixture over $d$-dimensional Gaussians with mean and covariances bounded (in parameter distance) by $R$.  Moreover, the Lipschitz constant of the score function can also scale as $\poly(d, R)$.  Therefore, applying black-box
polynomial approximation results (such as Jackson's theorem \--- see \Cref{lemma:multivariate-jackson}) would yield a polynomial of degree
at least polynomial in the dimension $d$ and the parameter radius $R$ 
yielding a trivial (exponential) runtime.  Instead of using the polynomial approximation results in a black-box manner, we will be constructing a piecewise
polynomial approximation of the score function where the partition 
is given by the clustering algorithm we designed in \Cref{sec:likelihood-clustering}.

In this section, we show that given the ``rough'' clustering function of 
\Cref{sec:likelihood-clustering} we can simplify the score function
inside each cell of the partition given by the clustering so that it is possible to prove the existence of a low-degree approximation inside each cell.
More precisely, we require that the clustering function $\classify(\x)$ assigns
each $\x \in \R^d$ to one of $n_c$ subsets $U_1,\ldots, U_{n_c}$ of $[k]$ 
that form a partition of the original $k$ components such that
if $\normal_i, \normal_j$ belong in different subsets $U_t$
and $U_{t'}$ have to be at least $\poly(\beta/\alpha)\cdot \log(k/\eps)$ far in parameter distance.   In other words, we require that components in different
subsets of the partition have to be sufficiently separated.  Moreover, for every $i \notin U_t$, we require that the clustering function $\classify$ incorrectly classifies a sample $\x \sim \normal_i$ as belonging to $U_t$ with probability
at most $\eps$.  Under those assumptions, we show that for any given 
$\classify(\x) = t$, we can ``simplify'' the score function by removing
the contribution of all components $\normal_j$ that do not belong in $U_t$.

In what follows, given a subset $U_t$ of indices of $[k]$ we denote by 
$\mathcal M(U_t)$ the submixture containing the components $\normal_i$ for $i \in U_t$ and by 
$\vec s(\x; \mathcal{M}(U_t))$ the score function containing only the contribution of
components from $U_t$, i.e.,
\[
\vec s(\x; \mathcal M(U_t)) = \sum_{i \in U_t} \lambda_i \vec g_i(\x) \frac{\normal_i(\x)}{\sum_{j \in U_t}\lambda_j \normal_j(\x)}
\]
The main result of this section is the following proposition showing that,
inside each cell $t$ of the partition given by $\classify(\cdot)$, we can replace the original score function $\vec s(\x; \mathcal M)$ by the score function of the 
sub-mixture $\vec s(\vec x;\mathcal M(U_t))$.
\begin{proposition}[Score Simplification]
\label{prop:score-simplification}
Fix $\eps > 0$ and let $\mathcal M$ be a mixture of $k$ Gaussian distributions $\normal_1,\ldots, \normal_k$ with mean and covariances $\vec \mu_i, \cov_i$ such that for every pair $i, j$ $\dpar(\normal_i, \normal_j) 
= \|\vec \mu_i - \vec \mu_j\|_2^2 + \|\cov_i - \cov_j \|_F^2  \leq R$
for some $R> 1$
Moreover, assume that for some $\alpha \leq 1 \leq  \beta$ it holds that 
$\alpha \Id \preceq \cov_i \preceq \beta \Id$ for all $i \in [k]$
for $\alpha \leq 1 \leq \beta$.
\begin{enumerate}
\item 
Let $\numclust \in [k]$ and let $U_1,\ldots, U_{\numclust}$ be a partition of $[k]$ such that for every $i \in U_t$, and $j\notin U_t$ it holds that $\dpar(\normal_i, \normal_j)$  is larger than a sufficiently large absolute constant multiple of $\beta^4/\alpha^2 \log( k \beta/(\alpha \eps))$.
\item 
Assume that $\classify:\R^d \mapsto [n_c]$ is a $\eps$-approximate 
clustering function, i.e., $\pr_{\x \sim \normal_i}[\classify(\x) = t] \leq \eps$ for all $t \in [\numclust]$ and $i \notin U_t$.
\end{enumerate}
Define the following piecewise approximation to the score function
\[
s(\x; \classify(\cdot)) = 
\sum_{t=1}^{\numclust} s(\x; \mathcal M(U_t)) ~ \1\{ \classify(\x) = t \} \,.
\]
It holds that
\[ 
\E_{\x \sim \mathcal M}
[\| s(\x; \mathcal M) - s (\x; \classify(\cdot)) \|_2^2 ] 
\lesssim k^{5/4}  R \frac{\beta^5}{\alpha^6}
\sqrt{\eps}  .
\]
\end{proposition}
\begin{proof}
We first observe that since $\sum_{t=1}^{\numclust} \1\{\classify(\x) = t\} = 1$ for all $\x$ (i.e.,
each point $x$ is only assigned to a single set $U_t$),
we can write $s(\x) = \sum_{t=1}^{\numclust} s(\x) \1\{\classify(\x) = t\}$ and therefore,
we have that
\[
\E_{\x \sim \mathcal M}[ \|\vec s(\x) - \vec s(\x; \classify(\cdot))\|_2^2 ]
=\sum_{t = 1}^{\numclust} 
\E_{\x \sim \mathcal M}[ \|\vec s(\x) - \vec s(\x; \mathcal M(U_t)) \|_2^2 ~ \1\{\classify(\x) =t \}] \,.
\]
We break down the total $L_2^2$ error into the case where $\x$ was actually
generated by a mixture component that belongs to the set $U_t$ (as
predicted by the clustering function $\classify(\x)$) and the case where $\x$ was
generated by some mixture component that is not in $U_t$.  Recall that we denote 
by $\mathcal M^J$ the joint density of the indexed pair $(i, \x)$ where $i$ corresponds
to the index of the mixture component that generates $\x$.  We have
\begin{align}
\E_{\x \sim \mathcal M}[&\|\vec s(\x) - \vec s(\x; \mathcal M(U_t)) \|_2^2 ~ \1\{\classify(\x) =t \}] \,
\\
&=
\E_{(i, \x) \sim \mathcal M^J}[\|s(\x; \mathcal M) -
s(\x; \mathcal M(U_t))\|^2 ~ \1\{\classify(\x) = t, i \in U_t\} ] 
\\
&+
\E_{(i, \x) \sim \mathcal M^J}[\|s(\x; \mathcal M)-
s(\x; \mathcal M(U_t))\|^2 ~ \1\{\classify(\x) = t, i \notin U_t\} ] \,.
\label{eq:score-simplification-error-decomposition}
\end{align} 
We first focus on the first part of the error, i.e., when the example $\x$ is generated by 
some component $\normal_i$ that belongs to the set $U_t$.
We have 
\begin{align*}
\E_{(i, \x) \sim \mathcal M^J}&[\|s(\x; \mathcal M)-
s(\x; \mathcal M(U_t))\|^2 ~ \1\{\classify(\x) = t, i \in U_t\}]
\leq
\sum_{i \in U_t} \lambda_i
\E_{\x \sim \normal_i}[\|s(\x; \mathcal M)-
s(\x; \mathcal M(U_t))\|^2 ]
\\
&\leq
\sum_{i\in U_t} \lambda_i \sqrt{\E_{\x \sim \mathcal \normal_i}[\|s(\x; \mathcal M)-
s(\x; \mathcal M(U_t))\|^4  ]},
\end{align*}
where the last inequality follows by Jensen's.

We show that as long as a component $\normal_j$ that we remove is far from the 
component $i \in U_t$ in parameter distance, their removal induces an exponentially small 
error in the score function.
\begin{lemma}
\label{lem:score-removal-multiple}
Let $\normal_1,\ldots,\normal_k$ be Normal distributions with 
means $\vec \mu_1,\ldots, \vec \mu_k$ and covariances $\cov_1,\ldots, \cov_k$
such that for all $i$, $\alpha \Id \leq \cov_i \leq \beta \Id$.
For any $i \in U_t$, it holds that
\[
\E_{\x \sim \normal_i}[\|\vec s(\x; \mathcal M) - \vec s(\x; \mathcal M(U_t))\|_2^4]
\lesssim 
\frac{k \beta^{10}}{\sqrt{\lambda_i} \alpha^{12}}  \sum_{j \notin U_t}
\exp\left(-c \frac{\alpha^2}{\beta^4} ~ \dpar(\normal_i, \normal_j)\right)  \,.
    \]
for some universal constant $c > 0$.  Moreover if $i \notin U_t$ it holds that
\begin{align*}
\E_{\x \sim \normal_i}&[\|\vec s(\x; \mathcal M) - \vec s(\x; \mathcal M(U_t))\|_2^4]
\\
&\lesssim 
\frac{\beta^2}{\alpha^8} 
\sum_{\ell=1, \ell \neq i}^k (\dpar(\normal_i, \normal_\ell)^2 + \dpar(\normal_i, \normal_\ell))
+ \sum_{j \notin U_t, j \neq i}
\frac{k \beta^{10}}{\sqrt{\lambda_i} \alpha^{12}}  
\exp\left(-c \frac{\alpha^2}{\beta^4} ~ \dpar(\normal_i, \normal_j)\right)  \,.
\end{align*}

\end{lemma}
Using \Cref{lem:score-removal-multiple} we obtain that 
\begin{align*}
\E_{(i, \x) \sim \mathcal M^J}&[\|\vec s(\x; \mathcal M)-
\vec s(\x; \mathcal M(U_t))\|^2  \1\{\classify(\x) = t\}  \mid i \in U_t]  
\\
&\leq 
\frac{1}{\sqrt{\sum_{i\in U_t} \lambda_i }}
\sum_{i\in U_t} \lambda_i \sqrt{\E_{\x \sim \mathcal \normal_i}[\|\vec s(\x; \mathcal M)-
\vec s(\x; \mathcal M(U_t))\|^4  ]}  
\\
&\lesssim
\frac{\sqrt{k} \beta^5}{\alpha^6} ~ 
\frac{\sum_{i\in U_t} \lambda_i^{3/4}}{\sqrt{\sum_{i\in U_t} \lambda_i } }
\max_{j \notin U_t} e^{-c \frac{\alpha^2}{\beta^4} D_p(\normal_i, \normal_j)}
\lesssim
\frac{k^{3/4} \beta^5}{\alpha^6} ~ 
\max_{j \notin U_t} e^{-c \frac{\alpha^2}{\beta^4} D_p(\normal_i, \normal_j)},
\end{align*}
where the last inequality follows from the fact that 
$\sum_{i \in U_t}\lambda_i^{3/4}
\leq {|U_t|}^{1/4} (\sum_{i \in U_t} \lambda_i)^{3/4}
\leq  {k}^{1/4} (\sum_{i \in U_t} \lambda_i)^{3/4}$.
Therefore, using this estimate we obtain that in the case where the sample is generated by some component in $U_t$, the error is 
\[
\sum_{t=1}^{\numclust} \frac{k^{3/4} \beta^5}{\alpha^6} \max_{j \notin U_t} e^{-c \frac{\alpha^2}{\beta^4} \dpar(\normal_i, \normal_j)} 
\leq \frac{k^{7/4} \beta^5}{\alpha^6} e^{- c \frac{\alpha^2}{\beta^4} \betweencluster}
\,.
\]

We next bound the error in the difference of the score functions when the clustering function
makes a mistake, i.e., $\classify(\x) = t$ but $\x$ is generated by $\normal_i$ for $i \not\in U_t$.

\begin{align*}
\E_{(i, \x) \sim \mathcal M^J}&[\|\vec s(\x; \mathcal M)-
\vec s(\x; \mathcal M(U_t))\|^2  ~ \1\{\classify(\x) = t, i \notin U_t\}]
\\
&=
\sum_{i \notin U_t} \lambda_i ~ 
\E_{\x\sim \normal_i}[\|\vec s(\x; \mathcal M)-
\vec s(\x; \mathcal M(U_t))\|^2 ~  \1\{\classify(\x) = t\}]
\\
&\leq 
\sum_{i \notin U_t} \lambda_i ~
 \sqrt{\E_{\x\sim \normal_i}[\|\vec s(\x; \mathcal M)-
\vec s(\x; \mathcal M(U_t))\|^4]} \sqrt{\pr_{\x \sim \normal_i}[\classify(\x) = t]}
\\
&\leq  \sqrt{2 \eps} ~ \sum_{i\notin U_t} \lambda_i \left(
\frac{\beta}{\alpha^4} 
\sqrt{\sum_{\ell = 1, \ell \neq i}^k D_p(\normal_i, \normal_\ell)^2 
+
\dpar(\normal_i, \normal_\ell)
}
+ 
\frac{\sqrt{k} \beta^5}{\lambda_i^{1/4} \alpha^{6}}
\sqrt{
\sum_{j \notin U_t, j \neq i} 
e^{-c \frac{\alpha^2}{\beta^4} \dpar(\normal_i, \normal_j)}
}
\right)
\\
&\lesssim \sqrt{\eps} ~ 
\left(
\frac{\beta}{\alpha^4} \max_{i\notin U_t} \sum_{\ell = 1, \ell \neq i}^k (D_p(\normal_i, \normal_\ell) + \sqrt{\dpar(\normal_i, \normal_\ell)})
+ \frac{k^{5/4} \beta^5}{\alpha^{6}}
\right),
\end{align*}
where for the third step we used the fact that by our assumption it holds
that $\pr_{\x \sim \normal_i}[\classify(\x) = t] \leq \eps$ when $i \notin U_t$
and for the last inequality we used the fact that there are at most $k$ elements that do not belong in $U_t$ and, similarly to the previous derivation, the fact
that  $\sum_{i \in U_t}\lambda_i^{3/4}
\leq {|U_t|}^{1/4} (\sum_{i \in U_t} \lambda_i)^{3/4}
\leq  {k}^{1/4} (\sum_{i \in U_t} \lambda_i)^{3/4}$.
\end{proof}

\subsection{Proof of Lemma~\ref{lem:score-removal-multiple}}
We first show the following lemma capturing the effect of removing a single component 
from the score function.  We show that the induced error is exponentially small in the distance
of the removed component $j$ and the component $i$. 
\begin{lemma}
\label{lem:score-removal-single}
Let $\normal_1,\ldots,\normal_k$ be Normal distributions with 
means $\vec \mu_1,\ldots, \vec \mu_k$ and covariances $\cov_1,\ldots, \cov_k$
such that for all $i$ $\alpha \Id \preceq \cov_i \preceq \beta \Id$
for some $\alpha \leq 1 \leq \beta$.
Let $\mathcal M$ be the mixture of $\normal_1,\ldots, \normal_k$ with weights $\lambda_1,\ldots,\lambda_k$.
Let $c>0$ be some universal constant. For all $i \neq j$, it holds that 
\[
\E_{\x \sim \normal_i}[\|\vec s(\x; \mathcal M) - \vec s^{-j}(\x)\|_2^4]
\lesssim \frac{k \beta^{10}}{\sqrt{\lambda_i} \alpha^{12}}  \exp\left(-c \frac{\alpha^2}{\beta^4} ~ \dpar(\normal_i, \normal_j)\right)  \,,
\]
where $\vec s^{-j}(\x) = \vec s(\x ; \mathcal M([k]\setminus j))$ is the score function
of the mixture after we drop the contribution of component $j$.
Moreover, it holds 
\(
\E_{\x \sim \normal_j}[\|\vec s(\x) - \vec s^{-j}(\x)\|_2^4]
\lesssim \frac{\beta^2}{\alpha^8} \sum_{\ell=1, \ell \neq j}^k 
(\dpar(\normal_j, \normal_\ell)^2 + \dpar(\normal_j, \normal_\ell)) \,.
\)
\end{lemma}
By iteratively applying \Cref{lem:score-removal-single}, and the (almost) triangle 
inequality $\|\vec a + \vec b\|_2^4 \leq 8 \|\vec a\|_2^4 + 8 \|\vec b\|_2^4$ 
we can remove all the components that do not belong in the set $U_t$ and obtain the 
error guarantee of \Cref{lem:score-removal-multiple}.

\begin{proof}[Proof \Cref{lem:score-removal-single}]
We first show the following claim bounding the gap between the original score function and the version where we drop the contribution of a component.  We remark that the following claim is 
a pointwise fact about the score function and holds for every $\x \in \R^d$.
\begin{claim}[Softmax Simplification]\label{clm:generic-component-removal}
Moreover let $D_1,\ldots, D_k$ be non-negative weight functions on $\R^d$
and $\vec g_1,\ldots, \vec g_k$ be functions $\vec g_i:\R^d \mapsto \R^d$.
Define $\vec s(\x) = \sum_{i=1}^k \vec g_i(\x) D_i(\x)/(\sum_{i=1}^k D_i(\x))$
and \[\vec s^{-j}(\x) = \sum_{i=1, i \neq j}^k \vec g_i(\x) D_i(\x) /\Bigl(\sum_{i=1, i \neq j}^k D_i(\x)\Bigr).\]
For every $i=1,\ldots, k$, it holds that
\[
\| \vec s(\x) - \vec s^{-j}(\x) \|^4_2
\leq
8  
\sum_{\ell=1, \ell \neq j}^k 
\Bigl(\frac{D_j(\x)}{A(\x)} \Bigr)  \Bigl(\frac{D_\ell(\x)}{B(\x)}\Bigr) 
\| \vec g_i(\x) - \vec g_\ell(\x) \|_2^4
+ 8 
\Bigl(\frac{D_j(\x)}{A(\x)} \Bigr) 
\| \vec g_j(\x) -\vec g_i(\x) \|_2^4 \,,
\]
where we denote by $A(\x) = \sum_{i=1}^k D_i(\x)$ and $B(\x) = \sum_{i=1, i \neq j }^k D_i(\x)$.
\end{claim}
\begin{proof}
By a direct computation, we observe that 
\begin{equation*}
\vec s(\x) - \vec s^{-j}(\x)
=
\frac{D_j(\x)}{A(\x)}  \Bigg( \vec g_j(\x) - \sum_{\ell=1, \ell \neq j }^k \vec g_\ell(\x) \frac{D_\ell(\x)}{B(\x)} \Bigg)\,.
\end{equation*}
Adding and subtracting $\vec g_i$, we obtain that the above expression is equal to
\[
\frac{D_j(\x)}{A(\x)} \Bigg( 
\vec g_j(\x) - \vec g_i(\x) + \sum_{\ell=1, \ell \neq j }^k  
\frac{D_\ell(\x)}{B(\x)}
 (\vec g_\ell(\x) - \vec g_i(\x) ) \Bigg)\,.
\]
We observe that the normalized weights $D_{\ell}(\x)/B(\x)$ form a distribution over $\ell \in [k]\setminus j$ and therefore, using Jensen's inequality, we obtain that 
\[
\Big\|\sum_{\ell=1, \ell \neq j}^k \frac{D_\ell(\x)}{B(\x)} (\vec g_i(\x) - \vec g_\ell(\x)) \Big \|_2^4
\leq 
\sum_{\ell=1, \ell \neq j}^k \frac{D_\ell(\x)}{B(\x)} \|\vec g_i(\x) - \vec g_\ell(\x))\|_2^4 \,.
\]
Combining the above we obtain the following upper bound for the $\ell_2$ error induced in the score function when we remove the contribution of the $j$-th component. We use the fact
that $\|\vec a + \vec b\|_2^4 \leq 8 \|\vec a\|_2^4 + 8 \|\vec b\|_2^4$ to obtain:
\begin{align*}
\| \vec s(\x) - \vec s^{-j}(\x) \|_2^4 
&\leq 
\Bigl(
\frac{D_j(\x)}{A(\x)} 
\Bigr)^4  \Bigl( 8 \Big\|\sum_{\ell=1, \ell \neq j}^k \frac{D_\ell(\x)}{B(\x)} (\vec g_i(\x) - \vec g_\ell(\x)) \Big \|_2^4
+ 8 \| \vec g_j(\x) -\vec g_i(\x) \|_2^4 \Bigr) 
\\
&\leq
8 
\sum_{\ell=1, \ell \neq j}^k 
\Bigl(\frac{D_j(\x)}{A(\x)} \Bigr)  \Bigl(\frac{D_\ell(\x)}{B(\x)}\Bigr) \|
\vec g_i(\x) - \vec g_\ell(\x) \|_2^4
+ 8 
\Bigl(\frac{D_j(\x)}{A(\x)} \Bigr) 
\| \vec g_j(\x) -\vec g_i(\x) \|_2^4 
\,,
\end{align*}
where for the last inequality we used the fact that $D_j(\x)/A(\x) \leq 1$ for all $\x$
and Jensen's inequality, since $D_\ell(\x)/B(\x)$ is a distribution over $\ell \neq j$
and $\|\cdot \|^4_2$ is convex.
\end{proof}
\noindent Using \Cref{clm:generic-component-removal}, with $D$ corresponding to
the component $\normal_i$ in the statement of \Cref{lem:score-removal-single}, we obtain that we have to control the terms
\begin{equation}
A^{(i, j, \ell) } = \E_{\x \sim \normal_i}\Bigl[
\Bigl(\frac{\lambda_j \normal_j(\x)}{S(\x)} \Bigr)  \Bigl(\frac{\lambda_{\ell} \normal_\ell(\x)}{S^{-j}(\x)} \Bigr) \|
\vec g_i(\x) - \vec g_\ell(\x) \|_2^4
\Bigr]
\label{eq:norm-bound-A}
\,,
\end{equation}
where $S(\x) = \sum_{s=1}^k \lambda_s \normal_s(\x)$ and $S^{-j}(\x) = S(\x)- \lambda_j \normal_j(\x)$.
Moreover, we have to control the term 
\begin{equation}
 B^{(i,j)} = \E_{\x \sim \normal_i}\Bigl[
\Bigl(\frac{\lambda_j \normal_j(\x)}{S(\x)} \Bigr) \| \vec g_j(\x) - \vec g_i(\x) \|_2^4
\Bigr]
\,.
\label{eq:norm-bound-B}
\end{equation}

Using the above notation, and \Cref{clm:generic-component-removal}, we obtain that
\begin{equation}
\E_{\x \sim \normal_i}[\|\vec s(\x) - \vec s^{-j}(\x)\|_2^2]
\leq 
8 B^{(i,j)} + 8 \sum_{\ell=1, \ell \neq j}^k A^{(i,j,\ell)} \,.
\label{eq:removal-A-B-expression}
\end{equation}

We first bound the term $B^{(i,j)}$. By Cauchy-Schwarz we have 
\begin{align}
 B^{(i,j)}
 &\leq 
 \E_{\x \sim \normal_i}\Bigl[
\Bigl(\frac{\lambda_j \normal_j(\x)}{S(\x)} \Bigr) \| \vec g_j(\x) - \vec g_i(\x) \|_2^4
\Bigr] \nonumber
\\
&\leq 
\left(
 \E_{\x \sim \normal_i}\Bigl[ \left(\frac{\lambda_j \normal_j(\x)}{S(\x)} \right)^2 \Bigr]\right)^{1/2}
\left( \E_{\x \sim \normal_i}\Bigl[
\| \vec g_j(\x) - \vec g_i(\x) \|_2^8 
\Bigr] 
\right)^{1/2} \nonumber
\\
&\leq 
\left(
 \E_{\x \sim \normal_i}\Bigl[ \frac{\lambda_j \normal_j(\x)}{S(\x)} \Bigr]\right)^{1/2}
\left( \E_{\x \sim \normal_i}\Bigl[
\| \vec g_j(\x) - \vec g_i(\x) \|_2^8 
\Bigr] 
\right)^{1/2} \nonumber
\\
&\leq 
\frac{1}{\sqrt{\lambda_i}}
\left(
 \E_{\x \sim \normal_i}\Bigl[ \frac{\normal_j(\x)}{\normal_j(\x) + \normal_i(\x)} \Bigr]\right)^{1/2}
\left( \E_{\x \sim \normal_i}\Bigl[
\| \vec g_j(\x) - \vec g_i(\x) \|_2^8 
\Bigr] 
\right)^{1/2}\label{eq:8-power-norm-bound}
\,,
\end{align}
where the third inequality follows because the ratio of weighted densities is pointwise smaller than $1$, and the last inequality follows by the fact that $\lambda_j \normal_j(\x)/(\lambda_i \normal_i(\x) + \lambda_j \normal_j(\x)) \leq \frac{1}{\lambda_i} \normal_j(\x)/(\normal_i(\x) + \normal_j(\x))$ for all $\x$.

We now need to control the following correlation between 
$\normal_j$ and $\normal_i$, \(
\E_{\x \sim \normal_j}\Bigl[\frac{\normal_i(\x)}{\normal_i(\x) + \normal_j(\x)}  \Bigr]
\).  We show that as long as the parameters of $\normal_\ell$ are far in $\ell_2$ from those 
of $\normal_j$ this correlation is exponentially small.  We prove the following claim.
\begin{claim}
\label{clm:hellinger-bound}
Let $\mathcal N(\vec \mu_1, \cov_1)$ and $\mathcal N(\vec \mu_2, \cov_2)$ be
normal distributions
with $\alpha I \leq \cov_1 \leq \beta I $, $\alpha I \leq \cov_2 \leq \beta I$ .  
For $c = 16 (1+\beta/\alpha)^2 \beta^2$, it holds that 
\[
\E_{\x \sim \normal(\vec \mu_1,\cov_1)}\Bigl[ 
\frac{\mathcal{N}(\x; \vec \mu_2, \cov_2)}
{\mathcal{N}(\x; \vec \mu_1, \cov_1) + \mathcal{N}(\x; \vec \mu_2, \cov_2)}
\Bigr]
\leq \exp\Bigl(- \frac1{\beta} \|\vec \mu_1 - \vec \mu_2\|_2^2 - \frac{1}{c} ~ \|\cov_1 - \cov_2 \|_F^2\Bigr) \,.
\]
\end{claim}
\begin{proof}
We first observe that we can bound by above the correlation between the two normals
by their Hellinger distance.  For brevity, we will denote $\normal(\vec \mu_1, \cov_1)$
as $\normal_1$ and $\normal(\vec \mu_2, \cov_2)$ as $\normal_2$.
Using the inequality $2 tz/(t+z) \leq \sqrt{t z}$ we obtain that
$\E_{\x \sim \normal_1}[\normal_2(\x)/(\normal_1(\x) + \normal_2(\x)]
\leq \frac12 (1- \vec H^2(\normal_1, \normal_2))$, where $\vec H^2$ is the squared
Hellinger distance between $\normal_1$ and  $\normal_2$.
For two normal distributions, we have that
\[
1- \vec H^2(\normal_1,\normal_2)
= \frac{|\cov_1|^{1/4} |\cov_2|^{1/4}}{|\cov_1/2 + \cov_2/2|^{1/2}}
\exp(- (1/8) \vec u^T (\cov_1/2 + \cov_2/2)^{-1} \vec u)\,,
\]
where $\vec u = \vec \mu_1 - \vec \mu_2$.
Assuming that $\lambda^1_i$ and $\lambda_i^2$ are the eigenvalues of $\normal_1,\normal_2$, 
we observe that we can write 
\[
\frac{|\cov_1|^{1/4} |\cov_2|^{1/4}}{|\cov_1/2 + \cov_2/2|^{1/2}}
= \exp\Bigl( \sum_{i=1}^d \frac{1}{4} \log\Bigl(\frac{\lambda^1_i}{\lambda^2_i}\Bigr) - \frac{1}{2} \log\Bigl(\frac12 + \frac{\lambda_i^1}{2 \lambda_i^2} \Bigr) \Bigr)\,.
\]
We can now use the following inequality showing that as long as the ratio $\lambda^1_i/\lambda^2_i$ is not very large the above difference of logarithms
behaves roughly as $(1-\lambda^1_i/\lambda^2_i)^2$.
\begin{fact}\label{fct:log-inequality}
Let $x > 0$. It holds 
$\frac14 \log x - \log(1/2 + x/2) \leq - \frac{1}{16} \frac{(1-x)^2}{(1+x)^2}$.
\end{fact}
\begin{proof}
We first use the following integral representation of the logarithm difference
\[
-\frac14 \log x + \frac12 \log(1/2 + x/2) 
=\frac 12 \int_1^x  \frac{1}{1+t} - \frac1{2 t}  d t 
=\frac 14 \int_1^x  \frac{t-1}{(1+t)t}  d t \,.
\]
We observe that if $0<x\leq 1$ we have that $(1+t) t \leq 2$ when $t\in[1,x]$.
In that case, by using the integral identity above, we obtain that $
-\frac14 \log x + \frac12 \log(1/2 + x/2) 
\leq -(1/16) (1-x)^2$.
When $x \geq 1$ we similarly obtain the upper bound
$-(1/8) (1-x)^2/((1+x)x)$.   Combining the two cases, we obtain the inequality.
\end{proof}
\noindent Using \Cref{fct:log-inequality} we obtain that 
$
\frac{|\cov_1|^{1/4} |\cov_2|^{1/4}}{|\cov_1/2 + \cov_2/2|^{1/2}}
\leq \exp\Bigl( -\frac{1}{16 C^2} \| \Id - \cov_2^{-1/2} \cov_1 \cov_2^{-1/2} \|_F^2 \Bigr),
$
where $C = 1+\max_{i=1}^d \lambda_i^1/\lambda_i^2 \leq 1+\beta/\alpha$. 
Moreover, since $\cov_2^{-1} \geq (1/\beta) \Id$ we obtain that
\begin{equation*}
\frac{|\cov_1|^{1/4} |\cov_2|^{1/4}}{|\cov_1/2 + \cov_2/2|^{1/2}}
\leq \exp\Bigl( -\frac{1}{16 C^2 \beta^2} \| \cov_1 - \cov_1 \|_F^2 \Bigr).
\end{equation*}
\end{proof}
In the following claim, we give a bound for the 
$ \E_{\x \sim \normal_1}
\Bigl[ \|\vec g_1(\x) - \vec g_2(\x) \|_2^8 \Bigr]
$ term that appears in the bound of term $B^{(i,j)}$ of 
~\Cref{eq:8-power-norm-bound}.

\begin{claim}
\label{clm:power-norm-bound}
Let $\normal_1 = \normal(\vec \mu_1, \cov_1)$,  
$\normal_2 = \normal(\vec \mu_2, \cov_2)$ and define 
$\vec g_1(\x) = \cov_1^{-1}(\x - \vec \mu_1)$, 
$\vec g_2(\x) = \cov_2^{-1}(\x - \vec \mu_2)$.
Assuming that $\alpha \Id \leq \cov_1, \cov_2\leq \beta \Id$, 
it holds 
\begin{align*}
\E_{\x \sim \normal_1}
\Bigl[ \|\vec g_1(\x) - \vec g_2(\x) \|_2^4 \Bigr] &\lesssim 
\frac{\beta^2}{\alpha^8} (
\|\cov_1 - \cov_2\|_F^2 + \|\vec \mu_1 - \vec \mu_2\|_2^2)^2 
+  \frac{1}{\alpha^2} \|\vec \mu_1 - \vec \mu_2\|_2^2 
\\
&\lesssim 
\frac{\beta^2}{\alpha^8} ( \dpar(\normal_1, \normal_2)^2
+ \dpar(\normal_1, \normal_2)) 
\,.
\end{align*}
Moreover, for $t \geq 2$ we have 
\[
\E_{\x \sim \normal_1}
\Bigl[ \|\vec g_1(\x) - \vec g_2(\x) \|_2^{2 t} \Bigr] \lesssim 
t^t \Big(\E_{\x \sim \normal_1}
\Bigl[ \|\vec g_1(\x) - \vec g_2(\x) \|_2^4 \Bigr] \Big)^{t/2} \,.
\]
\end{claim}
\begin{proof}
We first observe that 
\begin{equation*}
\E_{\x \sim \normal_1}
\Bigl[ \|\vec g_1(\x) - \vec g_2(\x) \|_2^4 \Bigr] 
= 
\E_{\x \sim \normal}
\Bigl[ \| (\cov_1^{-1/2} - \cov_2^{-1} \cov_1^{1/2}) ~ \x + \cov_2^{-1}(\vec \mu_2 - \vec \mu_1 ) \|_2^4 \Bigr] 
= 
\E_{\x \sim \normal(\vec b, \vec A)}
\Bigl[ \| \x \|_2^4 \Bigr] 
\,,
\end{equation*}
where $\vec b = \vec Q_2^{-1} (\vec \mu_2 - \vec \mu_1)$ and $\vec A = \vec S \vec S^T$ with
$\vec S =  \vec Q_1^{-1/2} - \vec Q_2^{-1} \vec Q_1^{1/2}$.
By \Cref{lem:wicks_scalar} we have that 
\begin{align*}
\E_{\x \sim \normal(\vec b, \vec A)}
\Bigl[ \| \x \|_2^4 \Bigr] 
&= \tr(\vec A)^2 + 2 \|\vec A\|_F^2 + 2 \|\vec A^{1/2} \vec b\|^2 + \|\vec b\|_2^2 (1 + 2 \tr(\vec A)) + 2 \vec b^T \vec A \vec b + \|\vec b\|_2^4
\\
&\lesssim \|\vec S\|_F^4 + \|\vec b\|_2^2 (1 + \|\vec S\|_F^2) + \|\vec b\|_2^4 
\lesssim (\|\vec S\|_F^2 + \|\vec b\|_2^2)^2 + \|\vec b\|_2^2
\,.
\end{align*}
We observe that $\| \vec S\|_F = \| \cov_1^{-1}( \cov_2 - \cov_1) \cov_2^{-1} \cov_1^{1/2}\|_F
\leq  \frac{\sqrt{\beta}}{\alpha^2} ~ \|\cov_1 - \cov_2\|_F $, where the inequality follows by the fact that $\|\vec A \vec B\|_F\leq \|\vec A\|_2 \|\vec B\|_F$ and the spectral bounds on $\cov_1,\cov_2$. 
Moreover, $\|\vec b\|_2 \leq (1/\alpha) ~ \|\vec \mu_1 - \vec \mu_2\|_2$, since $\|\cov_2^{-1}\|_2 \leq 1/\alpha$.
Therefore, we obtain that
\begin{align*}
\E_{\x \sim \normal_1}
\Bigl[ \|\vec g_1(\x) - \vec g_2(\x) \|_2^4 \Bigr] 
&\lesssim 
\frac{\beta^2}{\alpha^8} 
(\|\cov_1 - \cov_2\|_F^2 + \|\vec \mu_1 - \vec \mu_2\|_2^2 )^2 + 
\frac{1}{\alpha^2} \|\vec \mu_1 - \vec \mu_2\|_2^2 \,.
\end{align*}

To obtain the second bound of the claim, we will use the standard hypercontractivity inequality for polynomials (\Cref{fct:hypercontractivity}).

\begin{fact}[Gaussian hypercontractivity]
\label{fct:hypercontractivity}
Let $p:\R^d \mapsto \R$ be a polynomial of degree at most $\ell$ and let $t \geq 2$. It holds 
\(
\Bigl(\E_{\x \sim \normal}[p^t(\x)]
\Bigr)^{1/t} 
\leq 
(t-1)^{\ell/2}
~ 
\Bigl(\E_{\x \sim \normal}[p^2(\x)]
\Bigr)^{1/2}  \,.
\)
\end{fact}

We have that $p(\x) = \|\vec g_1(\x) - \vec g_2(\x)\|_2^2$ is a degree $2$ polynomial
and therefore the claimed bound follows from the previous bound on 
$\|\vec g_1(\x) - \vec g_2(\x)\|_2^4 = |p(\x)|^2$ and the hypercontractivity inequality of \Cref{fct:hypercontractivity}.
\end{proof}

We can now apply \Cref{clm:hellinger-bound} and \Cref{clm:power-norm-bound} to the bound of \Cref{eq:8-power-norm-bound} and obtain the following bound 
for some universal constant $c>0$:
\begin{align*}
B^{(i,j)}
&\lesssim 
\frac{\beta^2}{\sqrt{\lambda_i} \alpha^8} ~ ((\|\cov_i - \cov_j\|_F^2 + \|\vec \mu_i - \vec \mu_j\|_2^2)^2 + \|\vec \mu_i - \vec \mu_j\|_2^2) ~ 
e^{-c \frac{\alpha^2}{\beta^4} ~ (\|\cov_i - \cov_j\|_F^2 + \|\vec \mu_i - \vec \mu_j\|_2^2) } 
\\
&\lesssim
\frac{\beta^2}{\sqrt{\lambda_i} \alpha^8} (\dpar(\normal_i, \normal_j)^2 + \dpar(\normal_i, \normal_j)) ~ 
e^{-c \frac{\alpha^2}{\beta^4} \dpar(\normal_i, \normal_j) }  \\
&\lesssim
\frac{\beta^{10}}{\sqrt{\lambda_i} \alpha^{12} }  ~ 
e^{-(c/4) \frac{\alpha^2}{\beta^4} \dpar(\normal_i, \normal_j) }\,,
\end{align*}
where for the last inequality, we used the fact that for all $t \geq 0$, it holds that
$t^2 e^{-t} \leq e^{-t/4}$ and 
$t e^{-t} \leq e^{-t/2}$.

We now bound the cross-error term $A^{(i,j,\ell)}$ of Equation \Cref{eq:norm-bound-A}.
We first observe that $A^{(i,j,\ell)}$ (in contrast with term $B^{(i,j)}$ that we bounded previously) does not vanish when $i = j$.  We first focus on the case where $i \neq j$.
Using the Cauchy-Schwarz inequality we obtain
\begin{align*}
A^{(i, j, \ell) } 
&= \E_{\x \sim \normal_i}\Bigl[
\Bigl(\frac{\lambda_j \normal_j(\x)}{S(\x)} \Bigr)  \Bigl(\frac{\lambda_{\ell} \normal_\ell(\x)}{S^{-j}(\x)} \Bigr) \|
\vec g_i(\x) - \vec g_\ell(\x) \|_2^4
\Bigr]\\
&\leq 
\left(\E_{\x \sim \normal_i}\Bigl[
\Bigl(\frac{\lambda_j \normal_j(\x)}{S(\x)} \Bigr)^{4} \Bigr]\right)^{1/4}  
~ \left(\E_{\x \sim \normal_i}
\Bigl[
\Bigl(\frac{\lambda_{\ell} \normal_\ell(\x)}{S^{-j}(\x)} \Bigr)^4\Bigr]\right)^{1/4} 
~ \left(\E_{\x \sim \normal_i}
\Bigl[ \|\vec g_i(\x) - \vec g_\ell(\x)\|_2^8 \Bigr]\right)^{1/2}  \\
&\leq 
\frac{1}{\sqrt{\lambda_i}}
\left(\E_{\x \sim \normal_i}\Bigl[
\frac{\normal_j(\x)}{\normal_i(\x) + \normal_j(\x)} \Bigr]\right)^{1/4}  
~ \left(\E_{\x \sim \normal_i}
\Bigl[ \frac{\normal_\ell(\x)}{\normal_i(\x) + \normal_\ell(\x)} \Bigr]\right)^{1/4} 
~ \left(\E_{\x \sim \normal_i}
\Bigl[ \|\vec g_i(\x) - \vec g_\ell(\x)\|_2^8 \Bigr]\right)^{1/2} \,,
\end{align*}
where the third inequality follows because the ratio of weighted densities is pointwise smaller than $1$. We remark that the last inequality holds true because in the case where $i \neq j$ it holds that $S^{-j}(\x) \geq \lambda_i \normal_i(\x) + \lambda_\ell\normal_\ell(\x)$.  We can now use \Cref{clm:hellinger-bound}
and \Cref{clm:power-norm-bound} to bound each of the three terms of the above expression
for $A^{(i,j,\ell)}$ separately:
\begin{equation*}
A^{(i, j, \ell)} 
\lesssim \frac{\beta^2}{\alpha^8 \sqrt{\lambda_i}} 
e^{-c' \frac{\alpha^2}{\beta^4}(\dpar(\normal_i, \normal_j) + \dpar(\normal_i, \normal_{\ell})) }
(\dpar(\normal_i, \normal_\ell)^2 + \dpar(\normal_i, \normal_\ell) )
\lesssim \frac{\beta^{10}}{\alpha^{12} \sqrt{\lambda_i}} 
e^{-c' \frac{\alpha^2}{\beta^4} \dpar(\normal_i, \normal_j)} \,,
\end{equation*}
where $c'$ is some universal constant and for the last inequality we used the fact that for all $t$
where for the last inequality, we used the fact that for all $t \geq 0$, it holds that
$t^2 e^{-t} \leq e^{-t/4}$ and 
$t e^{-t} \leq e^{-t/2}$.

Putting together the bounds for $A^{(i,j,\ell)}$ and $B^{(i,j)}$ we obtain that
\begin{equation*}
\E_{\x \sim \normal_i}[\|\vec s(\x) - \vec s^{-j}(\x)\|_2^4]
\lesssim \sum_{\ell = 1, \ell \neq j}^k A^{(i,j,\ell)} + B^{(i,j)}
\lesssim \frac{k}{\sqrt{\lambda_i}} \frac{\beta^{10}}{\alpha^{12}} \exp\left(-c \frac{\alpha^2}{\beta^4} ~ \dpar(\normal_i, \normal_j)\right)  \,.
\end{equation*}

We now work out the case where $i = j$ (see the second estimate in \Cref{lem:score-removal-single}).  Using \Cref{clm:generic-component-removal}, for $i=j$, we obtain the following estimate 
\begin{equation*}
\E_{\x \sim \normal_j}[\|\vec s(\x) - \vec s^{-j}(\x)\|_2^4]
\leq 8 \sum_{\ell = 1, \ell \neq j}^k A^{(j,j,\ell)} \,.
\end{equation*}
In this case, we cannot guarantee that the weight
terms $\lambda_j\normal_j(\x)/S(\x)$ and $\lambda_\ell \normal_\ell(\x)/S^{-j}(\x)$ will be exponentially small and therefore we simply use the fact that they are at most 1:
\begin{align*}
A^{(j, j, \ell) } 
&= 
\E_{\x \sim \normal_j}\Bigl[
\Bigl(\frac{\lambda_j \normal_j(\x)}{S(\x)} \Bigr)  \Bigl(\frac{\lambda_{\ell} \normal_\ell(\x)}{S^{-j}(\x)} \Bigr) \|
\vec g_j(\x) - \vec g_\ell(\x) \|_2^4
\Bigr]
\leq 
\E_{\x \sim \normal_j}\Bigl[\|\vec g_j(\x) - \vec g_\ell(\x) \|_2^4 \Bigr]
\\
&\lesssim \frac{\beta^2}{\alpha^8} (\dpar(\normal_j, \normal_\ell)^2 
+ \dpar(\normal_j, \normal_\ell))
\,,
\end{align*}
where for the last inequality we used \Cref{clm:power-norm-bound}.
Substituting the estimate for $A^{(i,j,\ell)}$ yields the claimed bound. 
\end{proof}

\section{Existence and learning of a piecewise polynomial}
\label{sec:learning-piecewise-polynomial}

\subsection{Existence of a piecewise polynomial}

In this section, we will show the existence of a piecewise polynomial approximation for the score function. To show the desired polynomial existence result, we start by showing the polynomial existence result for the score function of each subset $U_i$ and combine the results with the clustering guarantee (\Cref{lem:main_cluster}) and the score simplification guarantee (\Cref{prop:score-simplification}) to obtain the result for the complete mixture. 

\subsubsection{Polynomial approximation of a sub-mixture with small parameter distance}

We will first obtain the result for a mixture $\mc M(U)$ where the mixture has $|U| = m \leq k$ components and the parameter distance between any two components $\| \vmu_i - \vmu_j \| + \| \cov_i - \cov_j \| \leq \withincluster$ for all $i, j \in [m]$. Our main result of this section is the following proposition. 

\begin{proposition}
\label{lem:cluster-poly-approx}
    Let $\mc M(U)$ be a mixture of $\nc$ well-conditioned Gaussians with $\alpha \Id \preceq \cov_i \preceq \beta \Id$ and parameters satisfying $\| \vmu_i - \vmu_j \| + \| \cov_i - \cov_j \|_F \leq \withincluster$ for all $i, j \in [\nc]$. Let $\{ \hatmu_i, \hatQ_i, \precest_i \}_{i=1}^\nc$ be the estimates of the parameters $\{ \vmu_i, \cov_i, \cov_i^{-1} \}_{i=1}^\nc$ within parameter distance $\| \hatmu_i - \vmu_i \| + \| \hatQ_i - \cov_i \|_F + \| \precest_i - \cov_i^{-1} \|_F \leq \paramerr$ and with the operator norm satisfying $\| \precest_i \|_\op \lesssim \frac{1}{\alpha}$ for all $i \in U$. Then, there exists a polynomial $p(\x; \mathcal M(U))$ of degree $\wt{O}( \frac{  \beta^2 \nc^2 \paramerr^5 \withincluster^6 }{ \alpha^6 \eps} )$ and coefficients bounded in magnitude by $d R \exp( \wt{O} ( \frac{  \beta^2 \nc^2 \paramerr^5 \withincluster^6 }{ \alpha^6 \eps} ) )$ such that for all $\x$, the following holds
    \begin{equation*}
        \E_{\x \sim \mc M(U)}[\| s(\x; \mathcal M(U)) - \lscore(\x; \mathcal M(U))  \|^2] \leq \eps \ ,
    \end{equation*}
    where the approximating function is 
$\lscore(\x; \mc M(U) ) \triangleq p(\x; \mc M (U) ) \1 \{ \hballt(\x; U ) \} + \precest_i (\x - \hatmu_i) \1 \{ \hballt^c(\x; U) \}$ for some $i \in U$ where $\hballt(\x; U)$ denotes the region $\hballt(x; \htheta_1, \htheta_2)$ of the polynomial approximation for cluster $U$. 
    where $\widehat{B}(\x) : \R^d \to \{0, 1\}$ function that only depends on the estimates $\{ \hatmu_i, \hatQ_i, \precest_i \}_{i=1}^\nc$.
\end{proposition}

Observe that the score function for the mixture can be written as a product between linear functions (i.e., $\cov_i^{-1}(\x - \vmu_i)$) and the softmax function. We define the softmax function $w: \R^\nc \mapsto [0,1]^\nc$ as follows:
\begin{equation}
\label{eq:softmax-definition}
    w_i(\vec y; \vec \theta) = \frac{ e^{\vec y_i + \vec \theta_i} }{ \sum_{j=1}^\nc e^{\vec y_j + \vec \theta_j} }
\end{equation}
for some fixed parameters $\{ \vec \theta_i \}_{i=1}^\nc$. We start by showing that in this special case, the score can be pointwise approximated by a low-degree polynomial over a bounded domain (\Cref{prop:poly-approx-softmax} below). 

For this, we will need the following classical polynomial approximation result for functions with bounded gradients:

\begin{lemma}[Multivariate Jackson's Approximation, \cite{Newman1964, diakonikolas2010bounded}]\label{lemma:multivariate-jackson} 
    For $F: \R^n \to \R$, define the modulus of continuity
    \begin{equation*}
        \omega(F, \delta) = \sup_{ \substack{\norm{\x}_2, \norm{\vec y}_2 \leq 1 \\ 
        \norm{\x - \vec y} \leq \delta } } | F(\x) - F(\vec y) |.
    \end{equation*}
    For any $\ell \geq 1$, there exists a polynomial $p_\ell$ of degree $\ell$ such that
    \begin{equation*}
        \sup_{\norm{\x}_2 \leq 1} | F(\x) - p_\ell(\x) | \lesssim \omega (F, n/\ell)\,. 
    \end{equation*}
    
\end{lemma}

To prove an upper bound on the coefficients of the polynomial, we will use the following result.

\begin{lemma}[Coefficients of bounded polynomials, \cite{ben2018classical}]\label{lem:coeff-bounds-polynomials}  
Let $p$ be a polynomial with real coefficients on $d$ variables with degree $\ell$ such that for all $\x \in [0, L]^d, |p(\x)| \leq R$. Then, the sum of the magnitude of all coefficients of $p$ is at most $R (2L (d + \ell))^{3 \ell}$ for any $L \geq 1$.
\end{lemma}

\noindent We now show the polynomial approximation result for the softmax function and, as a consequence, for the product of a linear function with the softmax function:

\begin{lemma}
[Polynomial Approximation]
\label{prop:poly-approx-softmax}
Let $\mathcal X$ be a subset of $\R^d$ and $w_i(\vec y; \vec \theta)$ be the softmax function defined in \eqref{eq:softmax-definition}. Let $\vec G(\x) = (\vec g_1(\x), \ldots, \vec g_\nc(\x)) :\R^d \mapsto \R^{d \times \nc}$ be 
such that $\|\vec g_i(\x)\|_2 \leq M$ for all $\x \in \mathcal X$ and 
$\vec g_i(\x)$ is linear in $\x$.
Let $\vec r:\R^d\mapsto \R^\nc$ with $\vec r = (\vec r_1(\x),\ldots, \vec r_\nc(\x))$
be such that $|\vec r_i(\x)| \leq L$ for all $\x \in \mathcal X$. 
There exists a polynomial
transformation $\vec q:\R^\nc \mapsto \R^\nc$ of degree at most 
$O(L M \nc^2/\eps)$ 
such that for all $\x \in \mathcal X$ it holds that 
$\| \vec G(\x) w(r(\x); \vec \theta) - \vec G(\x) \vec q(r(\x)) \|_2 \leq \eps $.
The sum of the magnitudes of the coefficients of $\vec q$ is at most $\nc \exp( \wt O ( L M \nc^2/\eps ) )$.
\end{lemma}

\begin{proof}
    The gradient of the softmax function is given by
    \begin{equation*}
        \frac{\partial w_i ( \vec y; \vec \theta )}{ \partial \vec y_j } = \begin{cases}
            w_i( \vec y; \vec \theta) ( 1 - w_i(\vec y; \vec \theta) ) & \text{if } \;\; i = j  \\
            - w_i( \vec y; \vec \theta) w_j(\vec y; \vec \theta) & \text{otherwise}.
        \end{cases}
    \end{equation*}
    We conclude that $\| \nabla w_i( \vec y; \vec \theta) \| \leq \sqrt{\nc}$ for all $i \in [\nc]$ 
    and any $\vec y \in \R^\nc$. Using multivariate Jackson's theorem (\Cref{lemma:multivariate-jackson}) for $w_i( \vec y; \vec \theta )$, we obtain that there exists a polynomial $q(\vec y)$ of degree $\ell$ such that 
    \begin{equation*}
        \sup_{\| \vec y \| \leq L\nc} | w_i(\vec y; \vec \theta) - q (\vec y)  | \lesssim \frac{L\nc^{\frac{3}{2}}}{\ell} \,. 
    \end{equation*}
    This implies that we have a set of polynomials $\{ q_i(\vec y) \}_{i=1}^\nc$ of degree $O(\frac{L \nc^{3/2}}{\eps})$ 
    such that for all $\vec y$ in $L_2-$ball of radius $\| \vec y \| \leq L \nc$, we have $\| w(\vec y; \vec \theta) - \vec q( \vec y ) \| \leq \eps.$ 
    Additionally, $\| \vec g_i( \x ) \|_2 \leq M$ implies that $\| \vec G(\x) \| \leq M \sqrt{\nc}$. 
    Therefore, we have
    \begin{equation*}
        \| \vec G(\x) \vec w( r(\x) ) - \vec G(\x) \vec q( r(\x) ) \|_2 \leq \| \vec G(\x) \| \| \vec w( r(\x) ) - \vec q( r(\x) ) \|_2 \leq M \sqrt{\nc} \eps.
    \end{equation*}
    We obtain the result by rescaling $\eps$. To obtain the bounds on the sum of the magnitude of coefficients, we use the fact that $| \vec q_i(\vec y) | \leq 2$ for all $\| \vec y \| \leq L \nc$. Therefore, using \Cref{lem:coeff-bounds-polynomials}, we obtain that the bounds on the sum of the magnitude of coefficients is at most $O((2 L \nc (\nc + \frac{L M \nc^2}{\eps})^{ \frac{L M \nc^2}{ \eps } })) = \exp( \wt O ( L M \nc^2/\eps ) )$. 
\end{proof}

\begin{lemma}
\label{lem:concentration-poly-approx-input}
    Let $\normal (\vmu_1, \cov_1)$ be a Gaussian distribution with $\alpha \Id \preceq \cov_1 \preceq \beta \Id$. Let $(\hatmu_2, \hatQ_2, \precest_2)$ and $(\hatmu_3, \hatQ_3, \precest_3)$ be any triplets of the same shape as $(\vmu_1, \cov_1, \cov_1^{-1})$ with condition that $\| \precest_2 \|_\op, \| \precest_3 \|_\op \lesssim \frac{1}{\alpha}$. Then, with probability at least $1 - \delta$ over $\x \sim \normal( \vmu_1, \cov_1 )$, we have
    \begin{align*}
        \big| \|\x - \hatmu_2 \|_{\precest_2}^2 &- \| \x - \hatmu_3 \|_{\precest_3}^2 - \iprod{ \cov_1,  (\precest_2 - \precest_3) } \big| \lesssim \beta \| \precest_2 - \precest_3 \|_F \log \frac{1}{\delta} \\ & + \frac{1}{\alpha} \big( \| \vmu_1 - \hatmu_2 \|^2 + \| \vmu_1 - \hatmu_3 \|^2 \big) + \sqrt{\beta} \log \frac{1}{\delta} ( \| \precest_2 - \precest_3  \|_\op \| \vmu_1 - \hatmu_2 \| + \frac{1}{\alpha} \| \hatmu_3 - \hatmu_2 \| )
    \end{align*}
    
\end{lemma}
\begin{proof}
    For $\x \sim \normal(\vec \mu_1, \cov_1)$, we rewrite $\|\x - \hatmu_2 \|_{\precest_2}^2 - \| \x - \hatmu_3 \|_{\precest_3}^2$ by writing $\x = \cov_1^{1/2}\z + \vmu_1$ for $\vec z \sim \normal(0, \Id)$, obtaining:
    \begin{equation}
    \label{eq:sq-diff-expansion}
    \begin{aligned}
        \|\x - \hatmu_2 \|_{\precest_2}^2 - \| \x - \hatmu_3 \|_{\precest_3}^2 = & \; \| \cov_1^{1/2} \vec z \|_{\precest_2}^2 - \| \cov_1^{1/2} \vec z \|_{\precest_3}^2  + \| \vmu_1 - \hatmu_2 \|_{\precest_2}^2 - \| \vmu_1 - \hatmu_3 \|^2_{\precest_3} \\
        & +2 (\cov_1^{1/2} \vec z)^\top \precest_2 (\vmu_1 - \vec \hatmu_2) - 2 (\cov_1^{1/2} \vec z)^\top \precest_3 (\vmu_1 - \hatmu_3)
    \end{aligned}
    \end{equation}
    We would like to bound the first two terms in the above equation using Hanson-Wright (\Cref{fact:HW}). Using $\| \cov_1  \| \leq \beta$, we have $\| \cov_1^{1/2} (\precest_2 - \precest_3) \cov_1^{1/2} \| \leq \beta \| \precest_2 - \precest_3 \|_F$. Using Hanson-Wright on the quadratic form $\vec z^\top \cov_1^{1/2}(\precest_2 - \precest_3) \cov_1^{1/2} \vec z$, we have for any $\delta > 0$ that
    \begin{equation*}
        \Pr_{\vec z \sim \normal (0, \Id)}\Bigl[ |\| \cov_1^{1/2} \vec z \|_{\precest_2}^2 - \| \cov_1^{1/2} \vec z \|_{\precest_3}^2 - \iprod{ \cov_1,  (\precest_2 - \precest_3) }  | \gtrsim \beta \| \precest_2 - \precest_3 \|_F \log \frac{1}{\delta} \Bigr] \leq \delta.
    \end{equation*}
    We simplify the sum of the last two terms in \eqref{eq:sq-diff-expansion} to obtain
    \begin{equation}
    \label{eq:sq-diff-linear-term}
    \begin{aligned}
        (\cov_1^{1/2} \vec z)^\top \precest_2 (\vmu_1 - \vec \hatmu_2) & - (\cov_1^{1/2} \vec z)^\top \precest_3 (\vmu_1 - \hatmu_3) = (\cov_1^{1/2} \vec z)^\top ( \precest_2 - \precest_3 ) (\vec \mu_1 - \hatmu_2) + (\cov_1^{1/2} \vec z)^\top  \precest_3 (\hatmu_3 - \hatmu_2).
    \end{aligned}
    \end{equation}
    Using the bounds $\| \cov_1 \|_{\rm op} \leq \beta$ and $\| \precest_3 \|_{\rm op} \lesssim 1/\alpha$, we can upper bound $\| \cov_1^{1/2} ( \precest_2 - \precest_3 ) (\vmu_1 - \hatmu_2) \| \lesssim \sqrt{\beta} \| \precest_2 - \precest_3  \|_\op \| \vmu_1 - \hatmu_2 \|$ and $\| \cov_1^{1/2} \precest_3 (\hatmu_3 - \hatmu_2) \| \lesssim \sqrt{\beta} \| \hatmu_3 - \hatmu_2 \| / \alpha$. So with probability at least $1 - \delta$, we have 
    \begin{equation*}
    \| (\cov_1^{1/2} \vec z)^\top \precest_2 (\vmu_1 - \vec \hatmu_2) - (\cov_1^{1/2} \vec z)^\top \precest_3 (\vmu_1 - \hatmu_3) \| \leq \sqrt{\beta} \log \frac{1}{\delta} ( \| \precest_2 - \precest_3  \|_\op \| \vmu_1 - \hatmu_2 \| + \frac{ \| \hatmu_3 - \hatmu_2 \|  }{ \alpha }). 
    \end{equation*}
    Putting everything together in \eqref{eq:sq-diff-expansion} and assuming $\alpha \leq 1$ and $\beta \geq 1$ to simplify, we obtain the result. 
\end{proof}

\begin{lemma}
\label{lem:poly-approx-score-bounded-interval}
    Let $\mc M(U)$ be a mixture of $\nc$ Gaussians with well-conditioned covariances $\alpha \Id \preceq \cov_i \preceq \beta \Id$ for all $i \in [m]$. Let $\withincluster$ be an upper bound on the parameter distance between any two components, i.e., $ \| \vec \mu_i - \vec \mu_j \| + \| \vec \cov_i - \vec \cov_j \|_F \leq \withincluster$ for all $i, j \in [m]$. Then, for $\x \sim \mc M(U)$ and for any $j \in [m]$, with probability at least $1 - \delta$, we have
    \begin{align*}
        &| \|\x - \vec \mu_j \|_{\cov_j^{-1}}^2 - \| \x - \vec \mu_1 \|_{\cov_1^{-1}}^2 - \iprod{ \cov_1,  (\cov_j^{-1} - \cov_1^{-1}) } | \lesssim \zeta_1 \hspace{3mm} \text{where} \hspace{3mm} \zeta_1 \triangleq \frac{ \beta \withincluster^2 }{ \alpha^2 } \log \frac{\nc}{\delta} \\ 
        & \text{and} \;\; \| (\cov_i^{-1} - \cov_1^{-1}) (\x - \vec \mu_i) \| \lesssim \zeta_2  \hspace{3mm} \text{where} \hspace{3mm} \zeta_2 \triangleq \frac{\sqrt{\beta} \withincluster^2}{\alpha^2} \log \frac{\nc}{\delta} \; .
    \end{align*}
    Combining it with \Cref{prop:poly-approx-softmax}, we obtain that there exists a polynomial $p(\x; \mathcal M(U))$ of degree $O(\frac{ \zeta_1 \zeta_2 \nc^2 }{ \eps  } )$ 
    and coefficients bounded in magnitude by $dR\exp(\wt{O}(\frac{\zeta_1 \zeta_2 \nc^2}{\eps} ))$ 
    such that
    \begin{equation*}
        \Pr_{\x \sim \mathcal M(U)} \big[ \| s(\x; \mathcal M(U)) - p(\x; \mathcal M(U)) \| \leq \eps \big] \geq 1 - \delta \; .
    \end{equation*}
\end{lemma}

\begin{proof}
    Recall that the score function for the mixture is
    \begin{equation*}
        s(\x; \mathcal M(U)) = \sum_{i \in U} w_i(\x) \cov_i^{-1} (\x - \vec \mu_i) \hspace{5mm} \text{where} \hspace{5mm} w_i(\x) = \frac{ \lambda_i \det(\cov_i)^{-1/2} \; e^{-\frac{1}{2} \|\x - \vec \mu_i \|_{\cov_i^{-1}}^2 } }{ \sum_{j \in U} \lambda_i \det(\cov_j)^{-1/2} \; e^{- \frac{1}{2}\|\x - \vec \mu_j \|_{\cov_j^{-1}}^2 } }.
    \end{equation*}
    We can rewrite the score function as $s(\x; \mathcal M(U)) = s_1(\x; \mathcal M(U)) + s_2(\x; \mathcal M(U)) + \cov_1^{-1} (\x - \vec \mu_1)$ where $s_1(\x; \mathcal M(U))$ and $s_2(\x; \mathcal M(U))$ are defined as
    \begin{align*}
        & s_1(\x; \mathcal M(U)) = \sum_{i \in U} w_i(\x) (\cov_i^{-1} - \cov_1^{-1}) (\x - \vec \mu_i) \text{ and } s_2(\x; \mathcal M(U)) =- \sum_{i \in U} w_i(\x) \cov_1^{-1} (\vec \mu_i - \vec \mu_1) \\
         & \text{ and } \hspace{3mm} w_i(\x) = \frac{ e^{- \frac{1}{2}( \|\x - \vec \mu_i \|_{\cov_i^{-1}}^2 - \| \x - \vec \mu_1 \|_{\cov_1^{-1}}^2 + \log ( \frac{\det(\cov_i)}{ \det(\cov_1) } ) ) + \log \frac{\lambda_i}{\lambda_1} } }{ 1 + \sum_{j=2}^\nc \; e^{- \frac{1}{2}( \|\x - \vec \mu_j \|_{\cov_j^{-1}}^2 - \| \x - \vec \mu_1 \|_{\cov_1^{-1}}^2 + \log ( \frac{\det(\cov_j)}{ \det(\cov_1) } ) ) + \log \frac{\lambda_j}{\lambda_1} } }\,.
    \end{align*}
    We show the polynomial approximation result for $s_1(\x; \mathcal M(U))$ and $s_2(\x; \mathcal M(U))$ using \Cref{prop:poly-approx-softmax}. To prove an upper bound on $\| \vec g_i(\x) \|$ in \Cref{prop:poly-approx-softmax}, we apply \Cref{lem:concentration-poly-approx-input} for all $j, \ell \in [\nc]$ and have that with probability at least $1-\delta$ over $\x \sim \normal(\vmu_\ell, \cov_\ell)$ (and hence over $\x \sim \mathcal M(U)$), we have
    \begin{align*}
        \big| \|\x - \vmu_j \|_{\cov_j^{-1}}^2 &- \| \x - \vmu_1 \|_{\cov_1^{-1}}^2 - \iprod{ \cov_\ell,  \cov_j^{-1} - \cov_1^{-1} } \big| \lesssim \beta \| \cov_j^{-1} - \cov_1^{-1} \|_F \log \frac{\nc}{\delta} \\ 
        & + \frac{1}{\alpha}  (\| \vmu_\ell - \vmu_j \|^2 + \| \vmu_\ell - \vmu_1 \|^2) + \sqrt{\beta} \log \frac{\nc}{\delta} ( \| \cov_j^{-1} - \cov_1^{-1}  \|_\op \| \vmu_\ell - \vmu_j \| + \frac{1}{\alpha} \| \vmu_j - \vmu_1 \| ).
    \end{align*}
    Using $\| \cov_i^{-1} \|_{\rm op} \leq 1/\alpha$ for all $i \in [k]$, we have $\| \cov_j^{-1} - \cov_1^{-1} \|_F = \| \cov_j^{-1} ( \cov_j - \cov_1 ) \cov_1^{-1} \|_F \leq \withincluster/\alpha^2$, we have
    \begin{align*}
        \beta \| \cov_j^{-1} - \cov_1^{-1} \|_F & \log \frac{\nc}{\delta} + \frac{1}{\alpha}  (\| \vmu_\ell - \vmu_j \|^2 + \| \vmu_\ell - \vmu_1 \|^2) \\ 
        &+ \sqrt{\beta} \log \frac{\nc}{\delta} ( \| \cov_j^{-1} - \cov_1^{-1}  \|_\op \| \vmu_\ell - \vmu_j \| + \frac{1}{\alpha} \| \vmu_j - \vmu_1 \| ) \leq \frac{ \beta \withincluster^2 }{ \alpha^2 } \log \frac{\nc}{\delta}
    \end{align*}
    We add and subtract $\iprod{ \cov_1,  (\cov_j^{-1} - \cov_1^{-1}) }$ on the left side and rearranging the terms and 
    \begin{align*}
        | \|\x - \vec \mu_j \|_{\cov_j^{-1}}^2 - \| \x - \vec \mu_1 \|_{\cov_1^{-1}}^2 - \iprod{ \cov_1,  (\cov_j^{-1} - \cov_1^{-1}) } | &\lesssim \frac{ \beta \withincluster^2 }{ \alpha^2 } \log \frac{\nc}{\delta} + \| \cov_\ell - \cov_1 \|_F \| \cov_j^{-1} - \cov_1^{-1} \|_F \\ 
        &\lesssim \frac{ \beta \withincluster^2 }{ \alpha^2 } \log \frac{\nc}{\delta}\,.
    \end{align*}
    We have $\| (\cov_i^{-1} - \cov_1^{-1}) \cov_\ell^{1/2} \|_F^2 \leq \frac{\beta \withincluster^2}{\alpha^4}$. For a fixed $\ell \in [\nc]$, when $\x \sim \normal (\vmu_\ell, \cov_\ell)$, we can rewrite $\|  (\cov_i^{-1} - \cov_1^{-1}) (\x - \vec \mu_i) \|$ by expressing $\x = \cov^{1/2}_\ell \z + \vmu_\ell$ for $\z \sim \normal(0, \Id)$ to get:
    \begin{equation}
    \label{eq:linear-term-bound}
        \| (\cov_i^{-1} - \cov_1^{-1}) (\x - \vec \mu_i) \| \leq \| (\cov_i^{-1} - \cov_1^{-1}) \cov_\ell^{1/2}\z \| + \| (\cov_i^{-1} - \cov_1^{-1}) ( \vmu_\ell - \vmu_i ) \|
    \end{equation}
    Using Hanson-Wright (\Cref{fact:HW}), with at least $1 - \delta$ probability over $\z \sim \normal(0, \Id)$, $\| (\cov_i^{-1} - \cov_1^{-1}) \cov_\ell^{1/2} \z\| \lesssim \| (\cov_i^{-1} - \cov_1^{-1}) \cov_\ell^{1/2} \|_F (1 + \log \frac{1}{\delta}) \lesssim \frac{ \sqrt{ \beta } \withincluster}{ \alpha^2 } \log \frac{1}{\delta}$. Using this bound in \eqref{eq:linear-term-bound}, with probability at least $1 - \delta$ over $\x \sim \mc M(U)$, we have
    \begin{equation*}
        \| (\cov_i^{-1} - \cov_1^{-1}) (\x - \vec \mu_i) \| \lesssim \frac{\withincluster^2}{\alpha} + \frac{\sqrt{\beta} \withincluster}{\alpha^2} \log \frac{\nc}{\delta} \lesssim \frac{\sqrt{\beta} \withincluster^2}{\alpha^2} \log \frac{\nc}{\delta}.
    \end{equation*}
    We apply \Cref{prop:poly-approx-softmax} to $s_1(\x; \mc M(U))$ with the softmax function taking input $r_j(\x) = -\frac{1}{2}\|\x - \vec \mu_j \|_{\cov_j^{-1}}^2 + \frac{1}{2}\| \x - \vec \mu_1 \|_{\cov_1^{-1}}^2 + \frac{1}{2}\iprod{ \cov_1,  (\cov_j^{-1} - \cov_1^{-1}) }$ and $\vec \theta_j = \log \frac{\lambda_j}{\lambda_1} - \frac{1}{2}\iprod{ \cov_1,  (\cov_j^{-1} - \cov_1^{-1}) } +  \frac{1}{2}\log \frac{\det(\cov_1)}{ \det(\cov_j) }$. We take $L$ and $M$ therein to be of order $\frac{\beta\withincluster^2}{\alpha^2}\log\frac{\nc}{\delta}$ and $\frac{\sqrt{\beta}\withincluster^2}{\alpha^2}\log\frac{\nc}{\delta}$ respectively. We conclude that there exists a polynomial transformation $p_1(\x; \mathcal M(U))$ with degree 
    $O( L M \nc^2 / \eps ) = O( \zeta_1 \zeta_2 \nc^2 / \eps )$ 
    such that with probability at least $1 - \delta$ over $\x \sim \mathcal M(U)$, we have
    \begin{equation*}
        \| s_1(\x; \mc M(U)) - p_1(\x; \mathcal M(U) ) \| \leq \eps\,.
    \end{equation*}
    Note that the multiplication of $(\cov_i^{-1} - \cov_1^{-1}) (\x - \vec \mu_i)$ to the polynomial approximation of the softmax can increase the sum of absolute values of coefficients at most by a factor of $\frac{d R \nc}{\alpha}$. The sum of absolute values of coefficients of the polynomial transformation $p_1(\x; \mathcal M(U))$ is $\frac{d R \nc}{\alpha} \exp( \wt{O}( \frac{ \zeta_1 \zeta_2 \nc^2 }{ \eps } ) )$.

    We also have $\|\cov_1^{-1} (\vmu_j - \vmu_1)\| \leq \withincluster / \alpha$. We apply \Cref{prop:poly-approx-softmax} for $s_2(\x; \mc M(U))$ with the same choice of $r_j(\x)$ and $L$ but we take $\vec g_j(\x)$ and $M$ as $\cov_1^{-1} (\vmu_j - \vmu_1)$ and $\withincluster / \alpha$. Therefore, we obtain that there exists a polynomial $p_2(\x; \mathcal M(U) )$ with degree $\frac{\beta\nc^2\withincluster^3}{\eps \alpha^3}\log\frac{\nc}{\delta}$
    such that with at least $1 - \delta$ probability, we have
    \begin{equation*}
        \| s_2(\x; \mc M(U)) - p_2(\x; \mc M(U)) \| \leq \eps\,. 
    \end{equation*}
    Combining the polynomials $p_1(\x; \mc M(U))$ and $p_2(\x; \mc M(U))$, we obtain the result. 
\end{proof}

We define $\ballc{j}_1(\x)$ to measure relative distance of $j^{\text{th}}$ input of the softmax to its mean and $\ballc{j}_2(\x)$ to measure norm of $(\cov_j^{-1} - \cov_1^{-1}) (\x - \vec \mu_j)$ as follows:
\begin{align*}
    \ballc{j}_1(\x) & \triangleq \|\x - \vec \mu_j \|_{\cov_j^{-1}}^2 - \| \x - \vec \mu_1 \|_{\cov_1^{-1}}^2 - \iprod{ \cov_1,  (\cov_j^{-1} - \cov_1^{-1}) }   \\
    \ballc{j}_2(\x) &\triangleq \| (\cov_j^{-1} - \cov_1^{-1}) (\x - \vec \mu_j) \|^2
\end{align*}

We similarly define $\hballc{j}_1$ and $\hballc{j}_2$ using estimates $\{ \hatmu_i, \hatQ_i, \precest_i \}_{i=1}^k$ instead of $\{ \vmu_i, \cov_i, \cov_i^{-1} \}_{i=1}^k$. Define $\ballt(\cdot)$ to be the indicator function for whether the input to the softmax is close to its mean and $(\cov_j^{-1} - \cov_1^{-1}) (\x - \vec \mu_j)$ is sufficiently small in norm: 
\begin{equation*}
    \ballt(\x, \theta_1, \theta_2) \triangleq \bigwedge^k_{j=1} \ball{j}(\x, \theta_1, \theta_2) \hspace{5mm} \text{where} \hspace{5mm}
    \ball{j}(\x, \theta_1, \theta_2) \triangleq \Ind \Big\{ \Bigl( \big| \ballc{j}_1(\x) \big| \le \theta_1 \Bigr) \wedge \Bigl(  \ballc{j}_2(\x) \leq \theta_2 \Bigr) \Big\} 
\end{equation*}
 Observe that the polynomial approximation result of \Cref{lem:poly-approx-score-bounded-interval} holds when $B(\x, \theta_1, \theta_2) = 1$. We also define $\hballt$ and $\hball{j}$ by replacing $\ballc{j}_1$ and $\ballc{j}_2$ with $\hballc{j}_1$ and $\hballc{j}_1$ in the definition of $\ballt$ and $\ball{j}$. 

Following the parameters used in the proof of \Cref{lem:poly-approx-score-bounded-interval}, we will take
\begin{equation}
\label{eq:ball-definition}
    \theta_1 \triangleq  \Theta\Bigl(\frac{ \beta \withincluster^2 }{ \alpha^2 } \log \frac{\nc}{\delta}\Bigr) \qquad \theta_2 \triangleq \Theta\Bigl(\frac{\sqrt{\beta} \withincluster^2}{\alpha^2} \log \frac{\nc}{\delta}\Bigr)\,.
\end{equation}

\begin{lemma}
\label{lem:ball-diff-actual-estimate} Let $\mc M(U)$ be a mixture of $\nc$ Gaussians with $\alpha \Id \preceq \cov_i \preceq \beta \Id$ and parameters satisfying $\| \vmu_i - \vmu_j \| + \| \cov_i - \cov_j \|_F \leq \withincluster$ for all $i, j \in [\nc]$. Let $\{ \hatmu_i, \hatQ_i, \precest_i \}_{i=1}^\nc$ be the estimates of the parameters $\{ \vmu_i, \cov_i, \cov_i^{-1} \}_{i=1}^\nc$ within parameter distance $\| \hatmu_i - \vmu_i \| + \| \hatQ_i - \cov_i^{-1} \|_F + \| \precest_i - \cov_i^{-1} \|_F \leq \paramerr$ and with the operator norm satisfying $\| \precest_i \|_\op \lesssim \frac{1}{\alpha}$ for all $i \in U$. 
Then, for any $\x \sim \mathcal M(U)$, with probability at least $1 - \delta$, the error in estimating $\ballc{j}_1(\x)$ by $\hballc{j}_1(\x)$ (similarly $\ballc{j}_2(\x)$ by $\hballc{j}_2(\x)$) is upper bounded by
    \begin{align*}
        \big| \ballc{j}_1(\x) - \hballc{j}_1(\x) \big| &\lesssim \thetadiff_1 \hspace{1cm} \text{where} \hspace{1cm} \thetadiff_1 \triangleq \frac{ \beta \withincluster^2 \paramerr^2 }{ \alpha } \log \frac{\nc}{\delta} \; , \\
        \big| \ballc{j}_2(\x) - \hballc{j}_2(\x) \big| &\lesssim \thetadiff_2 \hspace{1cm} \text{where} \hspace{1cm} \thetadiff_2 \triangleq \frac{\beta \withincluster^4 \paramerr^3}{\alpha^4} \log \frac{\nc}{\delta} \ .
    \end{align*}
\end{lemma} 

\begin{proof}
    The expression of $\ballc{j}_1(\x) - \hballc{j}_1(\x)$ can be rewritten as 
    \begin{equation}
    \label{eq:ball-condition-diff-estimated-actual}
        \begin{aligned}
            \ballc{j}_1(\x) - \hballc{j}_1(\x) = & \; \big( \| \x - \vmu_j \|_{\cov_j^{-1}}^2 - \| \x - \hatmu_j \|_{\precest_j} - \iprod{ \cov_\ell, \cov_j^{-1} -  \precest_j } \big) \\
            &\qquad- \big( \| \x - \vmu_1 \|_{\cov_1^{-1}}^2 - \| \x - \hatmu_1 \|_{\precest_1} - \iprod{ \cov_\ell, \cov_1^{-1} -  \precest_1 } \big) \\
            &\qquad+ \iprod{ \cov_\ell - \cov_1, \cov_j^{-1} - \precest_j + \precest_1 - \cov_1^{-1} } + \iprod{ \cov_1 - \hatQ_1, \precest_1 - \precest_j }.
        \end{aligned}
    \end{equation}
    Using \Cref{lem:concentration-poly-approx-input} by choosing $\normal (\vmu_1, \cov_1)$ as $\normal (\vmu_\ell, \cov_\ell)$ and $(\hatmu_2, \hatQ_2, \precest_2), (\hatmu_3, \hatQ_3, \precest_3)$ as $(\vmu_j, \cov_j, \cov_j^{-1})$ and $(\hatmu_j, \hatQ_j, \precest_j)$ and applying the union bound over $j, \ell \in U$, for $(\ell, \x) \sim \mathcal M^J(U)$, with at least $1 - \delta$ probability, we have
    \begin{align*}
        \big| \| \x - \vmu_j \|_{\cov_j^{-1}}^2 - \| \x - \hatmu_j \|_{\precest_j} - \iprod{ \cov_\ell, \cov_j^{-1} -  \precest_j } \big| &\lesssim \beta \paramerr \log \frac{\nc}{\delta} + \frac{1}{\alpha} ( \withincluster^2 + \paramerr^2) + \sqrt{\beta} \log \frac{\nc}{\delta} (\paramerr \withincluster + \frac{\paramerr}{\alpha}) \\
        &\lesssim \frac{ \beta \withincluster^2 \paramerr^2 }{ \alpha }  \log \frac{\nc}{\delta}.
    \end{align*}
    For $j=1$ in the above equation, we also have
    \begin{equation*}
        \big| \| \x - \vmu_1 \|_{\cov_1^{-1}}^2 - \| \x - \hatmu_1 \|_{\precest_1} - \iprod{ \cov_\ell, \cov_1^{-1} -  \precest_1 } \big| \lesssim \frac{\beta \withincluster^2 \paramerr^2}{ \alpha }  \log \frac{\nc}{\delta} 
    \end{equation*}
    Note that $\iprod{ \cov_1 - \hatQ_1, \precest_1 - \precest_j } \lesssim \paramerr (\paramerr + \frac{\withincluster}{\alpha^2} )$ therefore, the last term in \eqref{eq:ball-condition-diff-estimated-actual} can be upper bounded as 
    \begin{equation*}
        \big| \iprod{ \cov_\ell - \cov_1, \cov_j^{-1} - \precest_j + \precest_1 - \cov_1^{-1} } + \iprod{ \cov_1 - \hatQ_1, \precest_1 - \precest_j } \big| \lesssim \frac{ \paramerr^2 \withincluster }{\alpha^2} \ . 
    \end{equation*}
    When $\vec z \sim \normal(0, \Id)$, using Hanson-Wright (\Cref{fact:HW}), we have $\| (\cov_j^{-1} - \cov_1^{-1}) \cov_\ell^{1/2} \vec z \|^2 \lesssim \frac{\beta \withincluster^2}{\alpha^4}\log \frac{1}{\delta}$ with probability at least $1-\delta$. Additionally, we have $\|  (\vmu_\ell - \vmu_j) (\cov_j^{-1} - \cov_1^{-1})^\top (\cov_j^{-1} - \cov_1^{-1}) \cov_\ell^{1/2} \| \lesssim \sqrt{\beta} \withincluster^3/\alpha^4.$ Therefore, with at least $1-\delta$ probability, we obtain $|  (\vmu_\ell - \vmu_j) (\cov_j^{-1} - \cov_1^{-1})^\top (\cov_j^{-1} - \cov_1^{-1}) \cov_\ell^{1/2} \vec z | \leq \frac{\sqrt{\beta} \withincluster^3}{\alpha^4} \log \frac{\nc}{\delta}$  Moreover, $\| (\cov_j^{-1} - \cov_1^{-1}) (\vmu_\ell - \vmu_j) \|^2 \leq \beta \withincluster^4 / \alpha^4 $. Therefore, for $\x \sim \mathcal M(U)$, with probability at least $1-\delta$, we have 
    \begin{equation}
    \label{eq:V2j-bound}
        \| (\cov_j^{-1} - \cov_1^{-1})(\x - \vmu_j) \|^2 \lesssim \frac{\beta \withincluster^4}{\alpha^4}\log \frac{\nc}{\delta}.
    \end{equation}
    Similarly, for any $\ell \in [\nc],$ we have $\| ( \precest_j - \precest_1 ) \cov_\ell^{1/2} \|_F^2 \lesssim \beta (\| \precest_j - \cov_j^{-1} \|_F^2 +  \| \precest_1 - \cov_1^{-1} \|_F^2 + \| \cov_j^{-1} - \cov_1^{-1} \|_F^2) \lesssim \beta ( \paramerr^2 + \withincluster^2/ \alpha^4)$. Using Hanson-Wright inequality (\Cref{fact:HW}) for $\vec z \sim \normal ( 0, \Id )$, with probability at least $1-\delta$, we have $\| ( \precest_j - \precest_1 ) \cov_\ell^{1/2} \vec z \|^2 \lesssim (\beta \paramerr^2 \withincluster^2 \log (\nc / \delta) ) / \alpha^4.$ We also have $\| (\vmu_\ell - \hatmu_j) ( \precest_j - \precest_1 )^\top ( \precest_j - \precest_1 ) \cov_\ell^{1/2} \| \lesssim \sqrt{\beta} (\paramerr + \withincluster) ( \paramerr + \withincluster/\alpha^2 )^2 \lesssim  \sqrt{\beta} \paramerr^3 \withincluster^3 / \alpha^4.$ This implies that with probability at least $1-\delta$, we have $ \big| (\vmu_\ell - \hatmu_j) ( \precest_j - \precest_1 )^\top ( \precest_j - \precest_1 ) \cov_\ell^{1/2} \big| \lesssim \frac{ \sqrt{\beta} \paramerr^3 \withincluster^3 }{ \alpha^4 } \log \frac{\nc}{\delta}$. We also have $\| ( \precest_j - \precest_1 ) (\vmu_\ell - \hatmu_j) \|^2 \lesssim \withincluster^2 (\paramerr^2 + \frac{\withincluster^2}{\alpha^4}) \lesssim \frac{\paramerr^2 \withincluster^4}{\alpha^4}$. Combining all the bounds, for $\x \sim \mathcal M(U)$, with probability at least $1-\delta$, we have
    \begin{equation*}
        \| (\precest_j - \precest_1) (\x - \hatmu_j) \| \lesssim \frac{\beta \paramerr^3 \withincluster^4}{\alpha^4} \log \frac{\nc}{\delta}. 
    \end{equation*}
    Combining this bound with \eqref{eq:V2j-bound}, we obtain the result. 
\end{proof}

We now prove our main proposition of this section. 

\begin{proof}[Proof of \Cref{lem:cluster-poly-approx}]
    We set $\htheta_1 = c_1\frac{\beta \withincluster^2 \paramerr^2 }{ \alpha^2 } \log \frac{\nc}{\delta}$ and $\htheta_2 = c_2 \frac{\beta \paramerr^3 \withincluster^4}{\alpha^4} \log \frac{\nc}{\delta}$ for some large constant $c_1$ and $c_2$. 
    \begin{align*}
        & \E_{\x \sim \mc M(U)}[\| s(\x; \mathcal M(U)) - p(\x; \mathcal M(U)) \Ind \{ \hballt(\x, \htheta_1, \htheta_2) = 1 \} - \precest_1 (\x - \hatmu_1) \Ind \{ \hballt(\x, \htheta_1, \htheta_2) = 0 \}  \|^2 ] \\ 
        = & \E_{\x \sim \mc M(U)}[ \| s(\x; \mathcal M(U)) - p(\x; \mathcal M(U)) \|^2 \Ind \{ \hballt(\x, \htheta_1, \htheta_2 ) = 1 \} ] \\ 
        &+ \E_{\x \sim \mc M(U)}[ \| s(\x; \mathcal M(U)) - \precest_1 (\x - \hatmu_1) \|^2 \Ind \{ \hballt(\x, \htheta_1, \htheta_2 ) = 0 \} ]
    \end{align*}
    \Cref{lem:ball-diff-actual-estimate} gives us that $|\hballc{j}_1(\x)| \leq \htheta_1$ implies that $\ballc{j}(\x) \leq \htheta_1 + \thetadiff_1$ for all $\x$ and for all $j \in U$ and hence, $\ballt(\x, \htheta_1 + \thetadiff_1, \htheta_2 + \thetadiff_2)$. We apply \Cref{lem:poly-approx-score-bounded-interval} with $\zeta_1$ as $c_1\frac{\beta \withincluster^2 \paramerr^2 }{ \alpha^2 } \log \frac{\nc}{\delta}$ and $\zeta_2$ as $c_2 \frac{\beta \paramerr^3 \withincluster^4}{\alpha^4} \log \frac{\nc}{\delta}$ and obtain that there exist a polynomial $p(\x; \mathcal M(U))$ of degree $O( \frac{  \beta^2 \nc^2 \paramerr^5 \withincluster^6 }{ \alpha^6 \eps} \log^2 \frac{\nc}{\delta} )$ and coefficients bounded in magnitude by $d R \exp( \wt{O} ( \frac{  \beta^2 \nc^2 \paramerr^5 \withincluster^6 }{ \alpha^6 \eps} \log^2 \frac{\nc}{\delta} ) )$ such that the following holds:
    \begin{equation*}
        \E_{\x \sim \mc M(U)}[ \| s(\x; \mathcal M(U)) - p(\x; \mathcal M(U)) \|^2 \Ind \{ \hballt(\x, \htheta_1, \htheta_2 ) = 1 \} ] \lesssim \eps.
    \end{equation*}
    We can upper bound the error when $\hballt(\x, \htheta_1, \htheta_2 ) = 0 $ using Cauchy-schwarz inequality as follows:
    \begin{align*}
        & \E_{\x \sim \mc M(U)}[ \| s(\x; \mathcal M(U)) - \precest_1 (\x - \hatmu_1) \|^2 \Ind \{ \hballt(\x, \htheta_1, \htheta_2) \} ] \\ 
        = & \Bigl( \E_{\x \sim \mc M(U)}[ \| s(\x; \mathcal M(U)) - \precest_1 (\x - \hatmu_1) \|^4 ] \Bigr)^{1/2} \big( \Pr [ \hballt(\x, \htheta_1, \htheta_2 ) = 0 ] \big)^{1/2}
    \end{align*}
    We know that $\Pr_{\x \sim \mathcal M(U)} [ \hballt(\x, \htheta_1, \htheta_2 ) = 0 ] \leq \delta$. We upper bound the other term as follows:
    \begin{equation}
    \label{eq:delta-error-score-approx}
        \begin{aligned}
            \E_{\x \sim \mc M(U)} [ \| s(\x; \mathcal M(U)) - \precest_1 (\x - \hatmu_1) \|^4 ] \leq & \; \nc^4 \sum_{i=1}^\nc \E_{\x \sim \mc M} [ \| \cov_i^{-1}(\x - \vmu_i) - \precest_1(\x - \hatmu_1) \|^4 ].
        \end{aligned}
    \end{equation}
    Writing $\x$ in terms of standard Gaussian $\vec z \sim \normal(0, \Id)$ for any $i, \ell \in [\nc]$, we have 
    \begin{align*}
        \E_{\vec z \sim \normal} & [ \| \cov_i^{-1} ( \cov_\ell^{1/2} \vec z + \vmu_\ell - \vmu_i ) - \; \precest_1 ( \cov_\ell^{1/2} \vec z + \vmu_\ell - \hatmu_1 ) \|^4 ] \\ 
        &\lesssim \E_{\vec z \sim \normal}[ \|(\cov_i^{-1} - \precest_1) \cov_\ell^{1/2} \vec z \|^4 ] + \| \cov_i^{-1} (\vmu_\ell - \vmu_i) \|^4 + \| \precest_1 (\vmu_\ell - \hatmu_1) \|^4 \\
        &\lesssim \frac{\beta^2 \withincluster^4}{\alpha^8} + \beta^2 \paramerr^4 + \frac{ \withincluster^4 }{\alpha^4} + \frac{\withincluster^2}{\alpha^4}  \lesssim \frac{\beta^2 \paramerr^4 \withincluster^4}{ \alpha^{8} },
    \end{align*}
    where the last inequality follows from \Cref{lem:wicks_scalar} and $\|(\cov_i^{-1} - \precest_1) \cov_\ell^{1/2} \|^4 \lesssim \beta^2 (\| \cov_i^{-1} - \cov_1^{-1} \|_F^4 + \| \cov_1^{-1} - \precest_1 \|_F^4 ) \lesssim \frac{\beta^2 \withincluster^4}{\alpha^8} + \beta^2 \paramerr^4$. 
    Putting together the above bounds, we obtain that there exists a polynomial $p(\x)$ such that 
    \begin{align}
        & \E_{\x \sim \mc M(U)}[\| s(\x; \mathcal M(U)) - p(\x; \mathcal M(U)) \Ind \{ B(\x;  \mathcal M(U) ) = 1 \} - \precest_1 (\x - \hatmu_1) \Ind \{ B(\x \mathcal M(U)) = 0 \}  \|^2] \\
        &\lesssim \; \eps + \sqrt{\delta} \frac{\beta^2 \paramerr^4 \withincluster^4}{ \alpha^{8} }
    \end{align} 
    Choosing $\delta = \frac{\eps^2 \alpha^{16}}{ \beta^4 \hbeta^8 \Delta^8 \withincluster^8 }$, we obtain the result. 
\end{proof}

\subsubsection{Piecewise polynomial approximation of the complete mixture}

The goal of this section is to prove that there exists a piecewise polynomial that can approximate the $\score(\x; \mc M)$. More precisely, there exists $\lscore (\x; \mc M(U_t))$ when used with the $\classify(\cdot)$, $\lscore$ is $\eps$-approximate to the true score function $\score$, i.e.,
\begin{equation*}
    \E_{\x \sim \mc M} \big[ \| \score (\x; \mc M) - \lscore(\x, \classify (\cdot) ) \|^2 \big] \leq \eps,
\end{equation*}
where $\lscore(\x, \classify(\cdot))$ is defined as
\begin{equation*}
    \lscore(\x, \classify (\cdot) ) = \sum_{t=1}^{\numclust} \lscore (\x; \mathcal M(U_t)) ~ \1\{ \classify(\x) = t \}
\end{equation*}
We will bound the error for every subset $U_t$. The error for the subset corresponding to $U_t$ can be decomposed into an error due to the score simplification of $\mathcal M$ to $\mathcal M(U_t)$ and an error due to the approximation $\mathcal M(U_t)$ to the piecewise polynomial score function. 
\begin{align}
    & \E_{\x \sim \mc M} \big[ \| \score (\x; \mc M) - \lscore(\x, \mc M(U_t)) \|^2 \Ind \{\classify(\x) = t \} \big] \nonumber \\
    & \quad = \E_{\x \sim \mc M} \big[ \| \score (\x; \mc M) - s(\x, \mc M(U_t)) \|^2 \Ind \{\classify(\x) = t \} \big] \label{eq:score-approx-err} \\ 
    & \quad \quad + \E_{\x \sim \mc M} \big[ \| \score (\x; \mc M(U_t)) - \lscore(\x, \mc M(U_t)) \|^2 \Ind \{\classify(\x) = t \} \big] \label{eq:poly-approx-err} .
\end{align}
Recall that the score simplification (\Cref{prop:score-simplification}) bounds the term in \eqref{eq:score-approx-err}. We rewrite \eqref{eq:poly-approx-err} in two parts, when samples are coming from $\mathcal M(U_t)$ and $\mathcal M(U_t^c)$ as follows
\begin{align}
    &\E_{\x \sim \mc M} \big[ \| \score (\x; \mc M(U_t)) - \lscore(\x, \mc M(U_t)) \|^2 \Ind \{\classify(\x) = t \} \big] \nonumber \\
    & \quad = \Pr[j \in U_t] \cdot \E_{\x \sim \mc M(U_t)} \big[\| \score (\x; \mc M(U_t)) - \lscore(\x; \mc M(U_t))\|^2 \Ind \{ \classify(\x) = t \} \big] \label{eq:correct-clustering-poly-approx-err} \\ 
    & \quad \quad + \Pr[j \in U_t^c]\cdot \E_{\x \sim \mc M(U_t^c)} \big[\| \score (\x; \mc M(U_t)) - \lscore(\x; \mc M(U_t))\|^2 \Ind \{ \classify(\x) = t \} \big] \label{eq:incorrect-clustering-poly-approx-err}
\end{align}
The term in \eqref{eq:correct-clustering-poly-approx-err} is upper bounded by $\eps$ using \Cref{lem:cluster-poly-approx}. 
In the following Lemma, we upper bound the term in \eqref{eq:incorrect-clustering-poly-approx-err}.

\begin{lemma}
\label{lem:incorrect-poly-approx-err}
Let $\mc M$ be a $(\alpha, \beta, R)$-well-conditioned mixture and let $U_t \subset [k]$ be a subset of components. Assume that the clustering function $\classify : \R^d \to [\numclust]$ satisfies $\Pr_{\x \sim \normal_i}[ \classify(\x) = t ] \leq \delta$ for all $i \notin U_t$ and $t \in [\numclust].$ Then, we have
\begin{equation*}
    \E_{\x \sim \mc M(U_t^c)} \big[\| s(\x; \mc M(U_t)) - \lscore(\x; \mc M(U_t))\|^2 \Ind \{ \classify(\x) = t \} \big] \lesssim \frac{ \beta^2 }{ \alpha^8 } k^3 (\paramerr \withincluster R)^4 \sqrt{\delta} \; .
\end{equation*}
\end{lemma}
\begin{proof}
    The term in \eqref{eq:incorrect-clustering-poly-approx-err} can be upper bounded by Cauchy-Schwarz inequality as follows:
    \begin{multline*}
        \E_{\x \sim \mc M(U_t^c)} \big[\| s(\x; \mc M(U_t)) - \lscore(\x; \mc M(U_t))\|^2 \Ind \{ \classify(\x) = t \} \big] \\ 
        \leq \Bigl( \E_{\x \sim \mc M(U_t^c)} \big[\| s(\x; \mc M(U_t)) - \lscore(\x; \mc M(U_t))\|^4 \Ind \{ \classify(\x) = t \} \big] \Bigr)^{1/2} \Pr_{\x \sim \mathcal M(U_t^c)} \big[ \classify(\x) = t \big]^{1/2}\,.
    \end{multline*}
    Using the definition of $\lscore(\x; \mc M(U_t))$, we can simplify the first term as 
    \begin{equation}
    \label{eq:score-misclass-err-bound}
    \begin{aligned}
        & \E_{\x \sim \mc M(U_t^c)} \big[\| \score (\x; \mc M(U_t)) - \lscore(\x; \mc M(U_t))\|^4 \big] \\ 
        & \quad =  \E_{\x \sim \mc M(U_t^c)} \big[\| \score (\x; \mc M(U_t)) - p(\x; \mc M(U_t))\|^4 \1 \{ \hballt(\x; U_t ) \} \big] \\ 
        & \quad \quad + \E_{\x \sim \mc M(U_t^c)} [ \| \score (\x; \mc M(U_t)) - \precest_i (\x - \hatmu_i) \|^4 \1 \{ \hballt^c(\x; U_t ) \} ] \\
    \end{aligned}
    \end{equation}
    The first term in \eqref{eq:score-misclass-err-bound} is upper bounded by $\eps^4$. The second term in \eqref{eq:score-misclass-err-bound} can be upper bounded by
    \begin{equation*}
        \E_{\x \sim \mc M(U_t^c)} [ \| \score (\x; \mc M(U_t)) - \precest_i (\x - \hatmu_i) \|^4 ] \leq k^3 \sum_{j \in U_t} \E_{\x \sim \mc M(U_t^c)} [ \| \cov_j^{-1}(\x - \vmu_j) - \precest_i (\x - \hatmu_i) \|^4 ].
    \end{equation*}
    We can upper bound $\x \sim \normal (\vmu_\ell, \cov_\ell)$ by writing it in terms of the standard normal $\z \sim \normal(0, \Id)$:
    \begin{align*}
        & \E_{\x \sim \normal (\mu_\ell, \cov_\ell)} [ \| \cov_j^{-1}(\x - \vmu_j) - \precest_i (\x - \hatmu_i) \|^4 ] \\ 
        \lesssim & \E_{\vec z \sim \normal (0, \Id)} \big[ \| ( \cov_j^{-1} - \precest_i ) \cov_\ell^{1/2} \z \|^4 \big] + \| \cov_j^{-1} (\vmu_\ell - \vmu_j) \|^4 + \| \precest_i ( \vmu_\ell - \hatmu_i ) \|^4 \\ 
        \lesssim & \hspace{3mm} \| ( \cov_j^{-1} - \precest_i ) \cov_\ell^{1/2} \|^4 + \frac{R^4}{\alpha^4} + \hbeta^4 (R^4 + \Delta^4) \\
        \lesssim & \hspace{3mm} \beta^2 \Bigl( \frac{\withincluster^4}{\alpha^8} + \paramerr^4 \Bigr) + \frac{R^4}{\alpha^4} + \frac{(R^4 + \paramerr^4)}{\alpha^4} \\
        \lesssim & \hspace{3mm}  \frac{ \beta^2 }{ \alpha^8 }  (\paramerr \withincluster R)^4.
    \end{align*}
    Additionally, we have
    \begin{equation*}
        \Pr_{\x \sim \mathcal M(U_t^c)} \big[ \classify(\x) = t \big] \leq \max_{j \in [k]: j \notin U_t } \quad \Pr_{\x \sim \normal_j}[ \classify(\x) = t ] \leq \delta.
    \end{equation*}
    Combining \Cref{eq:score-misclass-err-bound} with the above bound, we obtain the result. 
\end{proof}

\begin{proposition}
\label{prop:piecewise-poly-approx-score}
    Let $\mathcal M$ be $(\alpha, \beta, R)$-well-conditioned mixture and then, there exists a piecewise polynomial $\lscore(\x; \classify(\cdot) )$ 
    such that it satisfies
    \begin{equation*}
        \E_{\x \sim \mc M} \big[ \| \score (\x; \mc M) - \lscore(\x, \classify (\cdot) ) \|^2 \big] \leq \eps,
    \end{equation*}
    where $\lscore(\x, \classify (\cdot))$ is defined as
    \begin{align*}
        & \quad \lscore(\x, \classify (\cdot) ) = \sum_{t=1}^{\numclust} \lscore (\x; \mathcal M(U_t)) ~ \1\{ \classify(\x) = t \} \\
        \text{and} & \quad \lscore (\x; \mathcal M(U_t)) = p(\x; \mathcal M(U_t)) \Ind \{ \hballt(\x;  \mathcal M(U_t) ) = 1 \} \\ 
        & \qquad \qquad + \precest_j (\x - \hatmu_j) \Ind \{ \hballt (\x; \mathcal M(U_t)) = 0 \} \quad \text{for some $j \in U_t$ and $\hballt$ defined in \Cref{eq:ball-definition}} 
    \end{align*}
    Moreover, every polynomial $p(\x; \mathcal M(U_t))$ has the degree at most $\poly( \frac{\beta k}{ \alpha \lambdamin \eps } \log R )$ and coefficients of the polynomials are bounded in magnitude by $\poly(d) \exp( \poly( \frac{\beta k}{ \alpha \lambdamin \eps } \log R ) )$.
\end{proposition}

\begin{proof}
    Combining \Cref{eq:correct-clustering-poly-approx-err}, \Cref{eq:incorrect-clustering-poly-approx-err} and \Cref{lem:incorrect-poly-approx-err}, for a fixed $t \in [\numclust]$, we have 
    \begin{equation}
    \label{eq:poly-approx-err-2}
        \E_{\x \sim \mc M} \big[ \| \score (\x; \mc M(U_t)) - \lscore(\x, \mc M(U_t)) \|^2 \Ind \{\classify(\x) = t \} \big] \lesssim \eps + \frac{ \beta^2 }{ \alpha^8 } k^3 (\paramerr \withincluster R)^4 \sqrt{\delta}.
    \end{equation}
    We now combine the bound of the above equation with the score simplification guarantee. The score simplification guarantee (\Cref{prop:score-simplification}) assumes that the clustering function $\classify : \R^d \to [\numclust]$ satisfies $\Pr_{\x \sim \normal_i}[ \classify(\x) = t ] \leq \delta$ for all $i \notin U_t$ and $t \in [\numclust].$ and obtains that
    \[ 
        \E_{\x \sim \mathcal M}
        [\| s(\x; \mathcal M) - s (\x; \classify(\cdot)) \|_2^2 ] 
        \leq O(k^{5/4} (\beta^3/\alpha^5) R)\,\sqrt{\delta},
    \]
    Combining the above bound with \Cref{eq:poly-approx-err-2}, we have
    \begin{equation}
    \label{eq:unnormalized-final-error-1}
        \E_{\x \sim \mc M} \big[ \| \score (\x; \mc M) - \lscore(\x, \classify (\cdot) ) \|^2 \big] \lesssim k \eps + \frac{ \beta^3 }{ \alpha^8 } k^3 (\paramerr \withincluster R)^4 \sqrt{\delta}.
    \end{equation}
    Using clustering guarantee from \Cref{lem:main_cluster} for any $t \in [\numclust], i \in U_t$ and $t' \in [\numclust]$ and $t' \neq t$, we have 
    \begin{equation*}
        \Pr_{\x \sim \normal_i} [ \classify(\x) = t' ] \leq \Pr_{\x \sim \normal_i} [ \classify(\x) \neq t ] \leq k^3\exp\Bigl(-\Omega\Bigl(\frac{(\betweenclustermean)^2}{\alpha\sqrt{k}} \wedge \frac{\alpha^6 (\betweenclustercov)^2}{\beta^6\coverr^2} \wedge \frac{\alpha^2\betweenclustercov}{\beta^3}\Bigr)\Bigr). 
    \end{equation*}
    Recall that $\meanerr \lesssim \beta / \lambdamin$ and $\coverr \lesssim k^{3/2} \beta / \lambdamin + k^2 \alpha \log R$. Therefore, we choose $\betweenclustermean$ and $\betweenclustercov$ for some large constants $c_1$ and $c_2$ as follows which satisfies the conditions in \Cref{lem:main_cluster}.
    \begin{equation*}
        \betweenclustermean = c_1 \frac{ \beta \sqrt{k} }{ \lambdamin } \log \frac{k R \beta }{ \lambdamin \alpha \eps} \hspace{5mm} \text{and} \hspace{5mm} \betweenclustercov = c_2 \frac{\beta^4 k^2 \log R }{ \alpha^3 \lambdamin } \log \frac{k R \beta }{ \lambdamin \alpha \eps}.
    \end{equation*}
    We also choose $\withinclustermean \asymp k\betweenclustermean$ and $\withinclustercov \asymp k\betweenclustercov$. 
    Using the chosen values of $\betweenclustermean$ and $\betweenclustercov$, we have
    \begin{equation*}
        \Pr[ \classify(\x) = t \mid j \notin U_t ] \leq \eps^2 \; \poly \Big( \frac{\alpha \lambdamin }{ \beta k R } \Big) \ . 
    \end{equation*}
    Using this bound in \Cref{eq:score-misclass-err-bound}, we have
    \begin{equation*}
        \E_{\x \sim \mc M} \big[ \| \score (\x; \mc M) - \lscore(\x, \classify (\cdot) ) \|^2 \big] \lesssim k \eps .
    \end{equation*}
    Rescaling $\eps$ as $\eps / k$ and using $\withincluster = \withinclustermean + \withinclustercov$ and $\paramerr = \meanerr + \coverr$ in \Cref{lem:cluster-poly-approx}, we obtain the result. 
\end{proof}

\subsection{Learning polynomials using denoising objective}

The goal of this section is to provide details about our learning algorithm using denoising objective.  Recall that to sample from the data distribution, the diffusion reverse process uses an approximation to the score function $\nabla_{\x} \log q_t(\x)$. To learn the score function, we minimize the following DDPM objective in which one wants to predict the noise $\vec z_t$ from the noisy observation $\x_t$, i.e.

\begin{equation}
\label{eq:diffusion-loss-appendix}
    \min_{\vec g \in \mathcal G} \;\; L_t(\vec g_t) = \E_{\x_0, \vec z_t} \Bigl[ \Big\| \vec g_t(\x_t) + \frac{\vec z_t}{ \sqrt{1 - \exp(-2t)} } \Big\|^2 \Bigr]\,.
\end{equation}

Given parameter candidates $\{ (\hatmu_i, \hatQ_i) \}_{i=1}^k$ and a clustering function $\classify(\cdot)$, our learning algorithm minimizes the following empirical loss

\begin{equation}
\label{eq:empirical-loss}
\begin{aligned}
    & \min_{ \substack{ p(\x; \mathcal M(U_i)) \\ \forall i \in [k] } \;\; } \frac{1}{n} \sum_{i=1}^n \lossdenoisingclipped_t ( \lscore_t, \x_t^{(i)}, \z_t^{(i)}) \\
    \text{where} &= \Big\| \lscore_t(\x_t, \classify(\cdot)) + \frac{\z_t}{ \sqrt{1 - \exp(-2t)} } \Big\|^2 \Ind \{ \| \x_t \| \leq R_{\x}, \| \z_t \| \leq R_{\z} \},
\end{aligned}
\end{equation}
for some large choices of $R_{\x}, R_{\z} = \poly(d R \tau / \eps)^\ell$.  Clipping the loss for large values of $\| \x_t \|$ and $\| \z_t \|$ is for analysis purposes and in fact, we show that the choice of the value of $R_{\x}$ and $R_{\z}$ are sufficiently large such that the unclipped loss will be at most $O(\eps)$ in expectation.

\begin{proposition} 
\label{thm:learning-score-guarantee-fixed-t}
Let $\mathcal M$ be a $(\alpha,\beta,R)$-well-conditioned mixture. Then, for any $\eps, \delta > 0$ and noise scale $t \geq \eps$, there exists an algorithm that runs in $O(d^{\poly( \frac{\beta k \log R}{ \alpha \eps \lambdamin })} \poly(\log \frac{1}{\delta}))$ and returns a score function $\lscore_t$ such that with probability $1 - \delta$ over samples generated from the mixture $\mc M$, we have
\begin{equation}
\label{eq:small-population-loss-noise-t}
    \E_{\x_t \sim \noisymixture{t}} \big[ \| \lscore_t(\x_t) - \nabla_{\x} \log q_t(\x) \|^2 \big] \leq \eps.
\end{equation}
The algorithm to learn the score function takes input as noise scale $t$, target error $\eps$ and confidence $\delta$ and it is given by
\begin{itemize}
    \item Obtain a candidate list of parameters $\mathcal W \gets $ \textsc{CrudeEstimate}
    
    \item Brute force over the parameter candidate list $(\hatmu_1, \hatQ_1)\ldots(\hatmu_k, \hatQ_k) \in \mathcal{W}$
    \begin{itemize}
        \item Brute force over number of mean-based partition ($m$), number of covariance-based partition ($n$), mean-based partition $\mathcal S = \{S_1, S_2, \ldots, S_m\}$ and covariance-based partition $\mathcal T = \{T_1, \ldots, T_n \}$ 
        \begin{itemize}
            \item Brute force over possible thresholds $\{t_{ij}\}_{i,j=1}^k$ in range $[-c\frac{\beta d}{\alpha}, c\frac{\beta d}{\alpha}]$ for some large constant $c$.
            \begin{itemize}
                \item Clustering function $\classify \gets$ \textsc{Clustering}$(S, T, \{ (\hatmu_i, \hatQ_i) \}_{i=1}^k, \{ t_{i,j} \}_{i,j=1}^k )$
                \item $\lscore_t \gets$ minimizer of empirical loss \Cref{eq:empirical-loss}. 
                \item Compute the validation loss on the fresh samples for $\lscore_t$.
            \end{itemize}
        \end{itemize}
    \end{itemize}
\end{itemize}

In the end, the algorithm returns the $\lscore_t$ which has minimum validation loss across all brute force candidates. 
\end{proposition}

\subsection{Generalization error analysis}

As we can decompose the learning problem into learning a polynomial in the piece given by the clustering function $\classify (\cdot)$, we can start the generalization error argument by considering the loss function restricted to a fixed piece of the polynomial. 

Observe that the DDPM objective can be unbounded in general however, the loss becomes bounded assuming that $\| \x_t \| \leq R_{\x}$ and $\| \z_t \| \leq R_{\z}$. Therefore, we first derive the generalization error bound when we restrict the loss function to points $\| \x_t \| \leq R_{\x}$ and $\| \z_t \| \leq R_{\z}$ and then argue that the points outside of this region follow with a small probability because of the sub-Gaussian tail of the mixture model outside an appropriate radius. 

To simplify the notation, we define the clipped loss and clipped loss restricted to a particular piece as
\begin{align*}
    \lossdenoisingclipped_t ( \lscore, \x_t, \z_t) &= \Big\| \lscore(\x_t, \classify(\cdot)) + \frac{\z_t}{ \sqrt{1 - \exp(-2t)} } \Big\|^2 \Ind \{ \| \x_t \| \leq R_{\x}, \| \z_t \| \leq R_{\z} \} \\
    \lossdenoisingclipped_t ( \lscore, \x_t, \z_t, U_i, \hballt ) &= \Big\| \lscore(\x_t, \classify(\cdot)) + \frac{\z_t}{ \sqrt{1 - \exp(-2t)} } \Big\|^2 \Ind \{ \classify(\x_t) = i, \hballt(\x_t, U_i), \| \x_t \| \leq R_{\x}, \| \z_t \| \leq R_{\z} \}.
\end{align*}
Similarly define $\lossdenoisingclipped_t ( \lscore, \x_t, \z_t, U_i, \hballt^c )$ by replacing $\hballt$ with $\hballt^c$. Recall that for the region where $\Ind \{ \classify(\x_t) = i, \hballt(\x_t, U_i)\} = 1$, $\lscore(\x_t, \classify(\cdot))$ is simplified to $p_j(\x_t)$.

\begin{lemma}[Sample complexity] Assume that the sum of absolute values of the coefficient of the polynomial is $M$. Then, choosing $R_{\x}, R_{\z} = \Theta( (\beta R d / \alpha) \log (1/\delta') )$ for some $\delta' > 0$ and taking number of samples $n \geq \poly( \frac{d M R \beta }{ \alpha \eps t_{\min} } \log \frac{1}{\delta} ) \poly( \frac{ d \beta R }{ \alpha  } \log \frac{1}{\delta'} )^\ell $, with probability at least $1 - \delta$ over samples, we have
\begin{equation*}
    \E_{\x_t, \z_t} [ \lossdenoisingclipped_t ( \lscore, \x_t, \z_t) ] \leq \frac{1}{n} \sum_{i=1}^n \lossdenoisingclipped_t ( \lscore, \x_t^{(i)}, \z_t^{(i)} ) + \eps.
\end{equation*}
\end{lemma}
\begin{proof}
    Denote $\coeffvec$ as coefficients of the polynomials and $\phi(\x)$ denote the monomials up to degree $\ell$. Then, we know that $\| \coeffvec \|_2 \leq \| \coeffvec \|_1 \leq M$. Additionally, the bound on $\| \x \|$ implies that $\| \phi(\x) \|_\infty \lesssim R_{\x}^{\ell}$. This implies that $\| \phi(\x) \|_2 \lesssim (d R_{\x})^\ell$. The Lipschitz constant $\lossdenoisingclipped_t$ for each coordinate can be upper bounded by
    \begin{equation*}
        \| \nabla \lossdenoisingclipped_t ( \lscore, \x_t, \z_t, U_i, \hballt )\| 
        \lesssim \frac{ dMR_{\z}(d R_{\x})^{\ell}  }{ \sqrt{1 - \exp(-2t)} } \leq \frac{ dMR_{\z}(d R_{\x})^{\ell}  }{ \sqrt{t_{\min}} }.
    \end{equation*}
    Additionally, we have $\| \lossdenoisingclipped_t \| \leq \frac{ (dMR_{\z})^2 (d R_{\x})^{2\ell}  }{ t_{\min} }$ for any $\| \x_t \| \leq R_{\x}$ and $\| \z_t \| \leq R_{\z}$. We choose $R_{\x}, R_{\z} \asymp \frac{\beta R d}{ \alpha } \log(1/\delta')$ for some $\delta' > 0$ and apply standard generalization error analysis result using Rademacher complexity for linear function class (e.g., see \cite{shalev2014understanding}). If we choose the total number of samples $n$ to satisfy $n \geq \frac{(d M R_{\z})^4 (d R_{\x})^{4\ell} }{ t_{\min}^2 \eps^2} \log \frac{1}{\delta}$,then with at least $1 - \delta$ probability, we have
    \begin{equation*}
         \E_{\x_t, \z_t} [ \lossdenoisingclipped_t ( \lscore, \x_t, \z_t, S_j, \hballt ) ]  \leq \frac{1}{n} \sum_{i=1}^n \lossdenoisingclipped_t ( \lscore, \x_t^{(i)}, \z_t^{(i)}, S_j, \hballt ) + \eps
    \end{equation*}
    for all $j \in [\numclust]$. Using a similar argument to prove the boundedness of $\lossdenoisingclipped_t ( \lscore, \x_t, \z_t, S_j, \hballt^c )$, we also obtain
    \begin{equation*}
        \E_{\x_t, \z_t} [ \lossdenoisingclipped_t ( \lscore, \x_t, \z_t, S_j, \hballt^c ) ] \leq \frac{1}{n} \sum_{i=1}^n \lossdenoisingclipped_t ( \lscore, \x_t, \z_t, S_j, \hballt^c ) + \eps.
    \end{equation*}
    Because $\Ind \{\classify(\x_t) = j\}$ for any single $j$ for all $\x_t$, combining these bounds for all $j \in [\numclust]$ for $n \geq $, we have
    \begin{equation*}
         \;\; \E_{\x_t, \z_t} [ \lossdenoisingclipped_t ( \lscore, \x_t, \z_t) ] \leq \frac{1}{n} \sum_{i=1}^n \lossdenoisingclipped_t ( \lscore, \x_t^{(i)}, \z_t^{(i)} ) + \eps \ .
    \end{equation*}
\end{proof}

\begin{proposition}
    Let $\mathcal M$ be an $(\alpha, \beta, R)$-well-conditioned mixture. Then, for any $\eps > 0$ and any noise scale $t \geq t_{\min} \geq \alpha \eps / R$, there exist an algorithm that takes number of samples $n \geq (\log \frac{1}{\delta}) d^{\poly( \frac{\beta k \log R}{ \alpha \eps \lambdamin })}$ and runs in sample-polynomial time and returns a score function $\lscore_t$ such that
    \begin{equation*}
        \E_{\x_t} [ \| \nabla_{\x} \log q_t(\x_t) - \lscore_t(\x_t) \|^2 ] \leq \eps\,.
    \end{equation*}
\end{proposition}
\begin{proof}
    We define the loss function outside the radius $\| \x_t \| \geq \frac{\beta R d}{\alpha} \log \frac{1}{\delta'}$ or $\| \z_t \| \geq \frac{\beta R d}{\alpha} \log \frac{1}{\delta'}$ as
\begin{equation*}
    \lossdenoisingout_t (\lscore, \x_t, \z_t) = \Big\| \lscore(\x_t, \classify(\cdot)) + \frac{\z_t}{ \sqrt{1 - \exp(-2t)} } \Big\|^2 \Ind \{ \| \x_t \| \geq \frac{\beta R d}{\alpha} \log \frac{1}{\delta'} \vee \| \z_t \| \geq \frac{\beta R d}{\alpha} \log \frac{1}{\delta'} \}
\end{equation*}
    The $\lossdenoisingout$ can be simplified as 
    \begin{align}
    \label{eq:dsm-out-loss-split}
        \Bigl| \E_{\x_t, \z_t} [ \lossdenoisingout (\lscore, \x_t, \z_t) ] \Bigr| &\lesssim \E_{\x_t} \Bigl[ \| \lscore(\x_t, \classify(\cdot)) \|^2 \cdot\Ind \Bigl\{ \| \x_t \| \geq \frac{\beta R d}{\alpha} \log \frac{1}{\delta'} \Bigr\} \Bigr] \\
        &\qquad\qquad + \frac{ 1 }{t_{\min}} \E_{\z_t} \Bigl[  \| \z_t \|^2 \cdot\Ind \Bigl\{ \| \z_t \| \geq \frac{\beta R d}{\alpha} \log \frac{1}{\delta'} \Bigr\} \Bigr]. 
    \end{align}
    The second term in the above equation can be upper bounded by $( \Pr \{ \| \z_t \| \geq R_{\z}  \} )^{1/2} ( \E \big[  \| \z_t \|^4 \big] )^{1/2} \lesssim \sqrt{\delta} d$. To upper-bound the first term, we first upper-bound $\E_{\x_t} [ \| p(\x_t, \mathcal M(S_j)) \|^4 ]$:
    \begin{align*}
        \E \big[ \| p(\x_t, \mathcal M(S_j)) \|^4 \big] & \leq M^4 \E_{\x_t} \big[ \| \phi(\x_t) \|_1^4 \big] \leq M^4 d^\ell \; \Bigl( \max_{\vec v: \| \vec v \|_1 \leq 4\ell} \E \big[ \prod_{i=1}^d | \x_t^{(i)} |^{\vec v_i} \big] \Bigr) \\ 
        &\leq M^4 d^\ell \max_{\vec v: \| \vec v \|_1 \leq 4\ell} \prod_{i=1}^d \big( \E \big[ | \x_t^{(i)} |^{\vec v_i d}  \big] \big)^{1/d} 
    \end{align*}
    Using Gaussian hypercontractivity (\Cref{fct:hypercontractivity}), we can simplify $\E \big[ | \x_t^{(i)} |^{\vec v_i d}  \big] \lesssim  \sum_{i=1}^k \lambda_i (\vec v_i d)^{\vec v_i d} (\beta R)^{\vec v_i d} \leq (4\ell d \beta R)^{4\ell d}$. Using this bound in \eqref{eq:dsm-out-loss-split}, we have
    \begin{equation*}
        \Big| \E_{\x_t, \z_t} [ \lossdenoisingout (\lscore, \x_t, \z_t) ] \Big| \lesssim \sqrt{\delta' / t_{\min}} d + M^4 d^\ell (4\ell d \beta R)^{4\ell} \sqrt{\delta'}. 
    \end{equation*}
    Choosing $\delta' = \poly( \eps t_{\rm min} / (d M (4\ell d \beta R)^{\ell} )$, we obtain the result.
\end{proof}

\newpage

\section*{Acknowledgments} 

We thank Adam Klivans for many illuminating discussions about score estimation, polynomial regression, and diffusion models throughout the preparation of this work. We also thank the authors of~\cite{gatmiry2024learning} for coordinating the submission of manuscripts with us.

\bibliographystyle{alpha}
\bibliography{references}

\newpage

\appendix

\end{document}